\documentclass[11pt]{article}

\usepackage{hyperref}
\usepackage{microtype,amssymb,latexsym,graphicx, xspace, url, xcolor, subfigure}
\usepackage{amsmath,amsthm}
\usepackage[margin=1.25in]{geometry}
\usepackage{mathpazo}
\usepackage{calrsfs}
\usepackage{cite}
\usepackage{comment}
\usepackage{enumerate}
\usepackage{multirow}

\long\def\oldtext#1{}

\def\reals{{\mathbb R}}
\def\sphere{{\mathbb S}}
\def\Lspace{{\mathbb L}}

\def\Sspace{{\mathbb S}}
\def\Xspace{{\mathbb X}}

\def\ph{\varphi}
\def\cell{\tau}
\def\ceil#1{\lceil #1 \rceil}

\def\eps{{\varepsilon}}
\def\bd{{\partial}}
\def\A{\mathcal{A}}
\def\B{\mathcal{B}}
\def\C{\mathcal{C}}
\def\E{\mathcal{E}}
\def\F{\mathcal{F}}

\def\L{\mathcal{L}}
\def\P{\mathcal{P}}
\def\R{\mathcal{R}}
\def\S{\mathcal{S}}
\def\T{\mathcal{T}}
\def\U{\mathcal{U}}

\def\EE{\mathsf{E}}

\def\QQ{\mathsf{Q}}

\def\XX{\mathsf{X}}
\def\NN{\mathsf{NN}}
\def\DDelta{\boldsymbol{\Delta}}
\def\pp{\mathbf{p}}
\def\xx{\mathbf{x}}
\def\yy{\mathbf{y}}
\def\tO{t_o}
\def\tQ{t_q}

\def\cyl{C}
\def\CC{{\mathbf C}}

\def\test{{\mathbf Q}}
\def\Z{\mathcal{Z}}
\def\gcyl{\gamma}
\def\nn{\mathbf{n}}

\def\dist{\mathop{\mathsf{d}}}

\def\VeD#1{{#1}^{\mbox{\raisebox{4pt}{\tiny$\|$}}}}

\def\x24{{4}}

\def\dist{{\rm dist}}
\def\etal{\textsl{et~al.}}
\def\M{\mathsf{M}}

\def\cigar{\Lambda}

\newtheorem{theorem}{Theorem}[section]
\newtheorem{lemma}[theorem]{Lemma}
\newtheorem{corollary}[theorem]{Corollary}

\theoremstyle{remark}

\def\subsecref#1{Section~\ref{subsec:#1}}
\def\secref#1{Section~\ref{sec:#1}}
\def\thmref#1{Theorem~\ref{thm:#1}}
\def\xthmref#1{Thm~\ref{thm:#1}}
\def\lemref#1{Lemma~\ref{lem:#1}}
\def\figref#1{Figure~\ref{fig:#1}}
\def\eqref#1{(\ref{eq:#1})}

\makeatletter

\long\def\@makecaption#1#2{
   \vskip 10pt
   \setbox\@tempboxa\hbox{{\footnotesize {\bf #1.} #2}}
   \ifdim \wd\@tempboxa >\hsize         
       {\footnotesize {\bf #1.} #2\par}
     \else                              
       \hbox to\hsize{\hfil\box\@tempboxa\hfil}
   \fi}
\makeatother

\begin{document}

\title{Point and Line Nearest-Neighbor Searching in 3-Space\thanks{%
    A preliminary partial presentation of the results in this paper appears in \cite{VD1, VD2}.
    Work by Pankaj Agarwal has been partially supported by NSF grants CCF-20-07556 and CCF-22-23870 and by a US-Israel Binational Science Foundation Grant 2022131.
    Work by Esther Ezra has been partially supported by Israel Science Foundation Grant 800/22,
    and by a US-Israel Binational Science Foundation Grant 2022131.
    Work by Micha Sharir has been partially supported by Israel Science Foundation Grant 495/23.}}
\author{Pankaj K. Agarwal\thanks{
    Department of Computer Science, Duke University, Durham, NC 27708, USA;
    {\sf pankaj@cs.duke.edu,
      https://orcid.org/0000-0002-9439-181X}}
  \and
  Esther Ezra\thanks{
    School of Computer Science, Bar Ilan University, Ramat Gan, Israel;
    {\sf ezraest@cs.biu.ac.il,
      https://orcid.org/0000-0001-8133-1335}}
  \and
  Micha Sharir\thanks{
    School of Computer Science, Tel Aviv University, Tel Aviv, Israel;
    {\sf michas@tauex.tau.ac.il,
      http://orcid.org/0000-0002-2541-3763}}
}

\maketitle

\begin{abstract}
This paper presents data structures for nearest-neighbor (NN) searching problems involving 
points, lines, segments, and triangles in $\reals^3$, achieving significantly better performance
than the previously best-known results for these problems. For example, we present a linear-size 
data structure for answering NN queries with lines or segments amid $n$ points in $\reals^3$
with $O^*(n^{1/2})$ query time (where the $O^*(\cdot)$ notation hides subpolynomial factors). 
We also present a data structure of $O^*(n^4)$ size that answers such queries in 
$O^*(1)$ time. For the converse problem, in which we seek the nearest neighbor of a query point 
amid $n$ lines, segments, or triangles in $\reals^3$, we present a linear-size data structure 
with $O^*(n^{2/3})$ query time. These results constitute a 
significant improvement over previous solutions. We obtain improved solutions for the two
extreme regimes of (near-)linear storage and of fast query time. These results also yield
trade-off bounds between the query time and the size of the data structure.
Our results rely on several combinatorial and algorithmic results on arrangements of surfaces in 
$\reals^3$ and $\reals^4$, particularly on recent results on vertical decompositions of substructures
in such arrangements established by the authors~\cite{VD1}.
\end{abstract}

\section{Introduction}
\label{sec:intro}

The general \emph{nearest-neighbor searching} problem can be formulated as follows: Given a set $\S$ 
of input objects, often referred to as {\em sites}, in $\reals^d$, for some $d\ge 2$, and a 
distance function $\rho$, the goal is to preprocess $\S$ into a data structure that can efficiently 
find the object of $\S$ that lies closest to a query object $q$ (point, line, segment, or triangle)
under the distance function $\rho$.
In its most classical form, referred to as the \emph{post-office problem}~\cite{Kn73},
both the input sites and the queries are points in ${\reals}^2$ and $\rho$ is the Euclidean metric. 
Because of its wide-ranging applications across many disciplines, nearest-neighbor searching 
has been extensively studied, for about half a century, not only in computational geometry 
but also in many fields, such as database systems, information retrieval, GIS, machine learning, 
computer vision, and natural language processing. The problem has been extended in numerous ways,
e.g., to instances
where the sites and/or the queries may come from families of more general geometric objects, 
and $\rho$ may be an arbitrary metric. Since NN searching is hard (i.e., inefficient) in high 
dimensions, much of the work has focused on computing approximate nearest neighbors for $d$ large.
We refer to a few comprehensive surveys on nearest-neighbor searching, for details on 
proximity problems of this kind and their solutions~\cite{AI-17,AM-04,MM-17,PMN15,SDI}.

In this paper, we focus on answering \emph{exact} NN queries where both the input sites and the query 
objects are points, lines, segments, or triangles in $\reals^3$. For a closed subset $X\subset\reals^3$ 
and a point $p \in \reals^3$, we define $\dist(p,X) = \min_{x\in X} \|p-x\|$, where $\|\cdot\|$
denotes the Euclidean norm, and for two closed subsets $X,Y \subset\reals^3$, we define 
\[
\dist(X,Y) = \min_{x\in X} \dist(x,Y) = \min_{x\in X,\;y\in Y} \dist(x,y) .
\]
For a set $\S$ of geometric objects and for a query object $q$, we define 
\[
\NN(q,\S) = \arg\min_{\sigma\in \S} \dist(q,\sigma),
\]
that is, the nearest neighbor of $q$ in $\S$. Our goal is to preprocess $\S$ into a data structure 
so that $\NN(q,\S)$, for a query object $q$, can be computed efficiently. This general goal
bifurcates into many specific variants, depending on the types of the query and the input objects,
as well as for the storage allocated for the structure. See below.

\subsection{Related work}

As mentioned above, NN searching, especially when the input sites and query objects are points, 
has been studied extensively. It is
beyond the scope of this paper to review all the related work on nearest-neighbor searching (see the aforementioned
surveys for details), so we focus on results relevant to this paper, involving NN-searching amid
points, and also amid lines, segments, or triangles, in three (and higher) dimensions.

Suppose $\S$ is a set of $n$ points in $\reals^d$. The set $\S$ can be preprocessed, in $O(n\log n)$ time, into 
a linear-size data structure so that an NN query for a query point under the Euclidean metric 
can be answered in $O^*(n^{1-1/\ceil{d/2}})$ time~\cite{Ag:rs} (where, as in the abstract, the $O^*(\cdot)$
notation hides subpolynomial factors). To answer a query in $O(\log n)$ time, the best known 
data structure uses $O^*(n^{\ceil{d/2}})$ space~\cite{Ag:rs}.
By combining these two data structures, for a given parameter $s \in [n, n^{\ceil{d/2}}]$, 
$\S$ can be preprocessed, in $O^*(s)$ time, into a data structure of size $O^*(s)$, 
so that an NN query for a point under the Euclidean metric can be answered in 
$O^*(n/s^{\ceil{d/2}})$ time. Known lower bounds suggest that these bounds are likely
close to optimal in the worst case~\cite{Ag:rs}. A similar bound (within a polylogarithmic factor)
can be achieved for NN searching under the $L_1$-metric using halfspace range-searching data structures, 
but no data structure with comparable performance bounds is known for answering NN queries under general 
$L_p$-metrics. 

We note that many data structures reduce an NN query to answering $O^*(1)$ decision 
queries where, given a query object $q$ and a parameter $\rho$, the query procedure 
determines whether any of the input sites lies within distance $\rho$ from $q$. The decision query is formulated 
as an instance of semi-algebraic range-emptiness query~\cite{AM93,AM94}.
	Roughly speaking, there are two general techniques for answering semi-algebraic range queries. The first one, based on
geometric cuttings, leads to a linear-size data structure with $O^*(n^{\alpha})$ query time where $\alpha=\max\{1-\tfrac{1}{d},1-\tfrac{1}{2d-4}\}$, while the second one, based on polynomial partitioning~\cite{Ag:rs}, yields a linear size data 
structure with $O^*(n^{1-1/d})$ query time. These techniques also allow space/query-time trade-offs~\cite{AAEKS}. 
Although the cutting-based technique is suboptimal for $d>4$, it can sometimes be adapted to yield better query time for range-emptiness queries (see, e.g.,~\cite{Ag:rs}), which is 
precisely what is needed for answering NN decision 
queries. Consequently, both methods have been used to answer NN decision queries under general metrics.

The computational hardness of exact NN queries in high dimensions has inspired the extensive work 
on approximate NN searching, where the goal is to return a point whose distance from the query point 
is at most $(1+\eps)$ times the distance to $\NN(q,\S)$, for a given parameter $\eps>0$. 
For a fixed dimension $d$, a set $\S$ of $n$ points in $\reals^d$ can be preprocessed, in 
$O(n\log n)$ time, into a linear-size data structure that can return an $\eps$-approximate 
NN of a query point in $O(\log n + 1/\eps^{(d-1)/2})$ time; see~\cite{AFM-18, AMNSW-98, Bern-93,HP-01} 
and references therein. These data structures, however, suffer from the ``curse of dimensionality'',
because of the exponential dependence on $d$. This can be remedied by using locality sensitive 
hashing (LSH) where, for a parameter $c>1$, $\S$ can be preprocessed into a data structure of 
size $O(n^{1+1/c^2+o(1)})$ so that a $c$-approximate NN of a query point under the Euclidean metric 
can be computed in $O(dn^{1/c^2+o(1)})$ expected time~\cite{AI-06}. The LSH-based approach has been 
used for answering NN queries under various other metrics; see~\cite{AI-17}. Recently, several 
graph-based approaches have been proposed to answer approximate NN queries, which build a proximity 
graph on the input sites and follow a greedy search method in the graph to compute an approximate 
NN~\cite{MY20,PKGZ24}. However, very little has been proved on the performance guarantees of these data 
structures~\cite{IX23,HRR26}.

There is also some work on NN searching when the input sites or query objects are not points. 
Roughly speaking, most of the linear-size data structures map the input sites to points in 
some high-dimensional parametric space, while $O(\log n)$ query-time data structures map the
query objects to points in a similar suitable parametric space. 
For instance, using this approach, the nearest neighbor of a point amid $n$ lines in 
$\reals^3$ can be computed in $O(\log n)$ time using a data structure of $O^*(n^3)$ size~\cite{MS}, 
or in $O^*(n^{3/4})$ time using a linear-size data structure~\cite{AM94}.
Agarwal~\etal~\cite{ARS} presented data structures for answering approximate nearest-neighbor queries
with points amid simplices in $\reals^d$: For a storage parameter $s$ and an error parameter $\eps>0$, 
they constructed a data structure of size $O^*(s/\eps^d)$, so that an $\eps$-approximate NN of a 
query point amid $n$ $k$-dimensional simplices can be computed in $O^*(n/s^{\frac{1}{k+1}})$ time. 
See also~\cite{Mah15,AM21,AIK09}.
However, such an approach often results in suboptimal bounds (more precisely, in
bounds that are strongly suspected not to be optimal).

\subsection{Our contributions}

We present efficient data structures and algorithms for various proximity problems involving points 
and lines (or segments, or triangles) in $\reals^3$ that yield considerably improved performance bounds than the previously best known ones. Some of our data structures answer NN queries 
directly, while others answer decision queries, as described above, and use the 
parametric search technique to turn this decision procedure into a mechanism for answering NN queries. 
Some of our data structures can also answer \emph{fixed-radius-neighbor} queries, namely, for a 
query object $q$ and a parameter $\rho>0$, report all input sites (or just a single site) that lie 
within distance $\rho$ from $q$.

Our techniques formulate NN decision queries either as semi-algebraic range-emptiness queries
over a set of points or as point-location queries in a family of semi-algebraic sets. 
We primarily employ a geometric-cutting based approach to construct data structures for answering NN 
decision queries. Although our overall framework is
similar to that of~\cite{Ma:rph,ShSh}, several technical challenges 
arise in achieving the desired performance bounds. We address these challenges by exploiting the underlying 
geometric structure of the problem and by leveraging both combinatorial and algorithmic
results on arrangements of surfaces in $\reals^3$ and $\reals^4$, together with several new technical ideas, as summarized below.
Of particular relevance are the recent results by the authors~\cite{VD1} 
on vertical decompositions in 3- and 4-dimensional arrangements---see Section~\ref{sec:prelim}.

Table~\ref{table:results} summarizes our results and contrasts them with the previously best-known ones. 
We next expand this summary and describe the main results and techniques.

\begin{center}
 \begin{table}
	 \small
		 \begin{tabular}{|c|c|c|c|c|c|c|c|}
		\hline
			 \multirow{2}{*}{Query} & \multirow{2}{*}{Sites} & \multicolumn{3}{c|}{New results}&\multicolumn{3}{c|}{Old results}\\
		\cline{3-8}
		&&$S(n)$ & $Q(n)$ & Ref.&$S(n)$&$Q(n)$&Ref.\\
                \hline
	Line & \multirow{2}{*}{Points} & \multirow{2}{*}{$O(n)$} & 
	\multirow{2}{*}{$O^*(n^{1/2})$} & \xthmref{NNLP}&
	\multirow{2}{*}{$O(n)$}&\multirow{2}{*}{$O^*(n^{2/3})$}&\multirow{2}{*}{\cite{MP}}\\
	Segment &&&&\xthmref{NNSP}&&&\\
		\hline
			 Line&\multirow{3}{*}{Points}&\multirow{3}{*}{$O^*(n^4)$} &\multirow{3}{*}{$O^*(1)$} & \xthmref{line-pt-log-qtime}&\multirow{3}{*}{$O^*(n^5)$}&\multirow{3}{*}{$O^*(1)$}&\multirow{3}{*}{\cite{AAEZ}}\\
			 Segment &&&&\xthmref{tr-pt}&&&\\
			 Triangle &&&&\xthmref{tr-pt}&&&\\
		\hline
			 Line &\multirow{3}{*}{Points}& \multirow{3}{*}{$O^*(s)$} &\multirow{2}{*}{$O^*((\tfrac{n}{s^{1/4}})^{2/3})$}&\xthmref{n4}&\multirow{3}{*}{$O^*(s)$}&\multirow{3}{*}{$O^*((\tfrac{n}{s^{1/5}})^{5/6})$}&\multirow{3}{*}{\cite{AAEKS}}\\
			 Segment &&&&\xthmref{tr-pt-tradeoff}&&&\\\cline{1-1}\cline{4-4}
			 Triangle &&&$O^*((\tfrac{n}{s^{1/4}})^{8/9})$& \xthmref{tr-pt-tradeoff}&&&\\
		\hline
			 \multirow{3}{*}{Point} & Lines &\multirow{3}{*}{$O(n)$}&\multirow{3}{*}{$O^*(n^{2/3})$} &\xthmref{NNPL}&\multirow{3}{*}{$O(n)$}&\multirow{3}{*}{$O^*(n^{3/4})$}&\multirow{3}{*}{\cite{ShSh}}\\
	&Segments &&&\xthmref{NNPT}&&&\\
	&Triangles &&&\xthmref{NNPT}&&&\\
		\hline
\multirow{2}{*}{Line}&\multirow{2}{*}{Lines} & $O^*(n^4)$ & 
			 $O^*(1)$ &\multirow{2}{*}{\xthmref{n4lines}}&$O^*(n^5)$&$O^*(1)$&\multirow{2}{*}{\cite{AAEZ}}\\
			 &&$O^*(s)$& $O^*(\tfrac{n}{s^{1/4}})$&&$O^*(s)$&$O^*((\tfrac{n}{s^{1/5}})^{15/16})$&\\
		\hline
			 Segment & Segments &\multirow{2}{*}{$O^*(s)$}&$O^*((\tfrac{n}{s^{1/5}})^{15/16}+$&\xthmref{seg2segs}&\multirow{2}{*}{$O^*(s)$}&\multirow{2}{*}{$O^*(n/s^{1/6})$}&\multirow{2}{*}{\cite{AAEKS}}\\
			 Triangle & Triangles && {\hspace*{5mm}$(\tfrac{n}{s^{1/4}})^{16/15})$} & \xthmref{tri2tris}&&&\\
		\hline
	\end{tabular}
	 \caption{Summary of new results on NN searching in $\reals^3$.}
	 \label{table:results}
 \end{table}
\end{center}

 \paragraph{Nearest line-neighbor amid points} 
 (\secref{line-nn}).
Let $P$ be a set of $n$ points in $\reals^3$. We first show that $P$ can be preprocessed, in $O(n\log n)$ 
expected time, into a linear-size data structure so that, for a query line $\ell$ and a distance $\rho>0$, 
the structure can determine in $O^*(n^{1/2})$ time whether any point of $P$ lies within distance $\rho$ of $\ell$, 
or, equivalently, whether the cylinder of radius $\rho$ with $\ell$ as its axis contains any point of $P$; 
if the answer is {\sc Yes}, the data structure also returns such a point.
Our data structure is likely to be near-optimal in view of the known lower bounds on simplex range searching in the plane~\cite{Ag:rs}.
We note that Agarwal and Ezra~\cite{AE-23} presented such a data structure for the special case where the query cylinder 
has a fixed radius, say $1$. However, their approach does not extend to cylinders of arbitrary radius.

The main challenge in developing this data structure is the construction of a 
\emph{test set}, namely, a small set of representative semi-algebraic queries (not necessarily 
cylinders) such that, if the data structure can answer those queries efficiently, then it can answer 
efficiently any cylindrical range emptiness query in $\reals^3$. We propose a test-set construction algorithm 
(in \subsecref{test-construct}) by following the approach in~\cite{ShSh} and the recent results 
on vertical decomposition of arrangements of surfaces in $\reals^3$ and $\reals^4$~\cite{VD1}; 
our data structure can also be used to report all $k$ points of $P$ lying inside a query cylinder,
in additional $O(k)$ time. 
We finally extend the data structure (in \subsecref{seg-linear}) 
to the case where the query objects are line segments (but, unfortunately, not for the case of triangles),
with the same asymptotic performance bounds.

Addressing the complementary regime of large storage and $O(\log n)$ query time, we next show 
(in \subsecref{Fast_NN1}) that $P$ can be preprocessed, in expected $O^*(n^4)$ time, into a 
data structure of $O^*(n^4)$ size, so that the nearest neighbor in $P$ of a query line can be 
reported in $O(\log n)$ time. We show that the data structure can be extended to compute the NN 
in $P$ when the query object is a segment or a triangle in $\reals^3$.
Roughly speaking, after applying some geometric transformations, we reduce the nearest-neighbor query to 
a point-location query in a sandwich region enclosed between two envelopes of trivariate functions. 
We do not know how to perform this latter task efficiently in a direct manner due to the lack 
of good bounds on the complexity of the vertical decomposition of such a sandwich region (see~\cite{KS}, 
where this is stated as a major open problem). Therefore
we use a more involved scheme that exploits the 
special geometric structure of the problem, and achieves the desired efficiency. We next extend 
this data structure (in \subsecref{TrNN}) to answer NN queries, within the same asymptotic bound, 
with segments and triangles as query objects.

Combining these two data structures, for a storage parameter
$s\in[n,n^4]$, we can construct a data structure of size $O^*(s)$, in $O^*(s)$ expected preprocessing time,
that can answer an NN query with a line or a segment in $O^*((n/s^{1/4})^{2/3})$ time, 
and for a triangle in $O^*((n/s^{1/4})^{8/9})$ time. 

\paragraph{Nearest point-neighbor amid lines, segments, and triangles} (\secref{plnn}).
Next, we consider the converse problem where the input sites are lines, segments, or triangles, 
and the query objects are points. 
We show that a set $L$ of $n$ lines in $\reals^3$ can be preprocessed, in $O^*(n)$ expected time,
into a linear-size data structure so that $\NN(L,q)$, for a query point $q\in\reals^3$, can be computed 
in $O^*(n^{2/3})$ time. As in the setup considered earlier, using
parametric-search technique we reduce an NN query to
answering a semi-algebraic range-emptiness query in 
the $4$-dimensional parametric space of lines, and we 
build a linear-size partition tree for answering this query, as in~\cite{ShSh}.
Again, the main challenge is the construction a 
small-size test set. The algorithm is similar to that in \subsecref{pt-linear}, but now uses improved results 
on the vertical decomposition of the region lying below the lower envelope of a family of the graphs of 
constant-complexity, semi-algebraic trivariate functions in $\reals^4$~\cite{VD1}.

We also note that a data structure of $O^*(n^3)$ size and $O(\log n)$ query time can be obtained using 
the recent results in~\cite{VD1}. By combining this data structure with the linear-size data structure, 
for any parameter $s\in[n,n^3]$, we can construct, in $O^*(s)$ expected time, a data structure of 
size $O^*(s)$, so that $\NN(L,q)$, for a query point $q\in\reals^3$, can be computed in $O^*(n/s^{1/3})$ time. 
We extend the data structures to answer point NN queries, within the same asymptotic bound, also for the
cases where the input sites are line segments or triangles.

\paragraph{Nearest line-neighbor amid lines} 
 (\secref{Fast_NN2}).
Next, we consider the case in which both input sites and query objects are lines. 
By adapting the data structure in~\subsecref{Fast_NN1},
we can preprocess a set $L$ of $n$ lines in $\reals^3$, in $O^*(n^4)$ expected time,
into a data structure of size $O^*(n^4)$, so that the nearest neighbor in $L$ of a query line can 
be computed in $O(\log n)$ time. In the complementary regime, 
a linear-size data structure with $O^*(n^{3/4})$ query time can be obtained by
using standard four-dimensional semi-algebraic range searching techniques~\cite{AMS}. Combining these two data
structures, we can construct a data structure of size $O^*(s)$, for any storage parameter
$s\in[n,n^4]$, so that a query can be answered in $O^*(n/s^{1/4})$ time.

\paragraph{Nearest triangle-neighbor to a query triangle} (\secref{Triangle-NN}).
Finally, we consider the problem of answering an NN query with a triangle amid a set of $n$ triangles 
in $\reals^3$ (the most general setup in our work). We break the problem into multiple cases, and 
the hardest (i.e., most expensive) case is answering an NN query with a segment amid a set of segments. 
We show that the NN decision problem for a set $\E$ of $n$ segments in $\reals^3$, which, for a given query 
segment $e$ and a distance parameter $\rho>0$, asks whether
there exists a segment of $\E$ within distance $\rho$ of $e$, can 
be formulated as a $5$-dimensional semi-algebraic range searching problem, even though the parametric 
dimension of a segment in $\reals^3$ is six. This follows 
because the \emph{reduced parametric dimension}
(in the terminology of \cite[Appendix]{AAEKS}) of both input and query segments is only five. We further 
observe that each polynomial inequality involves either only four parameters of the input segments or 
only four parameters of the query segment. Using this additional observation, we manage to obtain a slightly better 
performance than what one would get for a $5$-dimensional semi-algebraic range-searching data 
structure~\cite{AAEKS}. In particular,
for a storage parameter $s\in [n,n^5]$, we present a data structure of size and expected preprocessing
cost $O^*(s)$, so that an NN query with a triangle can be answered in 
$O^*((n/s^{1/5})^{15/16}+(n/s^{1/4})^{16/15})$ time. (The breakpoint between the two terms is
$s = n^{31/19}$, and the query time is then $O^*(n^{12/19})$.)

\paragraph{Closest line-point pair} (\secref{ann}).
We conclude this paper with an off-line problem (not listed in Table~\ref{table:results}):
Let $P$ be a set of $m$ points and $L$ a set of $n$ lines in $\reals^3$. The goal is to compute 
the closest line-point pair in $L\times P$. We describe a randomized algorithm with expected time
$O^*(m^{3/5}n^{4/5}+m+n)$ (\thmref{bcp}). The on-line data structure presented in \secref{plnn} 
yields an algorithm with a worse performance bound of $O^*(n^{6/7}m^{4/7}+m+n)$ expected time. 
Our algorithm can also be extended to compute a closest 
pair between points and line segments or pairwise-disjoint 
triangles within the same time bound. 

\section{Preliminaries}
\label{sec:prelim}

In this section, we briefly review a few concepts and results related to geometric arrangements and 
semi-algebraic range searching that will be used repeatedly throughout this paper.

\subsection{Arrangements and vertical decompositions}
\label{subsec:arr}

Let $\S$ be a family of $n$ semi-algebraic sets\footnote{%
  Roughly speaking, a semi-algebraic set in $\reals^d$ is the set of points in $\reals^d$ that 
  satisfy a Boolean formula over a set of polynomial inequalities; the complexity of a semi-algebraic 
  set is the maximum of the number of polynomials defining the set and their maximum degree.}
of constant complexity in $\reals^d$. $\S$ may consist of
$d$-dimensional regions with non-empty 
interiors or of $(d-1)$-dimensional surface patches that
are graphs of partially-defined 
semi-algebraic functions. 
 %
The \emph{arrangement} of $\S$, denoted by $\A(\S)$, is the decomposition of $\reals^d$ into 
maximal connected relatively open cells of all dimensions, so that all points within a cell 
lie in the relative interior of the same subfamily of sets of $\S$, and on the boundary of another 
such subfamily, and are disjoint from all other sets. 
See~\cite{AS00} for a detailed survey of arrangements.

\paragraph{Union of semi-algebraic sets.}
Suppose $\S$ is a collection of $d$-dimensional semi-algebraic sets in $\reals^d$.
Let $\U(\S)$ denote their union and $\C(\S)$ the (closure of the) complement of their union.
For a parameter $m \le n$, let $\psi(m)$ denote the maximum combinatorial complexity of the 
union of any subset of $\S$ of size at most $m$. It is known that 
for $d=3$, $\psi(m)=O^*(m^2)$ in many cases (e.g., when $\S$ is a set of pseudo-spheres, a set of cubes
or a set of cylinders); see~\cite{APS} for a survey of known results on the complexity of the union 
of a family of geometric objects. 
In this paper we will mostly be interested in the union of semi-algebraic sets in $\reals^3$.

\paragraph{Lower and upper envelopes.}
If $\S$ is a set of $(d-1)$-dimensional semi-algebraic sets in $\reals^d$, we assume that they are monotone 
in the $x_d$-direction, meaning that every $x_d$-parallel (`vertical') line intersects such a set at most
once, and regard each set $S_i$ as the graph of a (possibly partially defined) semi-algebraic function, 
which we also denote by $S_i$. The \emph{lower envelope} $\EE_\S$ of $\S$ is defined as 
$$
\EE_\S(x) = \min_{S\in\S} S(x),
$$ 
for $x\in\reals^{d-1}$ (where $S(x)$ is taken to be $+\infty$ when $S$ is undefined at $x$). 
The \emph{minimization diagram} of $\F$, denoted by $\M_\S$, is the projection of 
the graph of $\EE_\S$ onto the hyperplane $x_d=0$, i.e., a subdivision of $\reals^{d-1}$ (namely, 
of that hyperplane) into maximal connected cells so that the same function appears on the lower 
envelope for all points in any given cell.  The upper envelope and the maximization diagram of 
$\S$ are defined analogously. It is known that the complexity of the lower and upper 
envelopes and of the minimization/maximization diagrams is $O^*(n^{d-1})$; see~\cite{SA} for details. 

Let $\XX$ be a family of semi-algebraic sets of constant complexity in $\reals^d$. For each 
object $X_i \in \XX$, we can define a $d$-variate \emph{distance function} $f_i(x) = \dist(x,X_i)$; 
set $\F = \{f_1, \ldots, f_n\}$. Then answering an NN query $\NN(q,\XX)$, for a point $q\in \reals^d$, 
can be formulated as finding the function of $\F$ that appears on the lower envelope at $q$. 
The minimization diagram of $\F$ is the \emph{Voronoi diagram} of $\XX$ (it coincides with 
the standard Voronoi diagram when $\XX$ is a set of points).

\paragraph{Vertical decomposition.}
Vertical decomposition is a popular (and perhaps the only general-purpose) technique for decomposing 
a cell of $\A(\S)$ into constant-complexity subcells, each of which is homeomorphic to a ball of 
the appropriate dimension. These cells, referred to as ``vertical pseudo-prisms'' (prisms for short), 
are obtained by recursion on the dimension. See \cite{VD1,CEGS-91,SA} for details. Let $\Pi$ be a 
substructure of $\A(\S)$, represented as a collection of cells of $\A(\S)$ that satisfy certain 
properties (e.g., lying in the complement of the union of $\S$ if $\S$ is a set of $d$-dimensional 
semi-algebraic sets). The vertical decomposition of $\Pi$, denoted by $\VeD{\Pi}$,
is the union of the vertical decompositions of the cells of $\Pi$. 
It is known that the complexity of $\VeD{\A}(\S)$, the vertical decomposition of the entire 
arrangement $\A(\S)$, is $O^*(n^d)$ for $d \le 4$~\cite{Koltun-04a}, which is near optimal, but 
for higher values of $d$, the best known bound is $O^*(n^{2d-4})$~\cite{CEGS-91,Koltun-04a}. 
We will be mostly interested in computing the vertical decompositions of substructures of $\A(\S)$, for 
$d=3,4$, such as $\C(\S)$ in $\reals^3$ and $\M_\S, \EE_\S$ in $\reals^4$. 
In particular, we will be using the following results from~\cite{VD1}.

\begin{lemma}[\cite{VD1}]
  \label{lem:union}
(i) Let $\S$ be a collection of $n$ constant-complexity semi-algebraic regions in $\reals^3$, and let 
$\psi(m)$ be  the maximum overall complexity of the cells in $\C(\R)$, over the subsets $\R\subseteq \S$ 
of size at most $m$. Then the complexity of $\VeD{\C}(\S)$ is $O^*(n^2+\psi(n))$, and it can be 
computed in $O^*(n^2+\psi(n))$ expected time.

\smallskip
\noindent
(ii) Let $\S$ be a collection of the graphs of $n$ constant-complexity, (possibly) partially defined 
semi-algebraic trivariate functions in $\reals^3$. The complexity of the vertical decomposition of 
$\M_\S$ or of the region lying below $\EE_\S$ is $O^*(n^3)$, and the vertical decomposition
can be computed in $O^*(n^3)$ expected time.
\end{lemma}

\paragraph{Geometric cuttings.}
As above, let $\S$ be a collection of $n$ semi-algebraic sets of constant complexity in $\reals^d$, and
let $\Pi$ be a substructure of $\A(\S)$.
For a parameter $r>1$, a \emph{$(1/r)$-cutting} of $\Pi$ (with respect to $\S$) is a set 
$\Xi$ of pseudo-prisms with pairwise-disjoint relative interiors that cover $\Pi$ ($\Pi$ may
be properly contained in the union of $\Xi$), such that the
relative interior of each pseudo-prism $\tau \in \Xi$ is crossed by (intersected by, but not 
contained in) at most $n/r$ sets of $\S$. In the standard terminology, the subset of sets of $\S$ 
that cross a prism $\tau$ is called the \emph{conflict list} of $\tau$, and it is denoted
by $\S_\tau$.
We will need the following results from~\cite{VD1}:

\begin{lemma}[~\cite{VD1}]
  \label{lem:cuttings}
Let $\S$ be a collection of $n$ constant-complexity $3$-dimensional semi-algebraic sets in 
$\reals^3$, such that the complexity of the union of any subset of $\S$ of size $m$ is at most 
$\psi(m)$, for some superlinear function $\psi$. For any parameter $r \in [1,n]$, a $(1/r)$-cutting 
of $\C(\S)$, of size $O^*(r^2 + \psi(r))$, along with the conflict lists of its cells, can be 
computed in expected time $O^*(nr + n\psi(r)/r)$.
\end{lemma}
\paragraph{Point-location and Point-enclosure queries.}
We also need algorithms that preprocess a collection $\S$ as above into a data structure so that 
point-location queries (does a query point lie in a specific substructure of $\A(\S)$) and 
point-enclosure queries (report all sets of $\S$ that contain a query point) can be 
answered quickly. In particular, we will need the 
following  results from~\cite{VD1}:

\begin{lemma}[~\cite{VD1}]
  \label{lem:enc}
  (i) Let $\S$ be a set of $n$ semi-algebraic regions of constant complexity in $\reals^3$.
  $\S$ can be preprocessed, in $O^*(n^2+\psi(n))$ randomized expected time (where $\psi(n)$
is as above), into a data structure of size $O^*(n^2+\psi(n))$, so that, for a query point 
$q\in\reals^3$, determining whether $q\in\U(\S)$ can be done in $O(\log n)$ time.
If the answer is {\sc Yes}, the query procedure also 
returns a set of $\S$ containing $q$. Furthermore, 
all $k$ sets of $\S$ containing $q$ can be 
reported in additional $O(k)$ expected time.

  \smallskip \noindent
  (ii) Let $\F$ be a set of $n$ semi-algebraic trivariate functions of constant complexity. 
  $\F$ can be preprocessed, in $O^*(n^3)$ randomized expected time, into a data structure of 
  size $O^*(n^3)$, so that for a point $q\in \reals^3$, $\EE_\F(q)$ can be computed in 
  $O(\log n)$ time. Furthermore, for a point $q\in\reals^4$, all $k$ functions of $\F$ whose 
  graphs lie below $q$ can be reported in additional $O(k)$ expected time.
\end{lemma}

In view of the above discussion on lower envelopes and NN searching, an immediate corollary 
of \lemref{enc}(ii) is the following:
\begin{corollary}
\label{cor:NN}
Let $\S$ be a family of $n$ semi-algebraic sets of constant complexity in $\reals^3$. Then
$\S$ can be preprocessed, in $O^*(n^3)$ expected time, so that $\NN(\S,q)$, for a query point 
	$q\in\reals^3$, can be computed in $O(\log n)$ time. Furthermore, the $k$ nearest neighbors of $q$ can be reported in an additional $O(k)$ time.
\end{corollary}

\subsection{Semi-algebraic range searching}
\label{subsec:RS}

Let $\Pi\colon \reals^d\times \reals^k \rightarrow \{0,1\}$ be a \emph{semi-algebraic predicate} defined as the 
Boolean combination of a set of polynomial inequalities. 
Without loss of generality, we can assume that $\Pi$ is of the form
\begin{equation}
\label{eq:semi-pred}
\Pi (\xx, \yy) = \bigvee_{i=1}^r \bigwedge_{j=1}^{k_i} (g_{ij} (\xx, \yy) \ge 0),
\quad\text{for } \xx\in \reals^d,\;\yy\in\reals^k ,
\end{equation}
where each $g_{ij} \in \reals[\xx,\yy]$ is a polynomial over the joint space $\reals^d\times\reals^k$. 
For a fixed $\yy\in\reals^k$, let $\gamma(\yy) = \{ \xx\in\reals^d \mid \Pi(\xx,\yy) = 1\}$ be the 
semi-algebraic set defined by $\yy$. Set $\Gamma = \{ \gamma(\yy) \mid \yy\in\reals^k \}$. 
We refer to $\Gamma$ as a family of semi-algebraic ranges (e.g., rectangles, balls, simplices 
in $\reals^d$). (In some cases, we may define $\Gamma$ with respect to a subspace of $\reals^k$.) 
Let $\tO$ (resp., $\tQ$), referred to as the \emph{reduced parametric dimension} of the input 
(resp., query) objects, be the maximum number of coordinates of $\xx$ (resp., $\yy$) that appear 
in any polynomial inequality $g_{ij}$ in \eqref{semi-pred}. 

Let $P$ be a set of $n$ points in $\reals^d$. We wish to preprocess $P$ into a data structure so 
that for a query range $\gamma \in \Gamma$, we can quickly retrieve $\gamma \cap P$, where we may 
wish to report all the points of $\gamma\cap P$, count the number of such points, or simply determine 
whether $\gamma \cap P = \emptyset$. This problem is referred to as  \emph{$\Pi$-range searching}, 
or more generally as semi-algebraic range searching; see~\cite{Ag:rs}. 
Currently, the best-known data structure, based on the so-called polynomial-partitioning 
technique~\cite{Guth}, builds a \emph{partition tree} on $P$, each node of which stores 
a family of \emph{canonical subsets} of $P$, so that, for a query range $\gamma\in \Gamma$, 
$\gamma\cap P$ is returned as the union of a small number 
of pairwise-disjoint canonical 
subsets. By storing relevant information with each canonical subset, the various variants
of range searching mentioned above can be answered efficiently. The following lemma summarizes the result:
\begin{lemma}[Agarwal et al.~\protect{\cite[Appendix]{AAEKS}}]
	\label{lem:Pi-tradeoff}
Let $P$ be a set of $n$ points in $\reals^d$,
and let $\Pi\colon \reals^d\times\reals^k \rightarrow \{0,1\}$ be a semi-algebraic predicate of constant 
complexity, where $k\ge 1$ is a constant.
Let $\Gamma$ be the family of semi-algebraic 
query ranges induced by $\Pi$, each defined by some $\yy\in\reals^k$, and 
let $\tO, \tQ$ be the respective reduced parametric dimensions of $P$ and $\Gamma$
with respect to $\Pi$. For a storage parameter $s\in[n,n^{\tQ}]$,
$P$ can be preprocessed, in $O^*(s)$ expected time, into a data structure of size $O^*(s)$,
so that a $\Pi$-query on $P$ with an object $\gamma\in\Gamma$ can be answered 
in $O^*\biggl((n/s^{1/\tQ})^{\tfrac{1-1/\tO}{1-1/\tQ}} \biggr)$ time.
\end{lemma}

Let $P$ and $\Gamma$ be as above. As already mentioned, using the parametric-search 
technique~\cite{AM93}, an NN query $\NN(P,\gamma)$, for a range $\gamma\in\Gamma$, can be 
reduced to answering $O^*(1)$ semi-algebraic range-emptiness queries with query objects of the form 
$\B(\gamma, \rho) = \{ \xx \in \reals^d \mid \dist(\xx,\gamma) \le \rho \}$ for some varying
parameter $\rho>0$. We therefore focus on designing (more) efficient data structures for 
answering such range-emptiness queries (see, e.g., Sections~\ref{sec:line-nn} and~\ref{sec:plnn}),
using the ideas in~\cite{Ma:rph,ShSh}. We briefly describe the approach taken in
Sharir and Shaul~\cite{ShSh}.

\paragraph{NN searching and test-sets.}
Roughly speaking, the work in~\cite{ShSh} builds a linear-size partition tree recursively, using 
hierarchical clustering of the input points. At each node $v$ of the tree, we have a subset 
$P_v \subseteq P$ containing $n_v$ points. For a parameter $r>1$, the technique constructs an 
\emph{elementary-cell partition} $\Pi=\{(P_1,\tau_1), \ldots, (P_r,\tau_r)\}$ where 
$n_v/r \le |P_i| \le 2n_v/r$ for each $i$, $P_1, \ldots, P_r$ is a partition of $P_v$, 
and each $\tau_i$ is an elementary cell containing $P_i$ 
(these cells are not required 
to be pairwise disjoint). The main property of $\Pi$ is that any range in $\Gamma$ 
crosses only few elementary cells of $\Pi$. A major ingredient to accomplish this, also used
in~\cite{Ma:rph}, is to construct a so-called \emph{test set} of a small number of semi-algebraic 
sets (not necessarily ranges of $\Gamma$), which represent all ranges of $\Gamma$ well.
Formally, for a parameter $r>1$, we call a semi-algebraic set $\gamma \subset \reals^d$
\emph{$(n/r)$-shallow} (with respect to $P$) if $|P \cap \gamma| \le n/r$, and a finite 
collection $\Phi$ of semi-algebraic sets of constant complexity is called a \emph{test set} 
for $(n/r)$-shallow ranges in $\Gamma$ if it satisfies the following properties: 
  \begin{description}
\item[(P1)]
  Every set in $\Phi$ is $(n/r)$-shallow with respect to $P$. 
\item[(P2)]
  The complement of the union of any $m$ sets of $\Phi$ can be decomposed 
  into at most $\zeta(m)$ ``elementary cells'' (semi-algebraic sets of constant complexity, as above) 
  for any $m\ge 1$, where $\zeta(m)$ is a suitable monotone increasing super-linear function of $m$. 
\item[(P3)]
  Any $(n/r)$-shallow set $\gamma\in \Gamma$ can be covered by the union 
  of at most $\delta$ ranges of $\Phi$, where $\delta$ is a constant (independent of $r$).
\end{description}
The following lemma summarizes the main result of Sharir and Shaul~\cite{ShSh}:
\begin{lemma}[Sharir and Shaul~\cite{ShSh}]
	\label{lem:test-set}
Let $P$ be a set of $n$ points in $\reals^d$, and let $\Gamma$ be a family of semi-algebraic ranges 
of constant complexity. Suppose that a test set of size $r^{O(1)}$ for $(n/r)$-shallow ranges of 
$\Gamma$ can be constructed in $O^*(nr^{O(1)})$ time. Then a linear-size partition tree on $P$ can 
be constructed, in $O^*(n)$ expected time, that supports $\Gamma$-emptiness queries in 
$O^*(n/\zeta^{-1}(n))$ time; all $k$ points lying in a query $\Gamma$-range can be reported 
in an additional $O(k)$ time. 
\end{lemma}
We will present in Sections~\ref{sec:line-nn} and~\ref{sec:plnn} new algorithms for constructing 
test sets for the relevant ranges that arise in the analysis, and we will use \lemref{union} 
to prove (P2) for the test sets that we construct.

\section{Line, Segment, and Triangle NN Queries amid Points in $\reals^3$}
\label{sec:Fast_NN}
\label{sec:line-nn}

In this section we consider the setting in which the input sites are a set $P$ of $n$ points in $\reals^3$,
and the query objects are lines in $\reals^3$. We will later extend this data structure to handle segments and 
triangles as query objects. We first describe a linear-size data structure 
(\subsecref{pt-linear}) and then a fast-query-time data structure (\subsecref{Fast_NN1}).
Combining the two structures will yield a trade-off between the storage and query time.

\subsection{ Line NN queries: Linear storage}
\label{subsec:pt-linear}

As discussed above, using parametric search~\cite{Ag:rs,AM93}, an NN
query on $P$ can be reduced to answering $O^*(1)$ \emph{cylinder-emptiness} 
queries on $P$. That is, one aims to preprocess $P$ into a data structure of near-linear size that 
can efficiently determine whether a query cylinder $\cyl$ (of arbitrary radius) contains a point 
of $P$, and if the answer is {\sc Yes} than return such a point. 
For simplicity, we assume that the axis of the query cylinder $\cyl$ is not parallel to the $yz$-plane. 
When the axis is parallel to the $yz$-plane, one fewer degree of freedom is needed to specify $\cyl$,
and the solution becomes simpler and more efficient; for the sake of brevity, we omit the straightforward 
details of adapting the forthcoming analysis to this case. 

\paragraph{Data structure.}
We construct a two-level partition tree $\Psi$ on $P$ by following the same overall approach as in~\cite{AE-23}, 
as follows.
For a point $p\in P$, let $p^*$ be its $xy$-projection, and let $P^* = \{ p^* \mid p\in P\}$.
Without loss of generality, we assume that no two points in $P$ project to the same point.
Our top-level tree is a two-dimensional partition tree on $P^*$ (with respect to lines in
the $xy$-plane), based on simplicial partitions~\cite{m-ept-92}, and each of its nodes stores 
a second-level partition tree of a different type; see below.

Concretely, let $S$ be a set of $n$ points in $\reals^2$,  and let $r > 0$ be a parameter.
A \emph{simplicial $(1/r)$-partition} (a special case of elementary-cell partition reviewed earlier) for $P$ 
is a collection $\Pi = \{(S_1, \Delta_1), \ldots, (S_m, \Delta_m)\}$, where $m \le r$ is an integer, such that
(i) $\{S_1, \ldots, S_m\}$ is a partition of $S$ (into pairwise-disjoint subsets)
satisfying $n/r \le |S_i| \le 2n/r$, for each $i$, and
(ii) each $\Delta_i$ is a (possibly degenerate) triangle, referred to as a \emph{cell}, that contains $S_i$.
The \emph{crossing number} of $\Pi$ with respect to a line $\ell$ in $\reals^2$ is the
number of cells crossed by $\ell$. The crossing number
of $\Pi$ is defined as the maximum of the crossing numbers over all lines $\ell$.
Matou\v{s}ek~\cite{m-ept-92} described an algorithm for constructing a simplicial partition whose crossing number is 
$O(\sqrt{r})$. If $r=O(1)$, the running time of his algorithm is $O(n)$; see also~\cite{Chan-12}.

We choose $r$ to be a sufficiently large constant. By constructing simplicial partitions recursively 
and stopping the recursion as soon as the number of points becomes smaller than some sufficiently 
large constant threshold $n_0$, we construct a two-dimensional partition tree on $P^*$, 
which is the primary tree $\Psi$ of our structure. Each node $v\in \Psi$ is associated with a triangle
$\Delta_v$ and a subset $P_v^* \subseteq P^*\cap \Delta_v$. If $v$ is the root, then $\Delta_v=\reals^2$ and $P_v^*=P^*$. Let 
$P_v = \{ p\in P \mid p^* \in P_v^*\}$ be the subset of $P$ corresponding to $P_v^*$. Set $n_v = |P_v|$. 
For each node $v\in \Psi$, we build a second-level partition tree (now in $\reals^3$) $\Sigma_v$ on $P_v$ 
for answering range queries with cylinders whose $xy$-projections contain $\Delta_v$ (we refer to such
cylinders as being \emph{wide} at $v$), using the algorithm 
described below. By \lemref{cylinder_range_search}, given below, $P_v$ can be preprocessed, 
in $O (n_v\log {n_v})$ time, into a linear-size data structure, so that an emptiness 
query for a cylinder whose $xy$-projection contains $\Delta_v$ can be answered in  $O^*({n_v}^{1/2})$ time.

\begin{figure}[htb]
  \centering
  \includegraphics[scale=0.5]{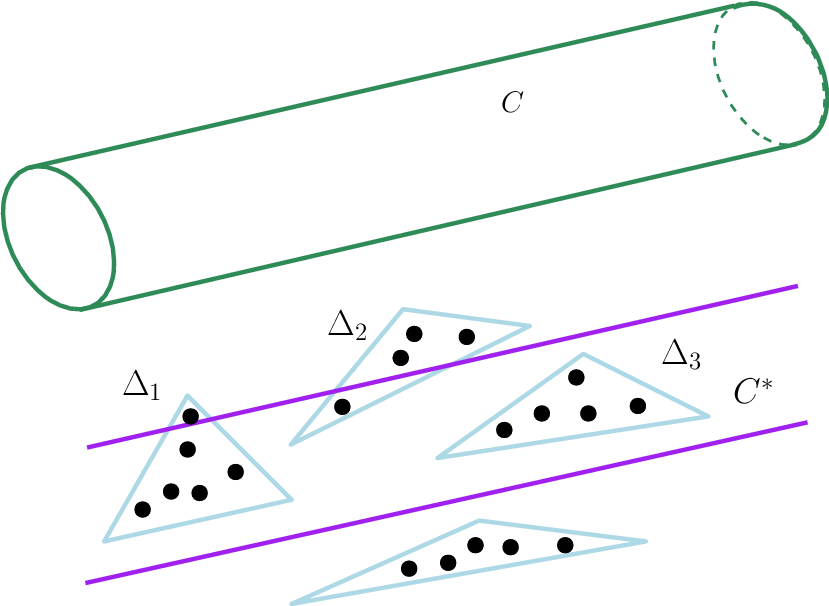}
  \caption{A query cylinder and its $xy$-projection. The query is answered recursively for $\triangle_1$ and $\triangle_2$ and using the secondary data structure for $\triangle_3$.}
  \label{fig:line-query}
\end{figure}

\paragraph{Query procedure.} 
Let $\cyl$ be a query cylinder whose axis is not parallel to the $yz$-plane, 
and let $\cyl^*$ be its $xy$-projection, which is a strip bounded by two parallel lines.
We visit $\Psi$ recursively in a top-down manner, starting from its root. Suppose we are at a node $v\in \Psi$. If 
$\cyl^* \cap \Delta_v = \emptyset$, then we simply return to the parent recursive step. 
If $v$ is a leaf, then we check all points of $P_v$ for containment in $\cyl$. If any one of 
them lies in $\cyl$, we return {\sc yes} and report such a point;
otherwise we return \textsc{no}. So assume that $v$ is an internal node and 
$\Delta_v \cap \cyl^* \ne \emptyset$. If $\Delta_v\subset\cyl $, then we 
use the secondary data structure $\Sigma_v$ stored at $v$ to test whether 
$P_v \cap \cyl \ne \emptyset$, as described below.
If $\bd\cyl^* \cap \Delta_v\ne\emptyset$ (i.e., one of the two lines bounding $\cyl^*$ intersects $\Delta_v$), 
then we recursively visit the children of $v$. 
See \figref{line-query}.
If any of the children returns a point of $P_v\cap \cyl$, we return that point (with a
{\sc yes} output), and otherwise we return \textsc{no}.
For an emptiness query, the procedure can terminate as soon as a point of $P$ inside $\cyl$ is found.

The query procedure can be adapted to report all points of $\cyl \cap P$: for a leaf node $v$, 
we report all points of $P_v\cap C$. For an internal node $v$, if $\Delta_v\subset\cyl^*$, the 
secondary structure returns all points of $P_v \cap \cyl$; and if $\bd\cyl^*\cap \Delta_v\ne\emptyset$, 
the recursive processing of $v$ will report all points of $P_v\cap \cyl$.

\paragraph*{Analysis.}
The height of $\Psi$ is $O(\log n)$, and each of its internal nodes $v$ stores a linear-size 
(linear in $n_v$) secondary structure, while each leaf uses $O(1)$ space. This 
implies that the overall size of $\Psi$ is $O(n \log n)$.  A similar argument shows that the 
preprocessing time of constructing $\Psi$ is $O(n\log^2 n)$.

Concerning the query time, we present the analysis for emptiness queries. 
Let $Q(n)$ denote the maximum runtime of an emptiness query in $\Psi$ for
a set of at most $n$ points. For $n \le n_0$, $Q(n)= O(n) = O(1)$. 
If $\cyl^* \subset \Delta_v$ then we use the secondary data structure 
stored at $v$, and answer an emptiness query in $O^*(\sqrt{n_v})$ time, as will be shown in
\lemref{cylinder_range_search} below. Since the crossing number of the simplicial partition is 
$O(\sqrt{r})$, the procedure recursively visits $O(\sqrt{r})$ children of $v$.
We therefore obtain the following recurrence.
\[
Q(n) \le  \left \{ \begin{array}{ll} O(1) & n \le n_0 \\[1mm]
c_1\sqrt{r}Q(2n/r) + O^*(n^{1/2}) & n > n_0,
\end{array}
\right .
\]
where $c_1>0$ is an absolute constant. Choosing $r$ to be a sufficiently large constant, 
the solution becomes $Q(n) = O^*(n^{1/2})$, with a slightly larger $O^*(\cdot)$ factor.
For the reporting query, by \lemref{cylinder_range_search} below, the overall query time 
is $O^*(n^{1/2})+O(k)$.

The $\log n$ factor can be removed from the space and preprocessing time bounds,
using a standard technique (see, e.g.,~\cite{AE99}) of storing the second-level structure 
only at levels of the primary structure that are $\eps\cdot \log{n}$ apart, for a sufficiently 
small constant $\eps > 0$.  This adds a factor of $O(n^\eps)$ to the query time, which is 
subsumed anyway by our $O^*(\cdot)$ notation. 

Returning to the original problem of finding a nearest neighbor for a line query amid points, 
we conclude:

\begin{theorem}
  \label{thm:NNLP}
  A set $P$ of $n$ points in $\reals^3$ can be preprocessed, in $O(n\log n)$ expected time, 
  into a data structure of size $O(n)$, so that (i) for any query line $\ell$ in $\reals^3$, 
  the point of $P$ nearest to $\ell$ can be computed in $O^*(n^{1/2})$ time; (ii) for a query 
  line $\ell$ and a distance parameter $\rho$, all $k$ points of $P$ within distance $\rho$ 
  from $\ell$ can be reported in $O^*(n^{1/2})+O(k)$ time.
\end{theorem}

\paragraph{Range queries for wide cylinders.}
\label{sec:wide}

To complete the analysis, we now consider the case of ``wide'' cylinders, with respect to a cell 
$\Delta$ in the $xy$-plane, namely cylinders $\cyl$ satisfying $\cyl^*\supseteq \Delta$.

Let $P$ be a set of $n$ points in $\reals^3$, and let $\Delta$ be a triangle in $\reals^2$ ($xy$-plane) 
that contains $P^*$, the $xy$-projection of $P$. We describe a data structure that, for a query cylinder 
$\cyl$ whose $xy$-projection $\cyl^*$ contains $\Delta$, detects whether $\cyl\cap P\ne\emptyset$ 
(or reports all points of $\cyl \cap P$). Since we assume here that $\Delta\subset\cyl^*$,  
the width of $\cyl^*$ is at least that of $\Delta$. By splitting $P$ into $O(1)$ subsets and 
querying each subset with $\cyl$ separately, we can assume, without loss of generality, 
that the width of $\Delta$ is at most $w_0$, for some sufficiently small but fixed constant $w_0 < 1$, 
whose value will be set below (see \lemref{good_directions}), and that the radius of the 
query cylinder is at least $1$. 
We refer to cylinders of radius at least $1$ as \emph{wide} cylinders.
We may thus assume, without loss of generality, that $P$ lies in a vertical 
slab of width $w_0 < 1$ bounded by two vertical planes $H$ and $H'$, and that
the normals of $H$ and $H'$ are $(\pm 1,0,0)$. 

By \lemref{test-set}, to answer a range query amid $P$ with a wide cylinder, 
it suffices to describe an algorithm for constructing a test set for wide cylinders. 
In the remainder of the subsection we describe an algorithm that, for a given parameter $r>1$,
constructs a test set of size $O^*(r^6)$ in time $nr^{O(1)}$, with $\zeta(m)=O^*(m^2)$ 
(cf.\ \lemref{generalized_cylinders}) and $\delta=1$. Plugging these bounds in 
\lemref{test-set}, we obtain the following:

\begin{lemma}
  \label{lem:cylinder_range_search}
  Let $P$ be a set of $n$ points in ${\reals}^3$, and let $\Delta$ be a triangle in $\reals^2$ 
  containing the $xy$-projection of $P$. Then $P$ can be preprocessed, in $O(n\log n)$ expected 
  time, into a data structure of size $O(n)$, so that, for a cylinder $\cyl$ whose $xy$-projection 
  contains $\Delta$, an emptiness query with $\cyl$ can be answered in $O^*(n^{1/2})$ time, and all 
  $k$ points of $\cyl\cap P$ can be reported in $O^*(n^{1/2}) + O(k)$ time.
\end{lemma}

\subsection{Test-set construction}
\label{subsec:test-construct}
As above, let $P\subset\reals^3$ be a  set of $n$ points lying in a vertical slab $\sigma$ of width 
$w_0=2\sin^2 (\pi/\kappa)$ for some sufficiently large constant $\kappa$ (whose value will be fixed later), bounded by two 
vertical parallel planes $H$ and $H'$ whose normals are $(\pm 1,0,0)$. We now describe an algorithm for constructing 
a test set for answer range queries with wide cylinders, i.e., cylinders of radius at least $1$. 

Let $\CC$ be the family of cylinders of radius at least $1$ whose axes are not parallel to the $yz$-plane. 
(Handling the degenerate case where the axis is parallel to that plane is easier, and omitted here.)
A cylinder $\cyl \in\CC$ can be represented by a point 
$\pp_{\cyl} = (y_{C,1},z_{C,1},y_{C,2},z_{C,2},\rho_C)\in\reals^5$, where $\rho_C\ge 1$ is the 
radius of $\cyl$, and $(0,y_{C,1},z_{C,1})$ and $(1, y_{C,2},z_{C,2})$ are
the intersection points of the axis of $\cyl$ with the respective planes $x=0$ and $x=1$.
We thus identify $\cyl$ with the point $\pp_\cyl$ in the halfspace\footnote{%
  This is the only place where we need the assumption of the axis of the query cylinder not 
  being parallel to the $yz$-plane, thus allowing us to use real, rather than projective, parametric 5-space.}
$\rho\ge 1$ in $\reals^5$.

For a cylinder $C_\yy$ represented by a point $\yy = (y_1,z_1,y_2,z_2,\rho)$ in $\reals^5$,
with $\rho\ge 1$, and for a point $\xx\in\reals^3$, we can define the predicate $\Pi(\xx,\yy)$ 
that is $1$ if and only if $\xx\in \cyl_\yy$ (or equivalently, $\xx$ is within distance $\rho$ 
from the axis of $\cyl_\yy$) by an $8$-variate polynomial inequality
\begin{equation} \label{eq:cyl-pred}
\Pi(\xx,\yy): 
\|\xx-\pp_1\|^2-\left\langle \xx-\pp_1, \frac{\pp_2-\pp_1}{\|\pp_2-\pp_1\|}\right\rangle^2 \le \rho^2,
\end{equation}
where $\pp_1=(0,y_1,z_1)$ and $\pp_2=(1,y_2,z_2)$.
For a point $q \in \reals^3$, we can define the region 
\begin{equation}
K_q := \{ \yy=(y_1,z_1,y_2,z_2,\rho) \in\reals^5 \mid \rho\ge 1, \Pi(q,\yy) =1\} ,
\end{equation}
to be the set of all points representing wide cylinders 
that contain $q$. $K_q$ is clearly a semi-algebraic set of constant complexity.

\paragraph{Algorithm.}
Let $r \ge 1$ be a parameter. We construct a test set $\test$ for $(n/r)$-shallow ranges in $\CC$, as follows.
We choose a random sample $N$ of $O(r\log{r})$ points of $P$, with an appropriate constant of proportionality. 
We then form the set of regions $N^{*} = \{ K_p \mid p \in N\}$ in $\reals^5$, and construct their 
arrangement $\A(N^{*})$. 
Let $k = \ceil{c\log{r}}$ be an integer parameter, where $c > 0$ is an appropriate constant of 
proportionality. Let $\A_{\le k}(N^*)$ be the set of all points of $\A(N^{*})$ at \emph{level} 
at most $k$, that is, these points represent all wide cylinders that contain at most $k$ points of $N$. 
We compute the vertical decomposition $\VeD{\A}_{\le k}(N^*)$ of $\A_{\le k}(N^*)$, whose size is 
$O^*(r^6)$~\cite{Koltun-04a}. (Here we simply construct the vertical decomposition of the
entire arrangement and extract from it the relevant prisms for $\A_{\le k}(N^*)$.)
By a standard random-sampling argument~\cite{HW87}, a sufficiently large choice of $c$ 
implies that, with probability at least $1/2$, each cell of $\VeD{\A}_{\le k}(N^*)$ 
is crossed by at most $n/r$ surfaces of $N^*$.  If this is not the case, we discard $N$ 
and choose another random subset, until we find, in expected $O(1)$ trials, one with the desired property.

Let $\tau$ be a cell in $\VeD{\A}_{\le k}(N^*)$. We associate with $\tau$ a 
\emph{generalized cylinder} $\gcyl_{\tau} = \bigcup \{\cyl \in \CC \mid \pp_\cyl \in \tau\}$,
the volume swept by a continuous motion of a wide cylinder, as its representative point varies over $\tau$.
See Figure~\ref{fig:generalized_cylinders} for an illustration.\footnote{This image was found on the web, but we were unable to trace its source.} 
We set $\test = \{ \gcyl_{\tau} \mid \tau \in \VeD{\A}_{\le k}(N^*) \}$. We state a few key properties 
of $\test$, which are similar to those stated in~\cite{ShSh} (proofs of some of them are missing in~\cite{ShSh}).

\begin{figure}[htb]
  \centering
  \includegraphics[width=4.0in]{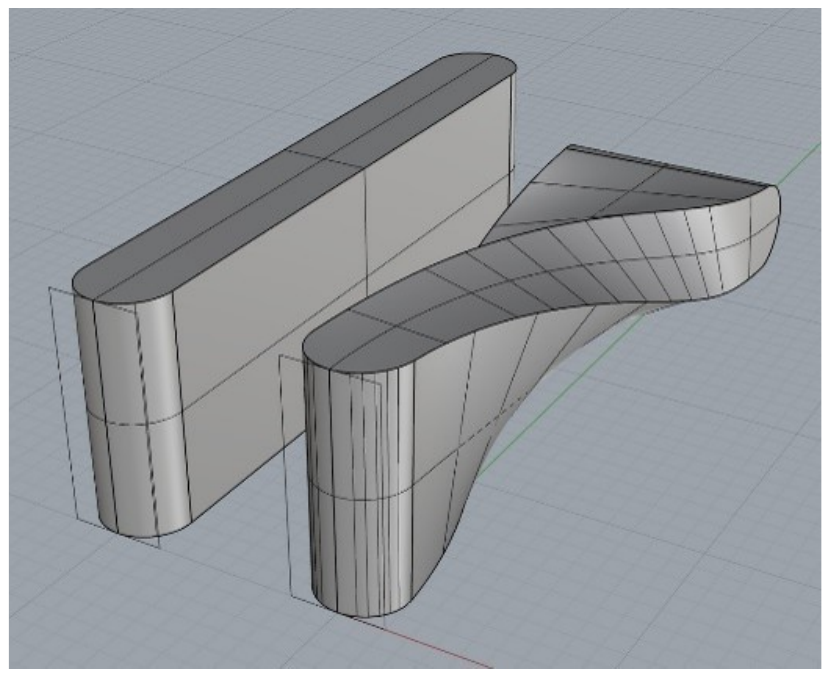}
  \caption{An illustration of generalized cylinders.}
  \label{fig:generalized_cylinders}
\end{figure}

\begin{lemma} \label{lem:test-cyl}
The set $\test$ has the following properties:
\begin{enumerate}[(i)]
\item Each generalized cylinder $\gcyl_{\tau} \in \test$ is a semi-algebraic set of constant complexity, 
and it is $(n/r)$-shallow with respect to $P$ with probability at least $1/2$ (assuming that the constant 
$r$ is chosen sufficiently large).
\item For any $(n/r)$-shallow wide cylinder $\cyl\in\CC$, there is a generalized cylinder 
$\gcyl_\tau\in\test$ that contains $\cyl$.
\end{enumerate} 
\end{lemma}
\begin{proof}
(i) It is well known, by the construction of vertical decompositions, that every cell $\tau$ is a 
semi-algebraic set defined as a conjunction of at most $2\cdot 5 = 10$ polynomial inequalities, 
each of constant degree. Let $\Phi_\tau: \reals^5\rightarrow\{0,1\}$ be the resulting predicate, i.e.,
$\Phi_\tau(\yy)=1$ if $\yy \in \tau$. By applying \eqref{cyl-pred}, the generalized cylinder 
$\gcyl_\tau$ can be expressed as
\[ 
\gcyl_\tau = \{ \xx\in\reals^3 \mid \exists \yy \text{ such that } \Pi(\xx,\yy) \wedge \Phi_\tau(\yy) \} .
\]
Using a quantifier-elimination procedure~\cite{BPR}, $\gcyl_\tau$ can be expressed as a quantifier-free 
formula, simply consisting of Boolean operations on polynomial inequalities, implying that
$\gcyl_\tau$ is a semi-algebraic set. Since $\Pi$ and $\Phi_\tau$ are constant-complexity semi-algebraic 
predicates, so is the formula produced by the quantifier-elimination procedure (albeit with larger
constants). Hence, $\gcyl_\tau$ is a semi-algebraic set of constant complexity.

Next, suppose that $\gcyl_\tau$ contains a point $q \in P$. Let $\cyl$ be a cylinder with 
$p_\cyl \in \tau$ such that $q\in \cyl$. Then, by definition, $p_\cyl \in K_q$. 
Since $p_\cyl\in\tau$ and $\tau$ is contained in a cell of $\A(N^*)$, we conclude that 
if $\tau \subset K_q$ then $q\in N$. But the level of $\tau$ in $\A(N^*)$ is at most 
$k=\ceil{c\log r}$, so at most $k$ regions of $N^*$ contain $\tau$. Following the same 
argument as in~\cite{ShSh} (based on the notion of so-called shallow nets), one can now 
argue that, with a suitable choice of constants, with probability at least $1-1/r^{O(1)}$, 
there are at most $n/r$ points $q$ of $P$ such that $K_q$ contains or crosses $\tau$. Hence,
$\gcyl$ is $(n/r)$-shallow with probability at least $1-1/r^{O(1)}$. Taking the probability
union bound over all cells of $\VeD{\A}_{\le k}(N^*)$, we conclude that all generalized 
cylinders in $\test$ are $(n/r)$-shallow with probability at least $1/2$.

(ii) Since $N$ is a random sample of size $O(r\log r)$ points, it follows from~\cite{ShSh}
that, with a suitable choice of the constants, for all $(n/r)$-shallow cylinders $\cyl$, 
we have $|\cyl\cap N | \le c\log r$. Hence, the level of the point 
$\pp_C$ in $\A(N^*)$ is at most $c\log r$, implying that there 
exists a cell $\tau \in \VeD{\A}_{\le k}(N^*)$ that contains 
$\pp_C$. By construction $\gcyl_\tau$ contains $C$, as claimed.
\end{proof}

\paragraph{Union complexity of the test-set.}
We next prove that, for any subset $\QQ\subseteq\test$ of $m$ generalized cylinders,
the complexity of $\C(\QQ)\cap\sigma$, where $\sigma$ is the vertical slab containing $P$, is $O^*(m^2)$. 
By \lemref{union}, we can then conclude that $\C(\QQ)\cap\sigma$ can be decomposed 
into $O^*(m^2)$ elementary cells.  It is known~\cite{SA,VD1} that the complexity of 
$\C(\QQ)\cap\sigma$ is $O(m^2+\chi)$, where $\chi$ is the number of vertices on 
$\bd\C(\QQ)\cap \sigma$, so it suffices to show that $\chi = O^*(m^2)$. Here is 
a high-level overview of the proof: we partition the boundary of each generalized 
cylinder $\gamma$ into $O(1)$ \emph{surface patches} and choose a set $\Z$ of $O(1)$ 
canonical directions. For each direction $u \in \Z$, we choose two subsets $\P^-_u, \P^+_u$ 
of surface patches over all generalized cylinders, and consider the boundary of the 
\emph{sandwich region} $\Sigma_u$ enclosed between the lower envelope of $\P^-_u$
and the upper envelope of $\P_u^+$ 
in direction $u$ (i.e., we rotate the coordinate frame so that $u$ is the $z$-axis). 
We argue that each vertex of $\C(\QQ)\cap\sigma$ appears on the boundary of at least one $\Sigma_u$.
The result in~\cite{ASS-96} implies that the overall complexity of these sandwich 
regions, and thus of $\C(\QQ)\cap\sigma$, is only $O^*(m^2)$.  We now give the proof in detail.

Let $\kappa$ be a sufficiently large constant. For each generalized cylinder 
$\gcyl \in \test$, we partition its boundary into \emph{canonical surface patches},
such that each patch is smooth and the outer normals of the points on a single patch 
do not vary by more than $2\pi/\kappa$ (as points on the unit sphere). 
Since $\gcyl$ is a semi-algebraic set of constant complexity, $\bd{\gcyl}$ 
can be partitioned into $O(1)$ canonical surface patches, where the 
constant of proportionality depends on $\kappa$ and on the complexity of $\gcyl$.\footnote{%
  Since $\gcyl$ is a semi-algebraic set of constant complexity, its Gauss map 
  $N_\gcyl: \gcyl\rightarrow\sphere^2$ is also a semi-algebraic set of constant complexity
  (with a suitable parameterization of $\sphere^2$). We can thus decompose $N_\gcyl$ into
  $O(1)$ pieces, each of which is monotone in the radial direction, i.e., any ray emanating 
  from the origin intersects a patch in at most one point. We can then partition each piece 
  into $O(\kappa^2)$ pieces, so that each piece lies in a spherical cap of radius $\pi/\kappa$. 
  Each patch of $\gcyl$ corresponds to each resulting piece. }
(As a concrete example demonstrating this decomposition, if $\gcyl$ is just a standard cylinder, 
the canonical surface patches are portions of the boundary of $\gcyl$, each of which is bounded by a 
pair of generator lines, parallel to the axis of $\gcyl$, with an angular span of $2\pi/\kappa$ of their normals. 
)
Let $\P$ be the set of canonical patches created over all generalized cylinders in $\QQ$,
clipped within the slab $\sigma$.

A direction $\rho$ in the unit sphere $\sphere^{2}$ is {\it good} for a canonical patch $\tau$ 
if the following two conditions hold:
\begin{enumerate}[(i)] 
\item The angle between $\rho$ and the (outer) normal $\nn$ of either of the planes $H$, $H'$ 
bounding $\sigma$ is smaller than $\pi/2 - \pi/\kappa$.  
(That is, $\rho$ is not ``nearly parallel'' to $H$, $H'$.)
\item Each line tangent to (the relative interior of) $\tau$ at any point $p$ forms an angle 
of at least $\pi/\kappa$ with $\rho$.  (That is, $\rho$ is not ``too tangent'' to $\tau$.)
\end{enumerate}

We say that a direction $\rho$ is \emph{good for a vertex} $v$ of $\A(\QQ)$
if it is good for each of the at most three patches of $\QQ$ containing $v$.
(When $\rho$ is good, $-\rho$ is also considered to be good, and we treat both directions together.)

Following the analysis in~\cite{AS-00, Ezra-11}, we obtain the following lemma, whose proof, 
omitted here, can be found in the original papers.
\begin{lemma}
\label{lem:sample-dirs}
There exists a set $\Z$ of $O(1/\kappa^2)$ directions in $\sphere^2$ so that every vertex
of $\A(\QQ)$ has at least one good direction in $\Z$, assuming that $\kappa$ is sufficiently
large.
\end{lemma}

\begin{figure}[htb]
  \centering
	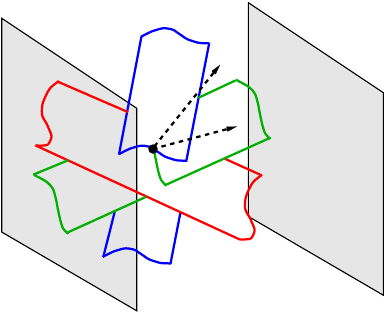
  	\includegraphics[width=2.3in]{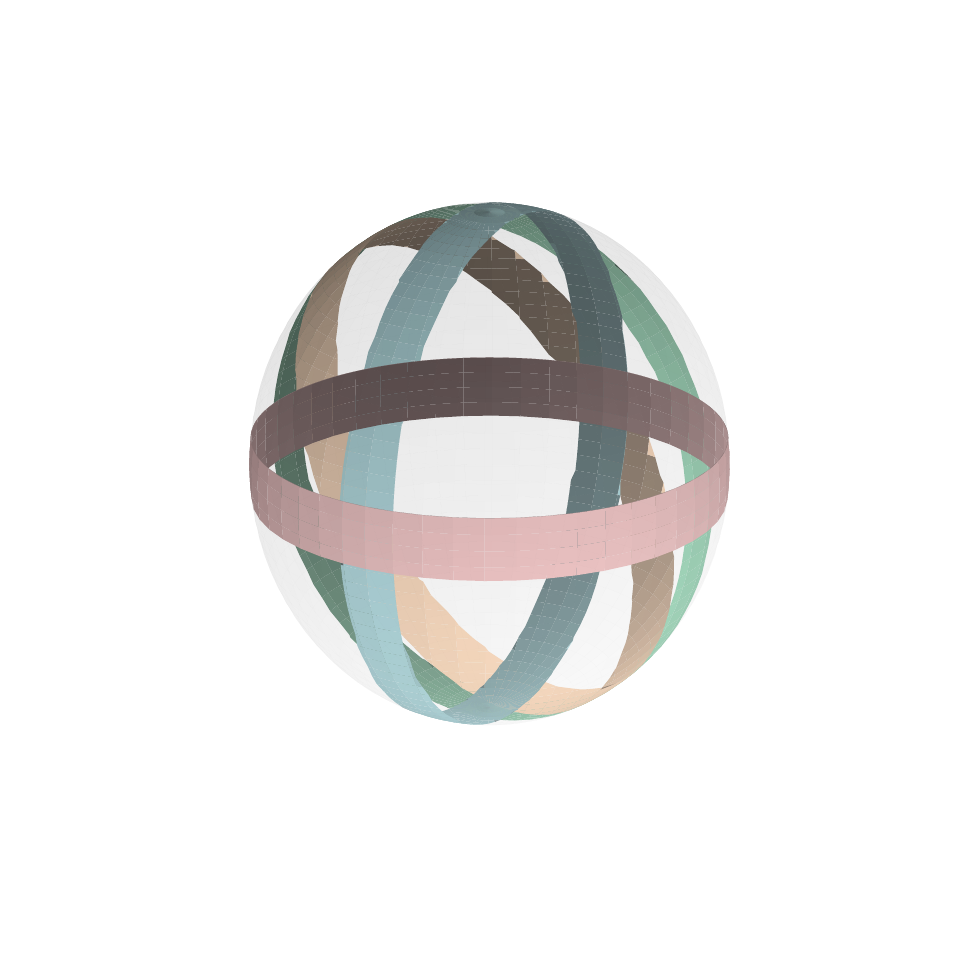}

	\caption{(left) Good direction of a vertex; (right) set of bad directions for a vertex.}
	\label{fig:good-dir}
\end{figure}

We now prove the main property of $\C(\QQ)$.

\begin{lemma} [Good directions; see~\cite{AS-00, Ezra-11}]
  \label{lem:good_directions}
Let $\sigma$ be a vertical slab of width at most $w_0 = 2\sin^2(\pi/\kappa)$, where $\kappa$ is a 
sufficiently large constant.  Let $\tau$ be a surface patch in $\P$ that intersects $\sigma$, let 
$\gcyl$ be the generalized cylinder whose boundary contains $\tau$, and let
$\rho \in \Z$ be a good direction for $\tau$.
For any point $u \in \tau \cap \sigma$, if a line $\ell$ in direction $\rho$ enters 
$\gcyl$ at $u$, then $\ell$ intersects $\bd\sigma$ before exiting $\gcyl$.
\end{lemma}

\begin{proof}
It follows by the definition of a generalized cylinder that $u$ lies on the boundary of a wide cylinder
$\cyl$ fully contained in $\gcyl$. Let $r\ge 1$ 
be the radius of $\cyl$, and let $x$ be the point on the axis of $\cyl$ closest to $u$, i.e., $\|x-u\| = r$.
Let $B$ be the ball of radius $r$ centered at $x$; we thus have $B \subset \cyl \subseteq \gcyl$ and $u \in \bd{B}$.
Let $h_u$ be the tangent plane to $B$ at $u$. Using techniques from elementary differential 
geometry~\cite{Bar-2010}, it can be verified that $h_u$ is also the tangent plane of $\gcyl$ 
at $u$ (which holds since $u$ is a point in the relative interior of $\tau$).

Let $u'$ be the other intersection point of $\bd{B}$ and $\ell$, 
and let $u''$ be the other intersection point of $\bd{\gcyl}$ and $\ell$.
Since $B \subset \cyl \subseteq \gcyl$ it follows that $u'$ must be contained in $\gcyl$, so,
as we walk along $\ell$ from $u$ into $\gamma$, $u''$ must lie further than $u'$.
We now observe that $|uu'|$ is minimized when $uu'$ 
forms an angle of $\pi/\kappa$ with $h_u$, since by property (ii) of good directions, 
any such angle is at least $\pi/\kappa$ with respect to the tangent plane to $\tau$ (which is 
identical to $h_u$ as just noted).  It thus follows that $|uu'| \ge 2\sin{(\pi/\kappa)}$. 
Indeed, assume, without loss of generality, that $h_u$ is parallel to the $xy$-plane and 
that $u$ is the north pole of $B$. By the above condition, the vector $uu'$ must lie in a 
vertical cone (its axis is the vertical ray from $u$ in the negative $z$-direction) whose apex 
angle is $\pi/2 - \pi/\kappa$. Furthermore  $r\ge 1$, so the lower bound on $|uu'|$ follows.

Finally, $uu'$ forms an angle of at most $\pi/2 - \pi/\kappa$ with the normal $\nn$
(to either of the two bounding planes of $\sigma$), so the length of the projection of $uu'$ 
on $\nn$ is at least $2|uu'| \sin{(\pi/\kappa)} \ge 2\sin^2{(\pi/\kappa)}$. Since the width 
of $\sigma$ is at most $2\sin^2(\pi/\kappa)$, the lemma follows.
\end{proof}

\begin{figure}[htb]
  \centering
  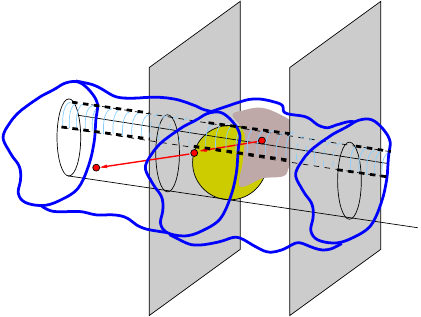
	\caption{Proof of Lemma~\ref{lem:good_directions}.}
  \label{fig:walk_ray}
\end{figure}

For a direction $\rho\in\Z$, let $\P_{\rho}$ be the set of canonical patches 
for which $\rho$ is a good direction.
A (clipped) patch $\tau$ in $\P_\rho$ is called \emph{$\rho$-upper} 
(resp., \emph{$\rho$-lower}) if, for any point $u \in \tau$, 
the point $u + \alpha \rho$ lies in the exterior (resp., interior) 
of the generalized cylinder $\gcyl$ associated with $\tau$, for any sufficiently small
parameter $\alpha > 0$. Let $\P_\rho^+$ (resp., $\P_\rho^{-}$) be the 
set of the $\rho$-upper (resp., $\rho$-lower) patches of $\P_\rho$.
The \emph{$\rho$-upper envelope} $\EE_\rho^+$ of $\P_\rho^{+}$ is the set of points 
$u$ on the patches of $\P_\rho^{+}$, such that a ray from $u$ in the
$(+\rho)$-direction does not meet any other patch in $\P_\rho^{+}$.
The \emph{$\rho$-lower envelope} $\EE_\rho^-$ of $\P_\rho^{-}$ is defined analogously.

\begin{lemma}
	\label{lem:sandwich-vertex}
Let $v$ be a vertex of $\C(\QQ)\cap\sigma$.
Then there exists a direction $\rho\in\Z$ such that $v$ is 
a vertex of the sandwich region enclosed between $\EE_\rho^+$ and $\EE_\rho^-$.
\end{lemma}
\begin{proof}
Suppose $v$ is incident to three patches $\tau_1,\tau_2, \tau_3$ (the case where $v$ is
formed by only two patches is much simpler). By \lemref{sample-dirs}, there is a direction 
$\rho\in\Z$ that is good for $v$, i.e., good for all three patches, so all of them belong 
to $\P_\rho$. Some (perhaps none, or all) of these patches belong to $\P_\rho^+$ and some
to $\P_\rho^-$. Fix one of these patches $\tau_i$, and suppose that $\tau_i\in\P_\rho^+$.
Then $v$ lies on the $\rho$-upper envelope $\EE_\rho^+$. Indeed, if not, the ray $v+\alpha\rho$,
for $\alpha>$, intersects another upper patch $\tau'$ at a point $v'$ before exiting the slab 
$\sigma$. Since $\tau'$ is an upper patch, the oppositely directed ray $v'-\alpha\rho$ enters 
the generalized cylinder $\gcyl'$ containing $\tau'$, but the ray exits 
$\gcyl'$ before leaving $\sigma$ because $v \in \C(\QQ)$ and 
$vv'\subset\sigma$, which contradicts \lemref{good_directions}. 
Hence, we conclude that $v$ lies on $\EE_\rho^+$. Similarly if $\tau_i$ is a lower patch, 
$v$ lies on the lower envelope $\EE_\rho^-$ of $\P_\rho^-$. Since this holds for each $i$,
we conclude that $v$ is a vertex of the sandwich region enclosed between $\EE_\rho^+$ 
and $\EE_\rho^-$. (It is either a vertex of $\EE_\rho^+$, or a vertex of $\EE_\rho^-$, 
or an intersection of an edge of $\EE_\rho^+$ with $\EE_\rho^-$, 
or an intersection of an edge of $\EE_\rho^-$ with $\EE_\rho^+$.) 
\end{proof}

The number of vertices of the sandwich region is $O^*(m^2)$~\cite{ASS-96}. 
Hence, for any direction $\rho\in\Z$, there are $O^*(m^2)$ vertices 
on $\C(\QQ)\cap\sigma$ for which $\rho$ is a good direction.
Summing over all good directions of $\Z$, we obtain the following:
\begin{lemma}
  \label{lem:generalized_cylinders}
  Let $\sigma$ be a vertical slab of width $w_0$, and let $\QQ \subseteq \test$ be any subset of $m$ 
  generalized cylinders. Then the complexity of $\C(\QQ)\cap \sigma$, the complement of the union of $\QQ$ clipped within $\sigma$, is $O^*(m^2)$.
\end{lemma}
Combining this lemma with \lemref{union}, we obtain the following:
\begin{corollary}
  \label{cor:generalized_cylinders}
  Let $\sigma$ be a vertical slab of width $w_0$, and let $\QQ \subseteq \test$ be any subset of $m$ 
  generalized cylinders. Then $\C(\QQ)\cap \sigma$ can be partitioned into $O^*(m^2)$ elementary cells.
\end{corollary}
Putting everything together, we obtain the main result of the test-set construction:

\begin{lemma}
  \label{lem:cyl-test-set}
  Let $P \subset \reals^3$ be a set of $n$ points in $\reals^3$ lying in a vertical slab of width 
  at most $w_0$, and let $r\ge 1$ be a parameter. A test set $\test$ of size $O^*(r^6)$ for 
  all cylinders of radius at least $1$ that are $(n/r)$-shallow with respect to $P$, with 
  $\zeta(m) = O^*(m^2)$ and $\delta=1$, can be computed in $nr^{O(1)}$ expected time.
\end{lemma}

\medskip
\noindent{\bf Remark:}
The proof of \lemref{generalized_cylinders} is a variant of the analysis in~\cite{Ezra-11}, 
where standard cylinders were considered. However, a major difference between the two settings, of 
standard cylinders and generalized cylinders, is that whereas the overall complexity of the union of 
standard cylinders is $O^*(m^2)$~\cite{AS-00, Ezra-11}, this is not necessarily the case for generalized 
cylinders. Indeed, we can form a set of $m$ generalized cylinders by taking the Minkowski sums of planes 
in $\reals^3$ with some sufficiently small ball (it is easily checked that such a sum can be a generalized 
cylinder). Given such a collection, we can easily form a 
three-dimensional grid-like construction, which results in $\Omega(m^3)$ union complexity. The 
crucial property, however, is the fact that the generalized cylinders are \emph{wide} in 
$\sigma$, in which case \lemref{generalized_cylinders} implies that the complexity 
of the union is relatively small (i.e., nearly quadratic) if we remain inside $\sigma$.

\subsection{Segment NN queries: Linear-size data structure}
\label{subsec:seg-linear}

We now extend the data structure described in \subsecref{pt-linear} for answering fixed-radius 
neighbor queries with respect to segments, namely, for a query segment $e$ in $\reals^3$ and a 
distance parameter $\rho>0$, the problem is to determine whether there is a point of $P$ within 
distance $\rho$ from $e$. If so, we want to return such a point. Combining this data structure,
whose details follow, with parametric search, we can also answer NN queries with segments.
For a segment $e$ and a real value $t>0$, let $\cigar(e,t) := e \oplus B_t = \{ x+y \mid x\in e, y \in B_t\}$, 
where $B_t$ is the ball of radius $t$ centered at the origin. Borrowing the terminology 
of~\cite{AS-00}, we refer to $\cigar(e,t)$ as a \emph{cigar}.
We build a data structure for answering cigar-emptiness queries.

\paragraph{Data structure.}
As for the case of lines, we build a two-level data structure, where the top level is a 
linear-size partition tree $T$ on $P^*$, the $xy$-projection of $P$, for 2D halfplane range 
searching. The only difference is that at each node $v$ of the top-level tree, we now build 
a secondary data structure for answering emptiness queries with respect to wide cigars. 
More precisely, this is for the setup where $P_v$ lies in a vertical slab $\sigma$ of width 
at most $w_0$, and the radius of a query cigar is at least $1$. This requires adapting the 
test-set construction to cigars, as described below.

\begin{figure}
  \begin{center}
    \includegraphics[scale=0.5]{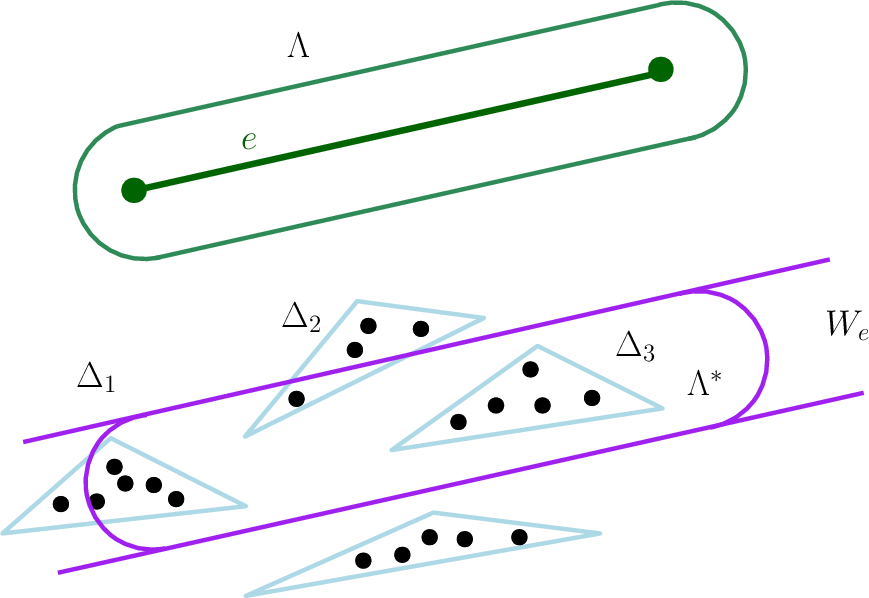}
    \caption{A query cigar and its $xy$-projection.}
    \label{fig:cigar-query}
  \end{center}
\end{figure}

\paragraph{Query procedure.}
Let $\cigar := \cigar(e, \rho)$ be a query cigar. Let $\cigar^*$ 
(resp., $e^*$) be the $xy$-projection of $\cigar$ 
(resp., $e$), let $\ell_e$ be the line containing $e$, and 
let $W_e\subset\reals^2$ be the strip of width $2\rho$ containing 
$\cigar^*$; see \figref{cigar-query}. 

We visit $T$ recursively in a top-down manner. Suppose we are at a node $v$. Let $\triangle_v$ be the triangle associated with $v$.
If $\cigar^* \cap \triangle_v = \emptyset$,
we return to the parent of $v$, so assume that $\cigar^*\cap\triangle_v\ne\emptyset$. 
If $v$ is a leaf, we test whether $\cigar_v\cap P_v\ne\emptyset$ in a brute-force manner and 
return a point in the intersection if the answer is {\sc Yes}. If $v$ is an internal node, we proceed as follows. 
If $\triangle_v \subset W_e$, we query the secondary structure with $\cigar$ to test whether 
$P_v\cap\cigar\ne\emptyset$. Finally, if $\triangle_v$ intersects a boundary line of $W_e$, we 
recursively visit the children of $v$. If any one of them returns a point of 
$P_v \cap \cigar$, we return that point; otherwise we return no.

Assuming that the secondary data structure answers a query correctly, the correctness of the 
overall query procedure follows immediately. A suitable extension of the analysis in 
\subsecref{pt-linear} implies that the overall query time is $O^*(\sqrt{n})$.

\paragraph{Test-set construction.}
Recall that we query the secondary structure with $\cigar$ only if $\triangle_v \subset W_e$, 
implying that the width of $\triangle_v$ is at most $2\rho$. We can therefore assume that we 
build the data structure for a point set $P$ that lies in a vertical slab of width at most 
$w_0$, for a sufficiently small constant $w_0$, and it is queried with a cigar of radius at 
least $1$. The test-set construction follows the same approach as in \subsecref{test-construct}. 
A cigar $\cigar :=\cigar(e,\rho)$ is represented using $7$ real parameters 
$(a, b, \rho) \in \reals^3\times\reals^3\times\reals_{\ge 0}$, where $a,b$ are the 
endpoints of the segment $e$, and $\rho$ is the radius of $\cigar$. For a point $p \in P$, 
we define $K_p \subset \reals^{7}$ to be the set of points that correspond to cigars
that contain $p$ and have radii that are at least $1$. With these definitions at 
hand, the construction of the test set $\test$, a set of \emph{generalized cigars}, 
is exactly the same as in \subsecref{test-construct}, except that we work in a seven-dimensional 
parametric space. \lemref{good_directions} still holds because for any generalized cigar $\gamma$ 
and for any point $u \in \bd\gamma\cap\sigma$, the ball of radius $1$ tangent to $\gamma$ at $u$ 
lies inside $\gamma$, and thus the proof follows verbatim, as can be verified. Hence, the complexity 
of $\C(\QQ)\cap\sigma$, for any subset $\QQ\subset\test$ of $m$ generalized cigars, is $O^*(m^2)$.
(The size of the test set is now $O^*(r^{10})$, which is the best known bound on the vertical
decomposition in $\reals^7$.)

Omitting all the fairly easy further details, we conclude that $P$, a point set lying in 
a vertical slab of width $w_0$, can be preprocessed into a linear-size data structure so that, 
for a cigar $\cigar$ of radius at least $1$, it can detect, in $O^*(\sqrt{n})$ time, whether 
$\cigar\cap P\ne\emptyset$. Putting everything together, we obtain the following main result
of this subsection:

\begin{theorem}
  \label{thm:NNSP}
  A set $P$ of $n$ points in $\reals^3$ can be preprocessed, in $O(n\log n)$ expected time, 
  into a data structure of size $O(n)$, so that (i) for any query segment $e$ in $\reals^3$, 
  the point of $P$ nearest to $e$ can be computed in $O^*(n^{1/2})$ time; 
  (ii) for a query segment $e$ and a distance parameter $\rho$, all points of $P$ 
  within distance $\rho$ from $e$ can be reported in $O^*(\sqrt{n})+O(k)$ time.
\end{theorem}

\smallskip
\paragraph{Remark.} 
A natural question is whether the nearest-neighbor of a query triangle in $P$ can also be 
computed in $O^*(\sqrt{n})$ time using a linear-size data structure. Unfortunately, it is 
unlikely that such a data structure exists, because computing the NN of a query plane is 
as hard as answering slab-emptiness queries in $\reals^3$. The known lower bounds for simplex 
range searching also hold for slabs~\cite{Ag:rs}. A linear-size data structure with $O^*(n^{2/3})$ query 
time follows easily from the known results on semi-algebraic range searching in $\reals^3$~\cite{MP,AMS}.

\subsection{Line NN queries: Fast query time}
\label{subsec:Fast_NN1}

We now describe a data structure that handles the complementary regime of fast query time.
That is, it answers an NN query with a line among points in $\reals^3$, in $O^*(1)$ time, 
using larger storage, by working in the $4$-dimensional parametric space $\Lspace$ 
of (the query) lines, where a line $\ell$ 
in $\reals^3$, given by the equations $y=ax+c$ and $z=bx+d$, is represented as the point 
$\ell^* = (a,b,c,d)\in \Lspace$ (we will actually slightly modify this representation---see below).\footnote{%
  For convenience (and with no loss of generality if one assumes general position of the lines of $\L$), 
  we ignore the fact that this space is actually projective, and that some lines may fail
  to have this representation.}
As we will show, the structure requires $O^*(n^4)$ storage, and can be constructed in $O^*(n^4)$
expected time.

We begin by describing the distance function between a point and a line in $\reals^3$, and the 
four-dimensional Voronoi diagram induced by the points of $P$ in $\Lspace$.

\paragraph{Distance function and Voronoi diagram.}
Let $\ell^* = (a,b,c,d) \in \Lspace$.
For $a, b\in\reals$ fixed, the (unnormalized) direction of $\ell$, $(1,a,b)$,
is fixed. Let $\pi_\ell$ be the plane that is orthogonal to $\ell$ (i.e., with normal direction $(1,a,b)$) 
and passes through the origin; note that $\pi_\ell$ depends on the parameters $a, b$ but not on $c,d$, 
so we will also use $\pi(a,b)$ to denote $\pi_\ell$.
Redefine the representation of $\ell$ so that $(c,d)$ is actually the intersection point
of $\ell$ with $\pi_\ell$, in a suitable canonical coordinate frame within $\pi(a,b)$ (we omit here the 
easy details of specifying this frame, noting that it does depend (only) on $(a,b)$). Write $u=(1,a,b)$.

For a point $p\in P$, let $p^\downarrow(a,b)$ denote its projection onto $\pi(a,b)$. 
It is well known, and easy to show, that 
\[
p^\downarrow (a,b) = p - \frac{p\cdot u}{|u|^2} u .
\]

Write $p^\downarrow (a,b) = (x_p(a,b),y_p(a,b))$ (in the coordinate frame within $\pi(a,b)$).
The distance between $p$ and $\ell$, denoted by $\dist(p,\ell)$, is then the distance between 
$p^\downarrow (a,b)$ and $(c,d)$. That is,
\begin{align}
	\label{eq:dist}
\dist^2(p,\ell) & = (x_p(a,b) - c)^2 + (y_p(a,b) - d)^2 \nonumber \\
& = (x_p^2(a,b) + y_p^2(a,b)) - 2cx_p(a,b) - 2dy_p(a,b) + (c^2+d^2) .
\end{align}

For a query line $\ell$, our goal is to compute the point $p\in P$ that is closest to $\ell$, 
which is equivalent to computing $\arg\min_{p\in P} \dist^2(p,\ell)$, i.e., the point that 
minimizes (\ref{eq:dist}). Since $c^2+d^2$ is common to all points $p$, we can drop it, and 
seek the point $\arg\min_{p \in P} f_p(a,b,c,d)$, where
\begin{equation} \label{eq:fp}
f_p(a,b,c,d) = g_p(a,b) - 2cx_p(a,b) - 2dy_p(a,b) ,
\end{equation}
and $g_p(a,b) = x_p^2(a,b) + y_p^2(a,b)$. The following property follows immediately from \eqref{fp}:
\begin{lemma}
\label{lem:bisector}
(i) For any fixed pair $(a,b)$, $f_p$ is a bivariate linear function in $c$ and $d$, i.e., its graph is a plane. 

\noindent
(ii) For any fixed triple, $(a, b, c)$ (resp., $(a,b,d)$), the graph of $f_p$ (as a function of
the fourth variable) is a line.
\end{lemma}

Let $\F = \{f_p \mid p\in P\}$ be the resulting set of $n$ $4$-variate functions.  
The Voronoi diagram of $P$ in $\Lspace$ is the minimization diagram $\M := \M_\F$ of $\F$
(namely, the 4D projection of the lower envelope of $\F$).
Each cell $\cell$ of $\M$ is associated with a point $p \in P$ that is the nearest neighbor of 
all lines $\ell$ such that $\ell^*\in\cell$.  Let $\M_p$ denote the family of cells of $\M$ that 
are associated with $p$, i.e., the Voronoi region of $p$ in $\M$ (note that a Voronoi region might 
be disconnected). For $q\ne p$, let 
\[
K_{p,q} = \{ \ell^*\in\Lspace \mid f_p(\ell^*) \le f_q(\ell^*) \}.
\]
Then $\M_p = \bigcap_{q\in P\setminus\{p\}} K_{p,q}$.

For a query line $\ell$, we wish to locate the cell of $\M$ containing $\ell^*$. The best known data 
structure can answer a query in $O^*(1)$ time using $O^*(n^5)$ space and expected preprocessing time. 
Currently, we do not know how to preprocess general four-dimensional minimization diagrams, like $\M$, 
into a data structure of size $O^*(n^4)$ for answering point-location queries in $O^*(1)$ time.
We circumvent this issue, and manage to achieve this storage bound, by exploiting some monotonicity 
properties of the Voronoi regions of $\M$, as discussed below.

\paragraph{Structure of Voronoi regions.}
For a pair of distinct points $p,q\in P$, let $\sigma_{p,q}$ denote the three-dimensional 
surface $f_p = f_q$, i.e., the bisector of $p$ and $q$ in $\Lspace$, which is defined by the equation
\[
g_p(a,b) - 2cx_p(a,b) - 2dy_p(a,b) = g_q(a,b) - 2cx_q(a,b) - 2dy_q(a,b).
\]
Assuming $y_q(a,b) \ne y_p(a,b)$ (see below for handling the case of equality), 
$\sigma_{p,q}$ is the graph of a trivariate function $\psi_{p,q}: \Lspace^{(0)} \rightarrow \reals$, where 
$\Lspace^{(0)} \subset \Lspace$ is the $3$-dimensional hyperplane $d=0$, given by
\begin{equation} \label{eq:psipq}
d = \psi_{p,q}(a,b,c) := \frac{g_q(a,b) - g_p(a,b)}{2(y_q(a,b) - y_p(a,b))} 
- \frac{x_q(a,b) - x_p(a,b)}{y_q(a,b) - y_p(a,b)} \; c .
\end{equation}
As we will be dealing with inequalities of the form $d \le \psi_{p,q}(a,b,c)$ of 
$d \ge \psi_{p,q}(a,b,c)$, we replace $\psi_{p,q}(a,b,c)$ by the following two partial functions 
$\psi_{p,q}^+, \psi_{p,q}^- : \Lspace^{(0)} \rightarrow \reals$, given by
\begin{align}
  \label{eq:surfaces}
  \begin{split}
    \psi_{p,q}^+(a,b,c) & = 
    \begin{cases}
      \psi_{p,q}(a,b,c) & y_q(a,b) - y_p(a,b) > 0 , \\
      +\infty & \text{otherwise.}
    \end{cases} \\
    \psi_{p,q}^-(a,b,c) & = 
    \begin{cases}
      \psi_{p,q}(a,b,c) & y_q(a,b) - y_p(a,b) < 0, \\
      -\infty & \text{otherwise} .
    \end{cases} 
  \end{split}
\end{align}
Using \eqref{surfaces}, we can express $K_{p,q}$ as
\begin{equation} \label{eq:kpq}
  K_{p,q} = \left\{ (a,b,c,d) \in \Lspace \mid 
  \psi_{p,q}^-(a,b,c) \le d \le \psi_{p,q}^+(a,b,c) \right\} ,
\end{equation}
noting that at each $(a,b)$ only at most one of $\psi_{p,q}^-$, $\psi_{p,q}^+$ is defined (i.e., $\ne \pm\infty$);
they are both undefined when the denominator vanishes.
For a triple $(a,b,c)$, let $\rho(a,b,c)$ be the line in the parametric space $\Lspace$ parallel to the $d$-axis 
and passing through the point $(a,b,c,0)\in\Lspace$. 
Then \eqref{kpq} and \eqref{surfaces} imply that
\begin{equation}
  \label{eq:bisect-line}
	  K_{p,q}\cap \rho(a,b,c)  = 
    \begin{cases}
	    (-\infty,\psi_{p,q}^+(a,b,c)],\qquad & y_q(a,b) - y_p(a,b) > 0 , \\
	    [\psi_{p,q}^-(a,b,c),+\infty),\qquad & y_q(a,b) - y_p(a,b) < 0 .
    \end{cases}
\end{equation}

Eq.~\eqref{surfaces} also implies that
\begin{equation} \label{eq:sandwich}
  \M_p = \left\{ (a,b,c,d) \in \Lspace \mid \max_{q\in P\setminus\{p\}} \psi_{p,q}^-(a,b,c) \le d \le 
  \min_{q\in P\setminus\{p\}} \psi_{p,q}^+(a,b,c) \right\} .
\end{equation}
That is, $\M_p$ is the \emph{sandwich region} between the lower envelope (with respect to the $d$-direction)
of the functions $\psi_{p,q}^+$ and the upper envelope of the functions 
$\psi_{p,q}^-$, over $q\in P\setminus \{p\}$. Alternatively, $\M_p=\M_p^-\cap \M_p^+$, 
where $\M_p^-$ (resp., $\M_p^+$) is the region lying below 
the lower envelope (resp., above the upper envelope) of the functions  $\psi_{p,q}^+$  (resp., $\psi_{p,q}^-$),
over $q\in P\setminus\{p\}$, in the $d$-direction.  This condition implies that $\M_p$ is monotone in the $d$-direction.
See Figure~\ref{fig:EMpm} for an illustration.

\begin{figure}[htb]
  \begin{center}
    \includegraphics[scale=0.50]{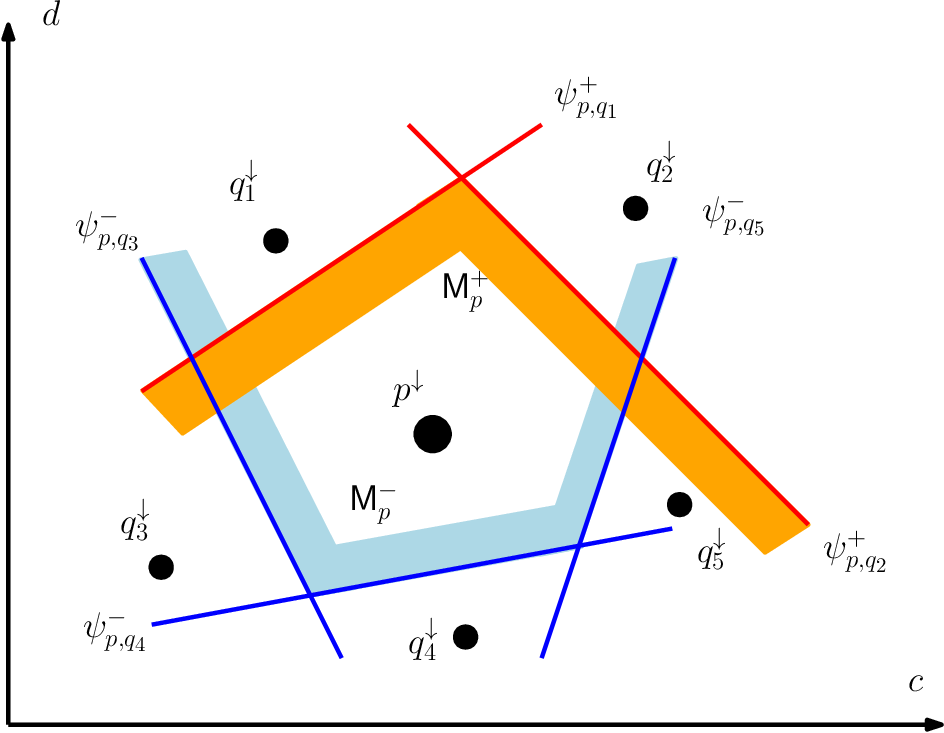} 
  \end{center}
  \caption{The structure of the decomposition of the lower envelope and the minimization diagram of $\F$.}
  \label{fig:EMpm}
\end{figure}

As already noted, $\psi_{p,q}(a,b,c)$ and its partial sub-functions $\psi^+_{p,q}(a,b,c)$, $\psi^-_{p,q}(a,b,c)$
are not defined when $y_q(a,b) = y_p(a,b)$. This situation will be dealt with later on.

We conclude this discussion with the following observation, which will be a key ingredient
for our data structure:

\begin{figure}[htb]
  \begin{center}
    \begin{tabular}{cc}
      \includegraphics[scale=0.52]{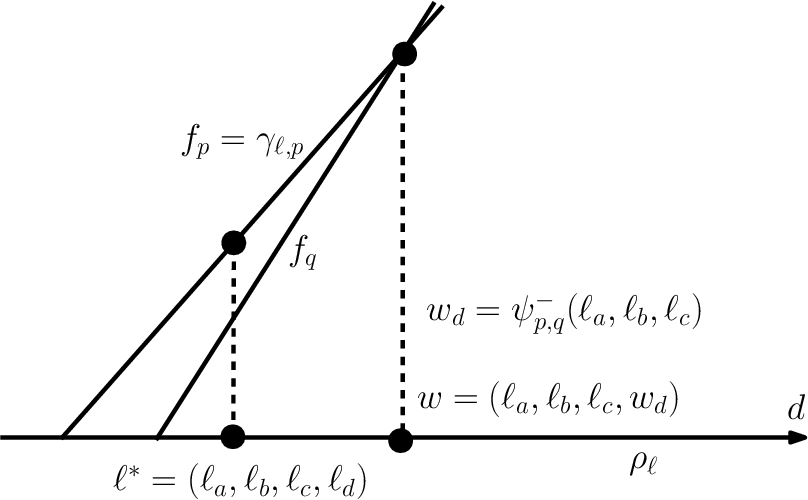} 
      & \includegraphics[scale=0.52]{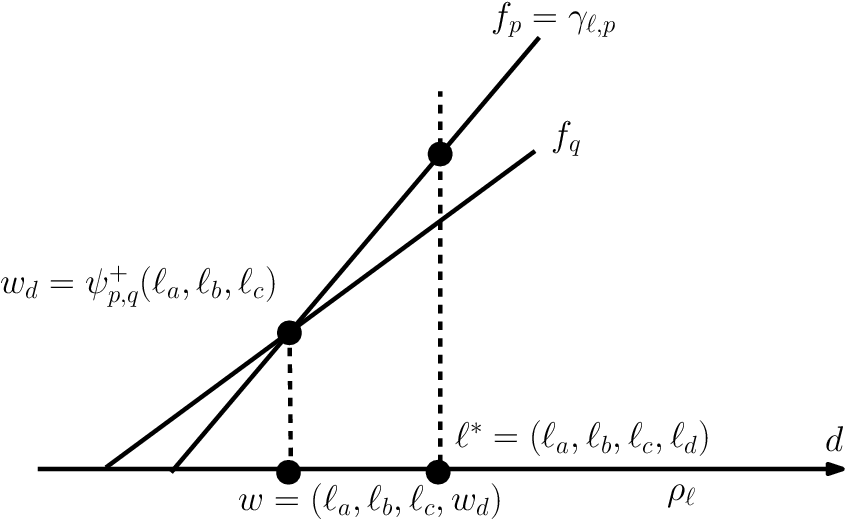}\\
      (a) 
      &(b)
    \end{tabular}
    \caption{Illustration of Lemma~\ref{lem:nearer-point}.
          (a) The case $y_q(\ell_a, \ell_b) < y_p(\ell_a, \ell_b)$;
          (b) the case $y_q(\ell_a, \ell_b) > y_p(\ell_a, \ell_b)$.
    }
    \label{fig:lem71}
  \end{center}
\end{figure}
  

\begin{lemma}
  \label{lem:nearer-point}
  Let $\ell^* = (\ell_a,\ell_b,\ell_c,\ell_d)\in \Lspace$, and let $p$ be a point of $P$. Let
  $\rho_\ell$ denote the line in the $d$-direction in $\Lspace$ passing through $\ell^*$,
  and let the line $\gamma_{\ell,p}$ be the graph of the function $f_{\ell,p}(d) = f_p(\ell_a,\ell_b,\ell_c,d)$
  (clearly, this function is independent of $\ell_d$).
  Let $q$ be a point of $P$ that is nearer to $\ell$ than $p$ (i.e., $f_q(\ell^*) < f_p(\ell^*)$)
  and $y_p(a,b)\neq y_q(a,b)$, assuming that such a point exists.
	Then $\gamma_{\ell,q}$, the graph of the linear function $f_{\ell,q}(d) = f_q(\ell_a,\ell_b,\ell_c,d)$  
	intersects $\gamma_{\ell,p}$ 
  at a unique point, whose projection on $\rho_\ell$ is denoted as $w=(\ell_a,\ell_b,\ell_c,w_d)$. 
  Furthermore, if $w_d > \ell_d$ (resp., $w_d < \ell_d$) then we have 
  $w_d = \psi_{p,q}^-(\ell_a,\ell_b,\ell_c)$ (resp., $w_d=\psi_{p,q}^+(\ell_a,\ell_b,\ell_c)$).
\end{lemma}

\begin{proof}
Refer to \figref{lem71}.
Since $q$ is closer to $\ell$ than $p$ and $y_p(\ell_a,\ell_b) \ne y_q(\ell_a,\ell_b)$,
the function $\psi_{p,q}(\ell_a,\ell_b,\ell_c)$ is defined. Since $\gamma_{\ell,p}$ and 
$\gamma_{\ell,q}$ are lines (cf.\ \lemref{bisector}), $\gamma_{\ell,q}$ lies below $\gamma_{\ell,p}$ before 
$w_d$ (i.e., for $d < w_d$) if and only if its slope is greater than that of $\gamma_{\ell,p}$, i.e.,
$-2y_q(\ell_a,\ell_b) > -y_p(\ell_a,\ell_b)$, or $y_q(\ell_a,\ell_b) - y_p(\ell_a,\ell_b) < 0$. 
Hence, by the definition of $\psi_{p,q}$ (cf.\ \eqref{surfaces}), $w_d > \ell_d$ (resp., $w_d < \ell_d$) 
implies that $w_d = \psi_{p,q}^-(\ell_a,\ell_b,\ell_c)$ (resp., $w_d=\psi_{p,q}^+(\ell_a,\ell_b,\ell_c)$).
This completes the proof of the lemma.
\end{proof}

We are now ready to describe the data structure, which is based on Lemma~\ref{lem:nearer-point}.

\paragraph{Overall data structure.}
We construct a two-level data structure. The top level is a 
partition tree $\T$ constructed on $P$. Each node $v$ of $\T$ is 
associated with a canonical subset $P_v \subset P$.
If $|P_v| \le n_0$, for some sufficiently large constant $n_0$, then $v$ is a leaf and we store $P_v$ at $v$.
Otherwise, we fix a sufficiently large constant parameter $r>0$ and choose a random subset
$R_v \subseteq P_v$ of $\alpha r\log r$ points, for a suitable absolute constant $\alpha>0$, independent of $r$.
We construct the four-dimensional Voronoi diagram $\M(R_v)$ induced by $R_v$ in $\Lspace$. 
For every point $p\in R_v$, we construct $\M_p^-, \M_p^+$, as defined above, but with respect 
to $\M(R_v)$. Let $\Xi_p^-$ (resp., $\Xi_p^+$) be the vertical decomposition of $\M_p^-$ 
(resp., $\M_p^+$). Let $\Xi_v$ be the set of cells in all these $2|R_v|$ vertical 
decompositions. By \lemref{union}, $|\Xi_v|= O^*(r\cdot r^3) = O^*(r^4)$ and
$\Xi_v$ can be constructed in overall $O^*(r^4)$ expected time. 

We define a \emph{conflict list} $P_\cell$ for every $\cell\in\Xi$, as follows.
For each point $p\in R_v$ and each cell $\cell$ of $\Xi_p^-$ (resp., $\Xi_p^+$), we define 
$P_\cell \subset P_v$ to be the subset of points $q\in P_v$ for which the surface 
$d = \psi_{p,q}^+(a,b,c)$ (resp., $d= \psi_{p,q}^-(a,b,c)$) crosses $\cell$. 
With a suitable choice of $\alpha$,
the size of each conflict list is at most $n/r$, with high probability, because, by construction, 
for a cell $\cell$ of $\M_p^-$ (resp., $\M_p^+$), none of the surfaces $d=\psi_{p,u}^+(a,b,c)$ 
(resp., $d=\psi_{p,u}^-(a,b,c)$), for $u\in R_v\setminus \{p\}$, intersect $\cell$;
see~\cite{HW87}. For each cell $\cell\in\Xi_v$, we create a child $w$ of $v$ and 
recursively build at $w$ the data structure on $P_\cell$. 

In addition, we also store a second-level data structure at $v$:
For this, define $\beta_{p,q}$ to be the one-dimensional locus 
\[
\beta_{p,q} = \{ (a,b) \mid y_p(a,b) = y_q(a,b) \} .
\]
(This is the case, ignored so far, of a vanishing denominator.)
For each point $p \in R_v$, we construct the two-dimensional arrangement $\A_p$ of the set of curves 
$\B_p = \{\beta_{p,q} \mid q \in P_v\setminus R_v\}$,
and preprocess it for point-location queries. 
For each edge $e\in\A_p$, if $e$ is a portion of the curve 
$\beta_{p,q}$, we set $P_e=\{q\}$. For a vertex $\chi=(\chi_a,\chi_b)$ of $\A_p$,
let $P_\chi \subseteq P_v\setminus R_v$ be the subset of points whose curves are incident 
on $\chi$. By definition, all points in $P_\chi$ have the same $y$-coordinate, namely,
$y_p(\chi_a,\chi_b)$. We store the points $q$ of $P_\chi$ in a 
sorted list ordered by $x_q(\chi_a,\chi_b)$ (the $x$-coordinates of their projections on $\pi(\chi_a,\chi_b)$),
so that, for a query line $\ell$ of the form 
$\ell^*=(\chi_a,\chi_b,\ell_c,\ell_d)$, 
the point in $P_\chi$ nearest to $\ell$ can be computed in $O(\log n)$ time. 
The total size of this second-level data structure, over all points of $R_v$, is 
$|R_v| \cdot O(|P_v|^2)= O(|P_v|^2)$ (since $r$ is a constant).

There are $O^*(r^4)$ subproblems at each node $v$ in the top tree $\T$, each of size at most $n_v/r$. 
A standard analysis of multi-level data structures shows that the size of 
the  overall structure, including the space needed to store the 
second-level structures, is $O^*(n^4)$, and that it can be constructed in $O^*(n^4)$ expected time.
(The quadratic overhead cost of the second-level structures is subsumed in this bound.)

\paragraph{Query procedure.}
Let $\ell$ be a query line with $\ell^*=(\ell_a,\ell_b,\ell_c,\ell_d)$.
We visit $\T$ in a top-down manner. Suppose we are at a node $v \in \T$. 
If $v$ is a leaf, we compute and return $\NN(\ell,P_v)$ in a brute-force manner, in $O^*(1)$ time. 
So assume that $v$ is an internal node.  We compute the nearest neighbor of $\ell$ in $R_v$, 
which we call $p$.  Next, we compute the cells $\cell^-, \cell^+$ of 
$\M_p^-, \M_p^+$, respectively, that contain $\ell^*$. All this is
done in brute force and takes $O^*(1)$ time since $r$ is a constant. 
We recursively search in the subtrees rooted at the cells $\cell^-, \cell^+$,
and compute the nearest point $p'_v$ among $p$ and the nearest points 
returned by these two recursive subproblems. 

In addition, we locate the point $(\ell_a, \ell_b)$ in the arrangement $\A_p$ to find, 
in $O(\log n)$ time, the nearest neighbor $p''_v$ of $\ell$ among the points of $P_v\setminus R_v$ 
whose projections on $\pi(\ell_a,\ell_b)$ have the $y$-coordinate $y_p(\ell_a,\ell_b)$, if such 
points exist. If the point $(\ell_a,\ell_b)$ lies inside a face of $\A_p$ then no such
point exists, if it lies on an edge of $\A_p$ then there is exactly one such point, and 
if it lies at a vertex of $\A_p$ there can be many such points and we compute the nearest one
by performing a binary search on the sorted list stored at that vertex, observing that the
distances along the sequence form a bimodal sequence with a unique minimum (or, in degenerate
situations, two consecutive minima, in which case $p_v''$ is either of these points).
We return the point among $p$, $p'_v$, $p''_v$ that is nearest to $\ell$.

The correctness of the algorithm is a consequence of the following observation:
Suppose $P_v$ contains a point $q$ that is nearer to $\ell$ than $p$, the point of 
$R_v$ closest to $\ell$. If $y_q(\ell_a,\ell_b)=y_p(\ell_a,\ell_b)$, then $q=p''_v$. 
Otherwise, by Lemma~\ref{lem:nearer-point}, $f_q$ intersects the curve $\gamma_{\ell,p}$ 
at a unique point $w=(\ell_a,\ell_b,\ell_c,w_d)$, where $w_d=\psi_{p,q}(\ell_a,\ell_b,\ell_c)$. 
Let $[w^-,w^+] = \M_p(R_v)\cap \rho_\ell$, where $\rho_\ell$ is the line parallel to the 
$d$-axis passing through $\ell^*$. Furthermore, if $w_d\ge \ell_d$, then $w_d\ge \ell_d>w^-$, 
and therefore, by \lemref{nearer-point}, $\psi_{p,q}^+$ intersects $\cell^+$ and thus $q$ 
belongs to the conflict list $P_{\cell^+}$. Similarly, if $w_d<\ell_d$, then $q$ belongs to the 
conflict list $P_{\cell^-}$, for the similarly defined cell $\cell^-$. 
In summary, if $q$ is closer to $\ell$ than $p$ (and $y_q(\ell_a,\ell_b) \ne y_p(\ell_a,\ell_b)$)
then $q$ lies in the conflict lists of one of
$\cell^-, \cell^+$. Hence, one of the two recursive subproblems returns $q$.

At each node $v$, we spend $O(\log n)$ time on searching in the secondary data structure and 
recurse in two subproblems, each of size at most $|P_v|/r$ (and $r$ can be chosen to be a 
sufficiently large constant). Hence, the total query time is $O^*(1)$ (this bound is not 
polylogarithmic, though). We thus conclude:

\begin{theorem} \label{thm:line-pt-log-qtime}
Let $P$ be a set of $n$ points in $\reals^3$.
$P$ can be preprocessed, in $O^*(n^4)$ expected time, into a data structure of size $O^*(n^4)$,
so that, for any query line $\ell\in\reals^3$, 
$\NN(\ell,P)$ can be computed in $O^*(1)$ time. 
\end{theorem}

Following a standard space/query-time trade-off techniques (see, e.g.,~Appendix~A.2 in~\cite{AAEKS}) 
and combining this data structure with the linear-size data structure of 
Sections~\ref{subsec:pt-linear}--\ref{subsec:test-construct}, we obtain the following result:

\begin{theorem} \label{thm:n4}
Let $P$ be a set of $n$ points in $\reals^3$ and $s\in[n,n^4]$ a so-called storage parameter. 
$P$ can be preprocessed, in $O^*(s)$ expected time, into a data structure of size $O^*(s)$, so that, 
for any query line $\ell\in\reals^3$, $\NN(\ell,P)$ can be computed in $O^*((n/s^{1/4})^{2/3})$ time. 
\end{theorem}
For sanity check, observe that this bound coincides with the above bounds at the two extreme
cases $s=n$ and $s=n^4$.

\subsection{Segment or triangle NN queries: Fast query time}
\label{subsec:TrNN}
We can extend this data structure 
to the case where the query objects are segments or triangles in $\reals^3$.

First, consider a query segment $e=ab$ in $\reals^3$. Let $\ell_e$ be the line containing $e$,  
let $W_e$ be the slab bounded by the two planes passing through the 
endpoints of $e$ and orthogonal to $e$. 
For a point $p_i\in P$, if $p_i \not\in W_e$ then $\dist(p_i,e)= \min \{\|p_i-a\|, \|p_i-b\|\}$, 
otherwise $\dist(p_i,e) = \dist(p_i,\ell_e)$. 

We therefore build two data structures. In the first
structure, we preprocess $P$ into an $O^*(n^2)$-size data structure
so that for a query point $q\in\reals^3$, $\NN(q,P)$ can be computed in $O(\log n)$ time~\cite{Ag:rs}. 
The second structure is a multi-level partition tree $T$ on $P$, 
built on the one just described. Each node $v$ of $T$ 
stores a canonical subset $P_v \subseteq P$. The first two levels 
build halfspace range-searching data structures so that, for a query segment $e\in\reals^3$, we 
can extract the subset $P_e =  P \cap W_e$ as a union of a few canonical subsets (see below for 
precise bounds). At each node $w$ of the second level, we build the data structure described in 
\subsecref{Fast_NN1} on the set $P_w$, so that for a query line $\ell_e$, $\NN(\ell_e, P_w)$ can be computed.

Let $e$ be a query segment.
By querying the first data structure with the endpoints of $e$, we compute 
the point $p_i \in P$ that is nearest to an endpoint of $e$.
By querying the second data structure with $e$, we compute $\NN(\ell_e, P\cap W_e)$.
The closer of the these two points is $\NN(e,P)$. 

The total size of the data structure is $O^*(n^4)$ and the query time remains $O^*(1)$.
(The cost of an NN query with a point amid points is subsumed by the cost of the other structure.)

Next, consider a query triangle $\triangle$ in $\reals^3$. Let $\pi_\triangle$ be the plane containing 
$\triangle$, let $\ell_\triangle$ be the line passing through the origin and normal to $\pi_\triangle$, 
and let $W_\triangle= \triangle\oplus \ell_\triangle$ 
be the unbounded prism spanned by $\triangle$ and orthogonal to it. 
If a point $p_i\not\in W_\triangle$, then the point of $\triangle$ closest to $p_i$ lies on 
one of its edges or at one of its vertices, and if
$p_i\in W_\triangle$ then $\dist(p_i,\triangle) = \dist(p_i,\pi_\triangle)$. 

We first construct the data structure just described for computing the NN of a query segment in $P$. 
Next, we preprocess $P$ into a simplex range-emptiness data structure of size $O^*(n^3)$ with
logarithmic query time. For a query triangle $\triangle$, it computes $\NN(\pi_\triangle, P\cap W_\triangle)$, 
as follows. For a real value $t>0$, let $W_{\triangle,t}\subset W_\triangle$ be the set of points 
within distance $t$ from $\pi_\triangle$; $W_{\triangle,t}$ is a bounded prism with its bases being 
translates of $\triangle$. Using the above data structure, we can determine, in $O(\log n)$ time, whether
$P\cap W_{\triangle,t}\ne \emptyset$, and return such a point if the answer is {\sc Yes}. By combining this 
data structure with parametric search, we can compute $\NN(\triangle,P\cap W_\triangle)$ in $O^*(1)$ time.

For a query triangle $\triangle$, we first compute the point of $P$ that is closest to its edges 
and vertices, using the first data structure, and then compute $\NN(\triangle, P\cap W_\triangle)$ 
using the second data structure.

Putting everything together, we obtain the following.
\begin{theorem} \label{thm:tr-pt}
Let $P$ be a set of $n$ points in $\reals^3$. $P$ can be preprocessed,
in $O^*(n^4)$ expected time, into a data structure of size $O^*(n^4)$,
so that, for any query triangle (or segment) $\triangle\in\reals^3$, $\NN(\triangle, P)$ 
can be computed in $O^*(1)$ time. 
\end{theorem}

We can again obtain a space/query-time trade-off. Recall that the NN of a query 
segment can be computed in $O^*(\sqrt{n})$ time using $O^*(n)$ space, but 
computing the NN of a query triangle takes $O^*(n^{2/3})$ time using $O^*(n)$ space. 
Hence, we obtain the following:

\begin{theorem} \label{thm:tr-pt-tradeoff}
Let $P$ be a set of $n$ points in $\reals^3$ and $s\in[n,n^4]$ a storage parameter. 
$P$ can be preprocessed, in $O^*(s)$ expected time, into a data structure of size $O^*(s)$,
so that, for any query segment $e$ in $\reals^3$, $\NN(e, P)$ can be computed in 
$O^*((n/s^{1/4})^{2/3})$ time. For a query triangle, the query time is $O((n/s^{1/4})^{8/9})$.
\end{theorem}

\section{ Point NN Searching amid Lines, Segments and Triangles in $\reals^3$}
\label{sec:plnn}

In this section, we study the problem of nearest-neighbor (NN) searching 
where the input is a set of $n$ lines, segments or triangles in $\reals^3$ and the queries 
are points in $\reals^3$. We wish to process this 
set into a data structure that for any given query point $q$, returns the input object that
is nearest to $q$. Recall that, for a closed geometric object $\gamma$ (such as a line, segment, or
triangle) and for a point $q\in\reals^3$, we define the Euclidean distance between them as
$\dist(q,\gamma) = \min_{x\in\gamma} \|q-x\|$. 
Given a set $\XX$ of geometric objects of these kinds in $\reals^3$ and a point $q\in\reals^3$, we define 
$\NN(q,\XX) = \arg\min_{\gamma\in\XX} \dist(q,\gamma)$. Our goal is thus to preprocess $\XX$ into 
a data structure so that, for a query point $q\in\reals^3$, $\NN(q,\XX)$ can be computed quickly.

\subsection{The case of lines}

We begin by considering the (simplest) case where the input objects are lines.
Let $\L$ be a set of $n$ lines in $\reals^3$. We wish to preprocess $\L$ into a data structure so that 
the line of $\L$ nearest to a query point $q \in \reals^3$ can be reported quickly.
For each line $\ell\in \L$, define a trivariate function $f_\ell: \reals^3 \rightarrow \reals$, 
so that $f_\ell(q)$ is the distance of a point $q\in\reals^3$ from $\ell$. 
Using \lemref{union}, 
we can preprocess $\L$, in $O^*(n^3)$ expected time, into a data structure of size $O^*(n^3)$, so that an 
NN query with a point $q\in\reals^3$ can be answered in $O(\log n)$ time. In the rest of the section, 
we thus focus on the near-linear size regime and show that a data structure of 
near-linear size can be constructed that answers an NN query in $O^*(n^{2/3})$ time. We will then 
combine both structures, as we did earlier, to obtain a trade-off between storage and query time.

\paragraph{Overall data structure.} 
Using the parametric-search technique~\cite{Ag:rs,AM93}, an NN query can be reduced to answering
$O^*(1)$ \emph{ball-intersection-detection} queries on $\L$. That is, we wish to preprocess $\L$ into a 
data structure that can efficiently determine whether a query ball $\sigma$ intersects any line of $\L$. 
We thus describe a data structure for answering ball-intersection-detection queries.
We rely on the following technical property, originally proved  by Mohaban and Sharir~\cite{MS}.
For a line $\ell$ in $\reals^3$, let $V_\ell$ be the vertical plane that contains $\ell$, and let 
$H_\ell$ be the plane orthogonal to $V_\ell$ that contains $\ell$. Let $\sigma = \sigma_p$ be a ball
centered at a point $p$. We say that $\ell$ is \emph{lower} (resp., \emph{higher})
than $\sigma_p$ if $p$ lies above (resp., below) $H_\ell$.
Let $\ell_\sigma^-$ (resp., $\ell_\sigma^+$) be the line parallel to $\ell$ contained in $V_\ell$ 
that is tangent to $\sigma$ from \emph{below} (resp., \emph{above}); see Figure~\ref{fig:point_line_NN}. 
For an object $\xi$ in $\reals^3$, let $\xi^\downarrow$ denote its $xy$-projection. 
For a subset $X \subseteq \L$, let $X^\downarrow$ denote the set of the $xy$-projections 
of the lines in $X$.

\begin{figure}[htb]
  \begin{center}
    \includegraphics[scale=0.7]{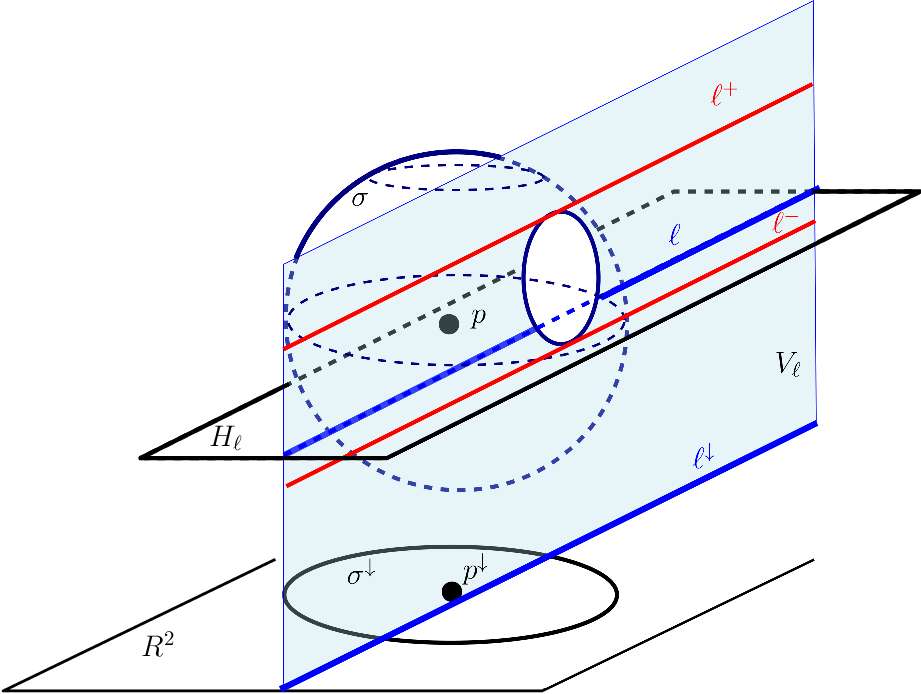}
  \end{center}
  \caption{Illustration of condition (ii-) of Lemma~\ref{lem:lower}. Here $\ell$ is lower than $\sigma = \sigma_p$.
    The small ellipse is the cross-section $\sigma \cap V_\ell$ (which is the drawing of a circle in ${\reals}^3$).
  }
      \label{fig:point_line_NN}
\end{figure}

\begin{lemma}[\cite{MS}]
  \label{lem:lower}
  Let $\ell$ be a line in $\reals^3$, and let $\sigma$ be a ball in $\reals^3$. Then $\ell$ 
  intersects $\sigma$ if and only if
  \begin{itemize}
	  \item[(i)\hspace{1mm}]
    the $xy$-projections $\ell^\downarrow$ and $\sigma^\downarrow$ 
    intersect; and one of the following two conditions holds.
  \item[(ii-)]
    $\ell$ is lower than $\sigma$ and $\ell$ lies above $\ell_\sigma^-$, or
  \item[(ii+)]
    $\ell$ is higher than $\sigma$ and $\ell$ lies below $\ell_\sigma^+$.
  \end{itemize}
\end{lemma}

Condition (i) is equivalent to testing whether a query disk in the plane intersects one of $n$
given lines in the plane.
Since a line in $\reals^2$ requires two parameters, the problem of extracting all lines of $\L$ whose $xy$-projections 
intersect $\sigma^\downarrow$ 
can be formulated as a two-dimensional semi-algebraic 
range query---the inequality that we need to test just involves the absolute value 
of a linear expression. By definition, testing whether a line $\ell$ in $\reals^3$ lies 
lower (or higher) than a ball amounts to testing whether the center of the ball lies above 
(or below) the respective plane $H_\ell$, so extracting all lines that lie lower (or higher) than 
$\sigma_p$ can be formulated as a $3$-dimensional halfspace-stabbing query, where the input halfspaces
are bounded from above or below by the planes $H_\ell$, and the query is the center $p$ of the ball. 

Finally, we describe below a linear-size data structure to preprocess a set $\XX$ of $m$ lines in $\reals^3$ 
so that, for a query ball $\sigma$ such that $\sigma^\downarrow$ intersects all lines of $\XX^\downarrow$ and 
all lines of $\XX$ lie lower (or all lie higher) than $\sigma$, we can detect in 
$O^*(m^{2/3})$ time whether $\sigma$ intersects any line of $\XX$ (see \lemref{lower-query} below).

With these pieces at our disposal, we construct a $3$-level partition tree~\cite{Ag:rs,AAEKS} on $\L$, 
each of whose nodes $v$ stores a ``canonical'' subset $\L_v \subseteq \L$ (see, e.g., Theorem~A.1 in 
the appendix of \cite{AAEKS}). The first level builds a two-dimensional semi-algebraic range-searching 
data structure that extracts the subset $\L_\sigma$ of lines of $\L$ whose $xy$-projections cross 
$\sigma^\downarrow$, where $\sigma = \sigma_p$ is the ball with the given radius centered at the query point $p$,
as the union of a few canonical subsets. For each canonical subset $\L_v$ of 
the first level, the second level builds a $3$-dimensional halfspace range-searching data structure, 
as mentioned above, that extracts all lines of $\L_v$ lying lower (or higher) than $\sigma$.
Finally, for each canonical subset $\L_w$ of the second-level tree, 
we build the data structures, described below (cf.\ \lemref{lower-query}), for intersection 
detection with respect to a query ball $\sigma$ that is higher (or lower) 
than all of the lines of $\L_w$, and $\sigma^\downarrow$ 
intersects all lines of $\L_w^\downarrow$ (in the $xy$-plane).
The overall size of the data structure\footnote{%
  A straightforward application of the multi-level data-structure framework leads to a 
  data structure of size $O^*(n)$. But the size can be improved to $O(n)$ while keeping 
  the query time $O^*(n^{2/3})$, by constructing secondary structures only at some of 
  the nodes~\cite{Ag:rs}.}
is $O(n)$. 

The query time at each level of the partition tree is $O^*(n^{2/3})$ (it is actually smaller for 
the first level, as the corresponding subproblem is two-dimensional), 
the properties of multi-level partition trees (see, e.g., Theorem~A.1 in the 
appendix of \cite{AAEKS}) imply that the overall query time is also $O^*(n^{2/3})$. 

Using standard space/query-time trade-off techniques~\cite{AAEKS}, we can combine the linear-size 
data structure with the $O^*(n^3)$-size data structure: For any parameter $s\in[n,n^3]$, we can 
construct a data structure of size $O^*(s)$, in $O^*(s)$ expected time, that answers a query in 
time $O^*(n/s^{1/3})$. Hence, modulo the missing details of the third level, elaborated below, 
we obtain the following lemma.
\begin{lemma} \label{lem:ballint}
  Let $\L$ be a set of $n$ lines in $\reals^3$, and let $s\in [n,n^3]$ be a parameter.
  Then $\L$ can be preprocessed, in $O^*(s)$ expected time, into a data structure of size $O^*(s)$,
  so that an intersection-detection query for a ball can be answered in $O^*(n/s^{1/3})$ time. 
\end{lemma}

\paragraph{Ball-intersection-detection queries for lines lower than the ball.}

Let $\L$ be a set of $n$ lines in $\reals^3$. 
Let $\Sigma^+$ (resp., $\Sigma^-$) be the set of balls $\sigma$ in $\reals^3$ that satisfy the 
conditions (C1) and (C$2^+$) (resp., (C1) and (C$2^-$)).
\begin{itemize}
	\item[{\bf (C1)\hspace{1mm}}] $\sigma^\downarrow$ intersects all lines of $\L^\downarrow$, and
  \item[{\bf (C$2^+$)}] all lines of $\L$ are lower than $\sigma$, or
  \item[{\bf (C$2^-$)}] all lines of $\L$ are higher than $\sigma$.
\end{itemize}
Consider the case of $\Sigma^+$. We wish to preprocess $\L$ into a linear-size data structure 
that, for a query ball $\sigma\in \Sigma^+$, can 
determine, in $O^*(n^{2/3})$ time, whether $\sigma$ intersects any line of $\L$.
If the answer is {\sc Yes}, return one such line. 
(A similar data structure can be constructed for answering queries with balls in $\Sigma^-$,
by reversing the $z$-direction.)

We work in the $4$-dimensional parametric space $\Lspace$ of lines defined in \subsecref{Fast_NN1}, i.e., 
a line $\ell$ is represented by the point $\ell^* = (a,b,c,d)$, for which the equations 
defining $\ell$ are $y=ax+c$, $z=bx+d$. Thus, as before, ignoring the true projective nature 
of $\Lspace$, we identify it
with $\reals^4$. Put $\L^* = \{ \ell^* \mid \ell \in \L\}$. 
A ball $\sigma$ is associated with a surface (patch) $\gamma_\sigma \subset \Lspace$, 
which is the locus of points $\ell^*$ such that the corresponding line $\ell$ is tangent to $\sigma$ from below. 
Let $\gamma_\sigma^+$ be the set of the points lying on or above $\gamma_\sigma$ in the $d$-direction; 
$\gamma_\sigma^+$ is a semi-algebraic set of constant complexity.\footnote{%
  For the symmetric case where all lines of $\L$ are higher than a query ball $\sigma$, 
  we define $\gamma_\sigma$ to be the locus of points $\ell^*$ such that $\ell$ is tangent 
  to $\sigma$ from above, and we set $\gamma_\sigma^-$ to be the set of points lying on or below $\gamma_\sigma$.}
It is easily seen that a line $\ell$ in $\L$ intersects $\sigma$ if and only if $\ell^*$ 
lies in $\gamma_\sigma^+$. Let $\Gamma$ be the collection of all sets 
$\gamma_\sigma^+$ such that $\sigma\in\Sigma^+$.
Thus the ball-intersection query for a ball $\sigma$ in our setting reduces to 
semi-algebraic range-emptiness  query in $\L^*$ with $\gamma_\sigma^+\in\Gamma$.

Similar to the discussion in \subsecref{RS}, it suffices to describe an algorithm for 
constructing a test set. We describe such an algorithm that constructs a test set 
of size $O(r^4)$ with $\zeta(m)=O^*(m^3)$ and $\delta=1$ (see \lemref{lower-query} below).  
Plugging these bounds into \lemref{test-set}, we obtain (modulo the still missing details
of the test-set construction) the following:

\begin{lemma} \label{lem:lower-query}
Let $\L$ be a set of $n$ lines in $\reals^3$. $\L$ can be preprocessed,
in $O^*(n)$ expected time, into a data structure of size $O(n)$, so that for any ball $\sigma$ 
in $\Sigma^+$ or in $\Sigma^-$, the intersection-detection query 
with $\sigma$ in $\L$ can be answered in $O^*(n^{2/3})$ time. 
\end{lemma}

\paragraph{Test-set construction.}
We construct the test set following the same approach as in~\subsecref{test-construct}. 
Let $\Sspace$ be the $4$-dimensional 
parametric space of balls in $\reals^3$, where a ball $\sigma$ of radius $r$  centered at a point $p$ 
is mapped to the point $\sigma^* = (p,r)\in\Sspace$. $\Sspace$ can thus be identified with 
(the positive halfspace $r>0$ of) $\reals^4$. A line $\ell$ in $\reals^3$ is mapped to a surface
$\omega_\ell$, consisting of all points $\sigma^* \in \Sspace$ that represent balls that touch 
$\ell$ from above. As is easily verified, these surfaces are monotone over the $xyz$-subspace, so that a point
$\sigma^*$ lies below the surface $\omega_\ell$ if and only if $\ell$ intersects $\sigma$,
assuming $\sigma\in\Sigma^+$ and $\ell$ is lower than $\sigma$. (Informally, this above/below location of $\sigma^*$ is the reason why we have to distinguish 
  between lines that pass below the ball and lines that pass above.)

Let $\Omega = \{ \omega_\ell \mid \ell \in \L\}$ denote the collection of these surfaces. We take a 
random subset $R\subseteq \Omega$ of $s = cr\log r$ surfaces, for some sufficiently large constant $r$, 
where $c>0$ is a suitable constant, independent of $r$, and construct the vertical decomposition 
$\VeD{\A}(R)$ of the arrangement $\A(R)$, of size $O^*(r^4)$, as 
stated in \lemref{union} (see also~\cite{Koltun-04a}). 
We choose a subset $\Xi$ of $\VeD{\A}(R)$, consisting of those cells that have at most 
$n/r$ surfaces of $\Omega$ passing \emph{fully above} them. By construction, these cells cover the 
highest $n/r$ levels of $\A(\Omega)$ (for else, as is easily checked,
there would exist a cell $\tau \notin \Xi$ that has strictly fewer than $n/r$ surfaces above it, 
contrary to the definition of $\Xi$),
and are contained in the at most $2n/r$ upper levels of $\A(\Omega)$.

Let $\tau$ be a cell of $\Xi$.
We now switch to the parametric line-space $\Lspace$, where each point $\sigma^*\in\tau$
becomes the surface $\gamma_\sigma$. Set $\Phi_\tau = \bigcup_{\sigma^*\in\tau} \gamma_\sigma^+$.
Set $\Phi = \{\Phi_\tau \mid \tau \in \Xi\}$ to be the desired test set. $\Phi$ is a family of 
$O^*(r^4)$ constant-complexity semi-algebraic sets in $\Lspace$. Following the same argument as 
in \subsecref{test-construct}, we can argue that each $\Phi_\tau$ is a semi-algebraic set of 
constant complexity and that it is $(2n/r)$-shallow with respect to $\L^*$. 
Since the cells in $\Xi$ cover 
$\A_{\le n/r}(\Omega)$, any 
$(n/r)$-shallow range $\gamma_\sigma^+$ must be associated with a cell $\tau$ (in the sense that
$\sigma^*\in\tau$), and thus $\gamma_\sigma^+ \subseteq \Phi_\tau$.
Therefore properties (P1) and (P3) of a test set hold. 
We prove property (P2).
Each set $\Phi_\tau$ is unbounded in the positive $d$-direction, therefore the complement of the union 
of a subset of $m$ sets is the region lying below their lower envelope. Hence, by \lemref{union}, the 
complement of the union of these regions (i.e., the portion of space below their lower envelope)
can be decomposed into $O^*(m^3)$ pseudo-prisms, and this construction takes 
$O^*(m^3)$ expected time. Hence, we obtain the following:
\begin{lemma} \label{lem:btest}
  For a parameter $r\ge 1$, a test set for the $(n/r)$-shallow ranges of $\Gamma$, of size $O^*(r^4)$, with 
  $\zeta(m)=O^*(m^3)$ and $\delta=1$, can be computed in $O^*(r^4)$ expected time.
\end{lemma}

We have thus established the missing ingredient for the proof of \lemref{lower-query}, and so 
this completes the description of the linear-size data structure for NN searching with query 
points amid lines. Combining this with the preceding analysis,
which culminates in Lemmas~\ref{lem:ballint} and \ref{lem:lower-query}, we obtain a data structure
of linear size, that answers an NN query with a point in $O^*(n^{2/3})$ time. 
By Lemma~\ref{lem:ballint} and the reduction from NN search to ball intersection detection, mentioned 
earlier in this section, we conclude the following main result.

\begin{theorem}
  \label{thm:NNPL}
  Let $\L$ be a set of $n$ lines in $\reals^3$, and let $s\in [n,n^3]$ be a storage parameter. 
  $\L$ can be preprocessed, in $O^*(s)$ expected time, into a data structure of size $O^*(s)$,
  so that for any query point $p\in\reals^3$, the line of $\L$ 
  nearest to $p$ can be computed in $O^*(n/s^{1/3})$ time. 
\end{theorem}

\subsection{Point NN queries amid segments or triangles}

Finally, we show that our data structure can be extended to answering NN queries with points when the input objects 
are segments or triangles. We first consider the case of segments, and focus on the regime of near-linear
storage.

\paragraph{The case of segments.}
Let $\E=\{e_1,\ldots, e_n\}$ be a set of $n$ segments in $\reals^3$. 
For a segment $e_i \in \E$, let $\ell_i$ be the line containing $e_i$ and $W_i$ the slab 
bounded by the two planes passing through the endpoints of $e_i$ and orthogonal to $e_i$.
We follow the same idea as in \subsecref{TrNN}.
We build two data structures: The first one preprocesses $P$, the set of endpoints of the segments 
in $\E$, into a data structure for computing $\NN(q, P)$, for any query point $q\in\reals^3$.
The second structure is a multi-level partition tree $T$ on $\E$, similar to the one just 
described, each of whose nodes $v$ stores a canonical subset $\E_v \subseteq \E$. The first two levels 
build halfspace range-searching data structures so that for a query point $q\in\reals^3$, we 
can extract the subset $\E_q$ of segments $e_i\in\E$ such that $q\in W_i$, as a union of a 
few canonical subsets (see below for precise bounds). At each node $w$ of the second level, let $\L_w$ be the 
set of lines containing the segments of the canonical subset $\E_w$. We construct an NN query 
data structure for $\L_w$, following Theorem~\ref{thm:NNPL}.

Let $q\in\reals^3$ be a query point. By querying the first data structure with $q$, we compute 
the segment $e_i$ of $\E$ that has an endpoint that is closest to $q$ among all endpoints.
Next, using the second data structure, we compute the segment $e_i$ that is nearest to $q$,
among those segments $e_j$ for which $q\in W_j$. The closer of the these two segments is $\NN(q,\E)$. 
The total size of the data structure is $O^*(n)$ and the query time remains $O^*(n^{2/3})$.
(The cost of an NN query with a point amid points is subsumed by the cost of the second structure.)

\paragraph{The case of triangles.}
Next, let $\T=\{\triangle_1,\ldots,\triangle_n\}$ 
be a set of $n$ triangles in $\reals^3$. Let $\E$ be the set of edges of the triangles in $\T$. 
For a triangle $\triangle_i$, let $\pi_i$ be the plane containing $\triangle_i$, let $\ell_i$ be the 
line passing through the origin and normal to $\pi_i$, and let $W_i= \triangle_i\oplus \ell_i$ 
be the prism spanned by $\triangle_i$ and orthogonal to it.

We again build two data structures. We first construct the NN data structure just described on the set $\E$. 
For a point $q\in\reals^3$, let 
\[
\T_q = \{ \triangle_i \in \T \mid q \in W_i\} .
\]
Next, we build a following four-level partition tree. 
For a query point $q\in\reals^3$, the first three levels 
extract the subset $\T_q$ of triangles, as a union of a few 
canonical subsets. For each canonical subset $\T_z$, let 
$P_z$ be the set of points dual to the planes containing the triangles of $\T_z$. 
At the fourth level of the structure, we preprocess $P_z$ into a linear-size 
semi-algebraic range-searching data structure~\cite{AMS,MP}, where the query ranges are the 
following dual regions of balls in $\reals^3$. That is, for a pair $q\in\reals^3$ and $\rho\ge 0$, let 
$\gamma_{q,\rho}$ be the set of points dual to the planes that are within distance $\rho$ from 
$q$ (namely, intersect the ball of radius $\rho$ centered at $q$).
For a node $z$ at the third-level tree such that $\T_z\subseteq \T_q$, for a query point $q\in\reals^3$, 
combining the fourth-level data structure built for $\T_z$ with parametric search, we can compute 
$\NN(q,\T_z)$ in time $O^*(|\T_z|^{2/3})$. By querying both data structures with $q$, we obtain 
$\NN(q,\T)$ as the nearest of the two output triangles.

\paragraph{The fast-query regime.}
For the fast query regime, we can again construct constant-complexity semi-algebraic trivate function 
$f_\triangle$ for each $\triangle\in\T$ that measures the distance of a point $q\in\reals^3$ from $p$.
We can preprocess these functions, 
in $O^*(n^3)$ expected time, into a data structure of $O^*(n^3)$ size, so that 
$\NN(q,T)$ for a query point $q\in\reals^3$ can be computed in $O(\log n)$ time.
By combining this data structure with the $O^*(n)$-size data structure just presented,
we can obtain space/query-time trade-off, with the same performance bound as before. 
Putting everything together, we obtain the following:

\begin{theorem}
  \label{thm:NNPT}
Let $\T$ be a set of $n$ triangles (or segments) in $\reals^3$, and let $s\in [n,n^3]$ be a parameter. 
$\T$ can be preprocessed, in $O^*(s)$ expected time, into a data structure of size $O^*(s)$,
so that for any query point $q\in\reals^3$, $\NN(q,\T)$ can be computed in $O^*(n/s^{1/3})$ expected time. 
\end{theorem}

\section{Line NN Queries amid Lines}
\label{sec:Fast_NN2}

Next, we consider the case when the input is a set $L$ of $n$ lines in $\reals^3$, and 
show that the machinery in \subsecref{Fast_NN1} can be adapted, with some twists, to compute 
$\NN(\ell,L)$, for a query line $\ell$. Similar to the other setups discussed above,
we note that, using a $4$-dimensional semi-algebraic range searching data structure, 
$L$ can be preprocessed, in $O^*(n)$ expected time, into a data structure of $O(n)$ space,
so that, for a query line $\ell$, $\NN(\ell,L)$ can be computed in $O^*(n^{3/4})$ time~\cite{AMS,AAEKS,MP}. 
We therefore focus, again, on a data structure with $O^*(1)$ query time, with the goal of reducing its storage.

For simplicity and without loss of generality, we assume that no line of $L$ passes through the origin and 
that no two lines of $L$ lie on a plane passing through the origin.\footnote{The second requirement is vacuous if
we assume general position of the lines, so that no two lines are coplanar, but it becomes significant
when we allow such degeneracies in the input.}
We can ensure these properties by perturbing the origin if necessary.

\paragraph{Distance function, Voronoi diagram.}
We use the same representation $\ell^*=(a,b,c,d)$ for a line $\ell$, using the orthogonal plane $\pi_\ell$ 
through the origin, as before.  
For an input line $\lambda\in L$, we define several objects:
\begin{enumerate}[(i)]
\item We define a function $f_\lambda: \reals^4 \rightarrow \reals_{\ge 0}$ such that for a line $\ell$ with $\ell^*=(a,b,c,d)$, $f_\lambda(a,b,c,d)=\dist(\ell,\lambda)$.
\item For a line $\ell$ with $\ell^*=(a,b,c,d)$,
  let $\lambda^\downarrow:= \lambda^\downarrow (a,b)$ denote the projection of $\lambda$ onto $\pi_\ell$. 
  If $\lambda^\downarrow(a,b)$ does not pass through the origin,
  the equation of $\lambda^\downarrow$, in the canonical coordinate frame within $\pi_\ell$, is of the form
  \[
  \xi_\lambda(a,b) x + \eta_\lambda(a,b) y + \zeta_\lambda(a,b) = 0 ,
  \]
  where $(\xi_\lambda(a,b), \eta_\lambda(a,b))$ is the unit vector within $\pi_\ell$ 
  orthogonal to $\lambda^\downarrow$
  and pointing away from the origin (so the origin lies in the negative halfplane bounded by 
  $\lambda^\downarrow$). 
  We omit the easily derivable explicit expressions for
  $\xi_\lambda(a,b)$, $\eta_\lambda(a,b)$ and $\zeta_\lambda(a,b)$. 
  See Figure~\ref{fig:line-line}.
    As we vary $(a,b)$,
    the unit vector $(\xi_\lambda(a,b),\eta_\lambda(a,b))$ changes continuously except when $\lambda^\downarrow(a,b)$ 
    passes through the origin.
\item Let $h_\lambda$ be the plane spanned by $\lambda$ and the origin.
\item 
  We define two \emph{critical} surfaces: 
  let $\ph_{\lambda,1}$ be the locus of all points $\ell^*=(a,b,c,d)\in\Lspace$ such that $\ell$ is parallel to $h_\lambda$ (namely $\lambda^\downarrow(a,b)$ passes through the origin in $\pi_\ell$), and 
  let $\ph_{\lambda,2}$ be the locus of all points $\ell^*\in\Lspace$ such that 
  $\ell$ touches $\lambda$. Since the planes $h_\lambda$ are distinct for all lines in $L$, the intersection surface 
  $\ph_{\lambda,1}\cap\ph_{\theta,1}$, for two distinct lines $\lambda,\theta\in L$, corresponds to all lines that are parallel to the 
  intersection line $h_\lambda\cap h_\theta$ and is thus 
  a $2$-flat of the form $(a_0,b_0)\times\reals^2$ for some fixed 
  $(a_0,b_0)\in\reals^2$ (ignoring, as before, the fact that $\Lspace$ is projective).   
  For all points $\ell^*=(a,b,c,d)$ within a cell of $\Lspace\setminus (\ph_{\lambda,1}\cup\ph_{\lambda,2})$, the origin lies on the same side of the line $\lambda^\downarrow(a,b)$ in $\pi_\ell$,
  the point $(c,d)$ 
  lies in this halfplane as well (that is, the same halfplane of $\pi_\ell$ bounded by $\lambda^\downarrow(a,b)$), 
  and the unit vector $(\xi_\lambda(a,b), \eta_\lambda(a,b))$ varies continuously. These properties will be crucial in building the data structure, as described below.
\end{enumerate}

\begin{figure}[htb]
  \begin{center}
    \includegraphics[scale=0.60]{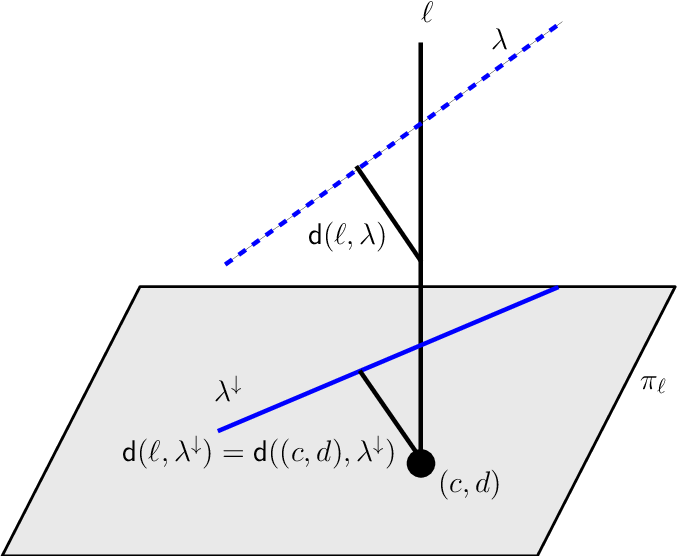}
    \end{center}
   \caption{The distance function for lines within $\pi_\ell$. The two segments that realize
         $\dist(\ell,\lambda)$ and $\dist(\ell,\lambda^\downarrow) = \dist((c,d),\lambda^\downarrow)$ 
         are parallel and equal, with the latter segment being the projection of the former onto $\pi_\ell$.}
	\label{fig:line-line}
\end{figure}

A crucial 
observation, which is easy to verify, is that
\[
f_\lambda(a,b,c,d) = \dist(\ell,\lambda^\downarrow) = \dist((c,d),\lambda^\downarrow) .
\]
Hence,
\begin{equation}
	\label{eq:l-dist}
	f_\lambda(a,b,c,d)  = \left| \xi_\lambda(a,b) c + \eta_\lambda(a,b) d + \zeta_\lambda(a,b) \right| .
\end{equation}
	Except for the absolute value, (\ref{eq:l-dist}) is linear in $c$ and $d$ and thus $f_\lambda$ can be treated as a linear function locally.  More precisely,
the critical surfaces $\ph_{\lambda,1}, \ph_{\lambda,2}$, defined above, induce a stratification of $\Lspace$ such that $f_\lambda$ can be treated as a linear function (in $c$ and $d$) within each cell of $\Lspace \setminus (\ph_{\lambda,1}\cup \ph_{\lambda,2})$.
This property will be crucial for our data structure.
The minimization diagram $\M$ of the functions $\{f_\lambda\mid \lambda\in L\}$ represents the 
Voronoi diagram induced by $L$ in $\Lspace$. 
As in \subsecref{Fast_NN1}, we will exploit the special structure of Voronoi 
regions in $\M$, resulting from the ``local linearity'' of the 
functions $f_\lambda$ in $c$ and $d$.

\begin{figure}[htb]
  \begin{center}
    \includegraphics[scale=0.7]{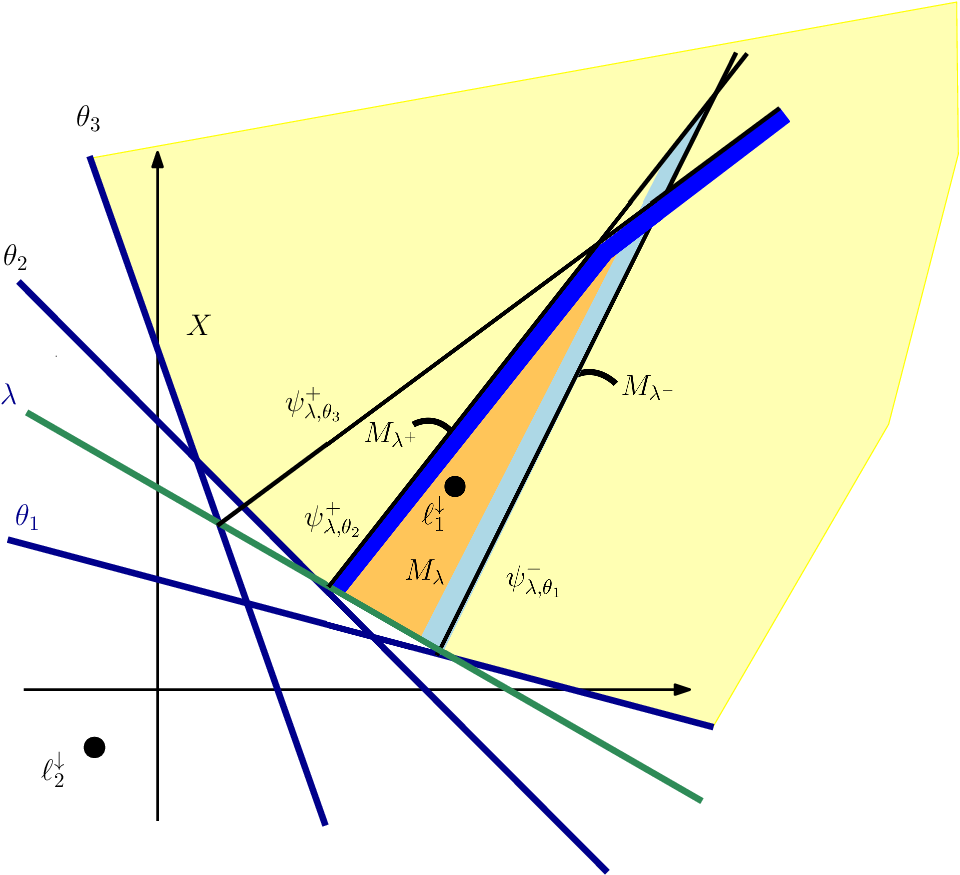}
	\caption{Voronoi cell of $\lambda$ in a fixed $(a_0,b_0)$-plane $\pi(a_0,b_0)$. Four lines $\lambda, \theta_1,\theta_2, \theta_3$. 
	  The lightly shaded region, $\Xspace$, corresponds to the lines $(a_0,b_0,c,d)$ that are positively oriented with respect to all four of 
	  them --- line $\ell_1$ (resp.,\ $\ell_2$) is positively (resp.,\ negatively) oriented with respect to them.
	  Linear functions $\psi_{\lambda,\theta_2}^+, \psi_{\lambda,\theta_3}^+, \psi_{\lambda,\theta_1}^-$,
	  lower envelope $M_{\lambda}^+$, the upper envelope $M_{\lambda^-}$, and the Voronoi cell $\M_\lambda$ (within $\Xspace$).
	}
	\label{fig:line-VD}
  \end{center}
\end{figure}

\paragraph{Structure of the Voronoi diagram.}
We call a line $\ell$, with $\ell^*=(a,b,c,d)$, \emph{positively (resp., negatively) oriented} with respect to 
$\lambda\in L$ if the point $(c,d)$ lies in the positive (resp., negative) halfplane of $\pi_\ell$ 
bounded by  $\lambda^\downarrow$, i.e., 
$\xi_\lambda(a,b) c + \eta_\lambda(a,b) d + \zeta_\lambda(a,b)$ is positive (resp., negative).

Let $\lambda, \theta \in L$ be a pair of input lines. 
We define a linear function $\psi_{\lambda,\theta}: \Lspace_0 \rightarrow \reals$,
where $\Lspace_0$ is the hyperplane $d=0$ (compare with the similar earlier expression in 
(\ref {eq:psipq})), assuming $\eta_\lambda(a,b) \neq \eta_{\theta} (a,b)$,
\begin{equation}
\label{eq:l-surface}
d = \psi_{\lambda,\theta} (a,b,c) := \frac{\zeta_{\theta} (a,b)-\zeta_\lambda (a,b)}{\eta_\lambda(a,b)-\eta_\theta(a,b)} + 
\frac{\xi_{\theta} (a,b)-\xi_\lambda (a,b)}{\eta_\lambda(a,b)-\eta_{\theta}(a,b)} \; c .
\end{equation}
The graph of $\psi_{\lambda,\theta}$ is the bisector of $\lambda$ and $\theta$ when restricted to the points $\ell^*=(a,b,c,d)$ such that $\ell$ is positively oriented with respect to both $\lambda$ and $\theta$. 

Note that the function $\psi_{\lambda,\theta} (a,b,c)$ is not defined when $\eta_\lambda(a,b) = \eta_{\theta} (a,b)$, 
which also implies that $\xi_\lambda(a,b) = \xi_{\theta} (a,b)$ (because $(\xi_\lambda(a,b),\eta_\lambda(a,b))$ 
is a unit vector and $\ell$ is positively oriented with respect to both $\lambda,\theta$), i.e., 
$\lambda^\downarrow (a,b)$ and ${\theta}^\downarrow(a,b)$ are parallel and lie on the same side of 
the origin in $\pi_\ell$. 
Let $\beta_{\lambda,\theta}$ denote the one-dimensional locus of all $(a,b)$ for which 
$\lambda^\downarrow(a,b)$ and ${\theta}^\downarrow (a,b)$ are parallel, namely,
\[
	\beta_{\lambda,\theta} = \{ (a,b) \mid \eta_\lambda(a,b) = \eta_{\theta}(a,b) \} .
\]

As in~\eqref{surfaces}, we replace $\psi_{\lambda,\theta}$ by two partially defined functions $\psi_{\lambda,\theta}^+$ 
and $\psi_{\lambda,\theta}^-$ depending on  whether $\eta_\lambda(a,b)-\eta_{\theta}(a,b)$ is positive or negative
(assuming it to be nonzero), namely,
\begin{align}
  \label{eq:l-surfaces}
  \begin{split}
	  \psi_{\lambda,\theta}^+(a,b,c) & = 
    \begin{cases}
      \psi_{\lambda,\theta}(a,b,c) & \eta_\lambda (a,b) - \eta_\theta(a,b) > 0 , \\
      +\infty & \text{otherwise.}
    \end{cases} \\
    \psi_{\lambda,\theta}^-(a,b,c) & = 
    \begin{cases}
      \psi_{\lambda,\theta}(a,b,c) & \eta_\lambda (a,b) - \eta_\theta(a,b) < 0, \\
      -\infty & \text{otherwise} .
    \end{cases} 
  \end{split}
\end{align}

Let $\Xspace \subset\Lspace$ and let $H\subset L$ be such that all lines in $H$ are positively oriented 
with respect to all lines $\ell$ with $\ell^*\in\Xspace$. For a line $\lambda\in H$, let 
\[
	\M_\lambda := \M_\lambda(\Xspace,H) = \{ \ell^*\in \Xspace \mid \dist(\ell,\lambda)\le \dist(\ell,\theta) \; \forall\theta\in H\}
\]
be the Voronoi region of $\lambda$ with respect to $H$ restricted to $\Xspace$.
Following the same argument as in \subsecref{Fast_NN1},
$\M_\lambda$ can be expressed as
\begin{equation} \label{eq:l-sandwich}
  \M_\lambda = \left\{ (a,b,c,d) \in \Xspace \mid \max_{\theta\in H\setminus\{\lambda\}} \psi_{\lambda,\theta}^-(a,b,c) \le d \le 
  \min_{\theta\in H\setminus\{\lambda\}} \psi_{\lambda,\theta}^+(a,b,c) \right\} .
\end{equation}
That is, $\M_\lambda$ is
the sandwich region between the 
upper envelope of $\{\psi_{\lambda,\theta}^- \mid \theta \in H\setminus\{\lambda\}\}$ and the 
lower envelope of $\{\psi_{\lambda,\theta}^+ \mid \theta\in H\setminus\{\lambda\}\}$, restricted within $\Xspace$.  Alternatively, $\M_\lambda=\M_\lambda^-\cap \M_\lambda^+\cap\Xspace$, 
where $\M_\lambda^-$ (resp., $\M_\lambda^+$) is the region lying below 
the lower envelope (resp., above the upper envelope) of the functions  $\psi_{\lambda,\theta}^+$  (resp., $\psi_{\lambda,\theta}^-$),
over $\theta\in H\setminus\{\lambda\}$, in the $d$-direction. Note that $\M_\lambda^+$ and $\M_\lambda^-$ are
monotone in the $d$-direction. 
See Figure~\ref{fig:line-VD} for an illustration. 

We can also state the analog of \lemref{nearer-point}, whose proof follows
in a similar manner. In preparation for the lemma, let 
$\ell^* = (\ell_a,\ell_b,\ell_c,\ell_d)\in \Lspace$, and 
let $\rho_\ell$ denote the line in the $d$-direction in $\Lspace$ passing through $\ell^*$. 
Let $\lambda\in L$ be a line such that $\ell$ is positively oriented with respect to $\lambda$. 
Let $f_{\lambda}^{(\ell)}(d)$ be the linear function over $\rho_\ell$ defined by
\[ f_{\lambda}^{(\ell)}(d) = 
	  \eta_\lambda(\ell_a,\ell_b) d +\xi_\lambda(\ell_a,\ell_b) \ell_c + \zeta_\lambda(\ell_a,\ell_b).
	 \]
If $\ell(d)$ is positively oriented with respect to $\lambda$, then $f_\lambda^{(\ell)}(d)$ is
the distance of the line $\ell(d) = (\ell_a,\ell_b,\ell_c,d)$ from $\lambda$.
Let $\gamma_{\lambda}^{(\ell)}$ be the graph of the function 
$f_{\lambda}^{(\ell)}$, which is clearly a line in the two-dimensional $(d,x_5)$-plane.
\begin{lemma}
  \label{lem:nearer-line}
  Let $\lambda,\theta \in L$ be two lines such that $\ell$ is 
  positively oriented with respect to both $\lambda$ and $\theta$.
  Assume  that $\theta$ is closer to $\ell$ than 
  $\lambda$ and (i.e., $f_{\theta}(\ell^*) < f_\lambda(\ell^*)$) and $(\ell_a,\ell_b) \not\in \beta_{\lambda,\theta}$.
  Then the line $\gamma_{\theta}^{(\ell)}$, the graph of the 
  function $f_\theta(\ell_a,\ell_b,\ell_c,d)$, intersects $\gamma_{\lambda}^{(\ell)}$ at a unique point
  $w=(\ell_a,\ell_b,\ell_c,w_d)$ where $w_d=\psi_{\lambda,\theta} (\ell_a,\ell_b,\ell_c)$. Furthermore,
  if $w_d > \ell_d$ (resp., $w_d < \ell_d$) then $w_d = \psi_{\lambda,\theta}^-(\ell_a,\ell_b,\ell_c)$
  (resp., $w_d = \psi_{\lambda,\theta}^+(\ell_a,\ell_b,\ell_c)$).
\end{lemma}

Refer to \figref{lem71} for the analogous situation when the input objects are points.
With this preparation, we are now ready to describe the data structure.

\paragraph{Overall data structure.}
Unlike the distance function in \subsecref{Fast_NN1}, the distance function $f_\lambda$ is only locally linear in $c$ and $d$,
so the data structure is slightly more involved than the one in \subsecref{Fast_NN1}. In particular, we first construct a hierarchical 
stratification of $\Lspace$, and then, for each cell $\cell$ in the stratification, we 
construct a data structure analogous to that in \subsecref{Fast_NN1} for a subset of lines whose 
critical surfaces do not intersect $\cell$.

We build the following three-level data structure. 
Let $\Phi=\{\ph_{\lambda,1}, \ph_{\lambda,2}\mid\lambda \in L\}$.
At the top level, we construct a partition tree $\T^{(1)}$ 
on $\Phi$, induced by the polynomial-partitioning based hierarchical decomposition of $\Lspace$
using the algorithm in~\cite{AAEZ}. 
Each node $v$ of $\T^{(1)}$ is associated with some region $\cell_v$ of $\Lspace$ and with a 
subset $L_v \subseteq L$ consisting of lines $\lambda$ such that
at least one of $\ph_{\lambda,1},\ph_{\lambda,2}$
crosses $\cell_v$. 
	The cells $\cell_v$ induce the desired hierarchical stratification of $\Lspace$.
For the subset $L_v^*$ of lines whose surfaces do not intersect $\cell_v$, we can compute 
$\NN(\ell, L_v^*)$, for any $\ell^*\in\cell_v$,
using the data structure described in \subsecref{Fast_NN1}. Roughly speaking, for a query line $\ell$, $\T^{(1)}$ provides a partition of $L$ into $O(\log n)$ canonical subsets 
so that $\ell$ is positively (or negatively) oriented with respect to all lines within each
canonical subset. A technical difficulty we have to overcome is that $\ell^*$ may lie on some of the critical 
surfaces in $\Phi$ and they may contain $\cell_v$, 
so we have to proceed more carefully,
as described below.

For a node $v\in\T^{(1)}$, if $\cell_v\subset\ph_{\lambda,2}$ for some $\lambda_v\in L_{p(v)}$ 
($\dist(\ell,\lambda)=0$ for any $\ell$ with $\ell^*\in\cell_v$), where $p(v)$ is the parent of $v$, we store $\lambda_v$ at $v$ (if there are multiple such lines, we store only one of them), make $v$ a leaf, 
and discard the subtree of $\T^{(1)}$ rooted at $v$. So assume that $\cell_v$ is not contained in any such surface.

\begin{figure}[htb]
  \begin{center}
    \includegraphics[scale=0.6]{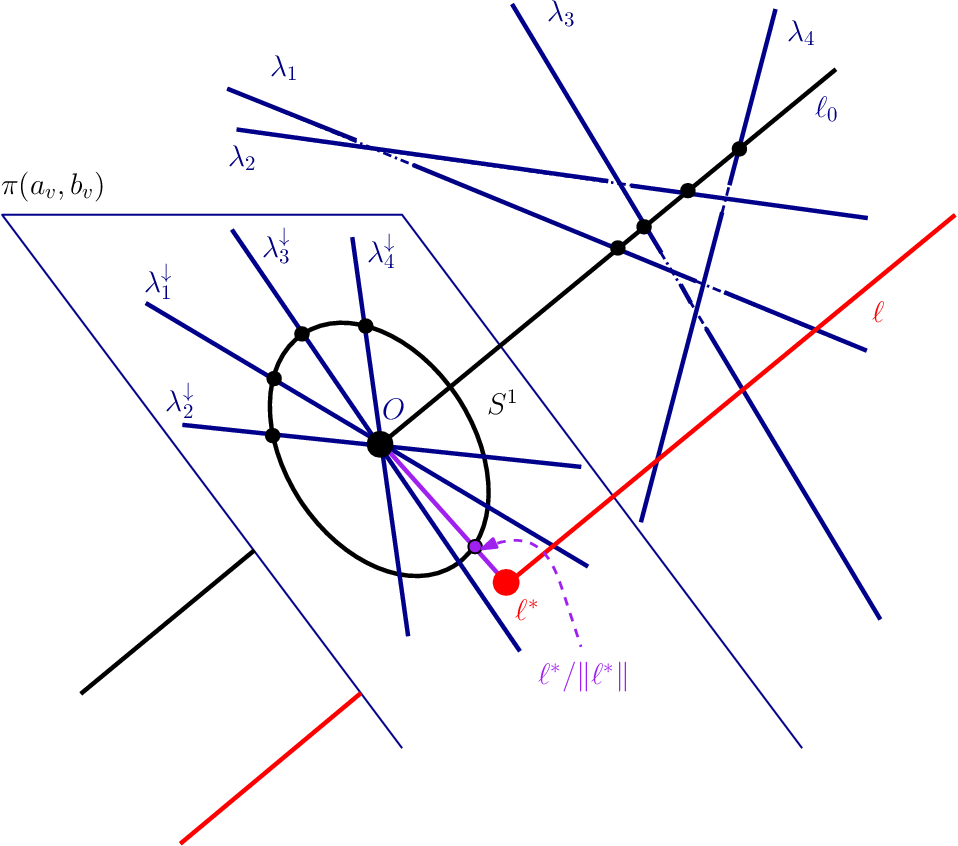}
    \end{center}
  \caption{All lines in $L_v^\circ=\{\ell_1, \ell_2,\ell_3,\ell_4\}$ are incident on the line $\ell_0^*=(a_v,b_v,0,0)$;
    Their projections on the plane $\pi(a_v,b_v)$ pass through the origin. The subdivision $\Gamma_v$ of $\sphere^1$, and a query line $\ell$ with $\ell^*=(a_v,b_v,c,d)$. The NN of $\ell$ is $\lambda_1$ or $\lambda_3$.
  }
  \label{fig:concurrent_lines}
\end{figure}

Let $L_v^\circ = \{ \lambda \in L_{p(v)} \setminus L_v \mid \cell_v \subset \ph_{\lambda,1}\}$ be the set of lines $\lambda$
such that $\lambda^\downarrow(a,b)$ passes through the origin for any $\ell^*=(a,b,c,d)\in\cell_v$
(a property equivalent to the definition of $\varphi_{1,\lambda}$). We store $L_v^\circ$ at $v$. 
If $|L_v^\circ|>1$ then,
as mentioned above, $\cell_v$ lies in a $2$-flat $\pi_v := (a_v,b_v)\times\reals^2$ for some fixed $(a_v,b_v)\in\reals^2$,
and thus $\lambda^\downarrow(a,b)$ is fixed for all $(a,b,c,d)\in\cell_v$. 
If $|L_v^\circ|>1$, let $\Gamma_v$ be the subdivision of the unit circle in $\pi_v$ (centered at the origin) induced by the set 
$L_v^\downarrow := \{ \lambda^\downarrow(a_v,b_v) \mid \lambda\in L_v^\circ\}$ of lines.  For a point 
$q\in\pi_v$, the nearest neighbor of $q$ in $L_v^\downarrow$ 
is one of the two lines that delimit the circular interval of $\Gamma_v$ containing the unit vector $q/\|q\|$. 
See \figref{concurrent_lines}.
For a query line $\ell$ with $\ell^*\in\cell_v$, we can thus compute $\NN(\ell,L_v^\circ)$, in $O(\log n)$ time, using $\Gamma_v$.

Since $\tau_v$ is not contained in 
$\ph_{\lambda,2}$ for any line $\lambda$ in $L_{p(v)}$, by the above discussion, either all lines $\ell$ 
with $\ell^*\in\cell_v$ are positively oriented with respect to all lines 
$\lambda\in L_{p(v)}\setminus (L_v\cup L_v^\circ)$ or all of them are negatively oriented.
Therefore, we associate two canonical 
subsets $L_v^+, L_v^- \subseteq L_{p(v)}\setminus (L_v \cup L_v^\circ)$ with $v$, where
$L_v^+$ (resp., $L_v^-$) is the set of lines $\lambda$ for which 
the lines $\ell$ with $\ell^*\in\tau_v$ are positively (resp., negatively) 
oriented with respect to $\lambda$. 
We construct two second-level tree data 
structures $\T^{(2+)}_v, \T^{(2-)}_v$ on the canonical subset $L_v^+$ and $L_v^-$, respectively, to compute
$\NN(\ell, L_v^+)$ and $\NN(\ell,L_v^-)$ for any $\ell\in\cell_v$.
We describe $\T^{(2+)}_v$; the 
other is constructed in a fully symmetric manner. $\T^{(2+)}_v$ is analogous to the two-level data 
structure described in \subsecref{Fast_NN1}.

Let $w$ be a node of $\T^{(2+)}_v$, and let $L_w^+\subseteq L_v^+$ be the subset of lines associated 
with $w$. If $|L_w^+| \le n_0$, for some sufficiently large constant $n_0$, then $w$ is a leaf and we store 
$L_w^+$ at $w$.  Otherwise, we fix a sufficiently large constant parameter $r>0$ and choose a random subset
$R_w \subseteq L_w^+$ of $\alpha r\log r$ lines, for a suitable constant $\alpha>0$, independent of $r$.
For every line $\lambda\in R_w$, we construct $\M_\lambda^-, \M_\lambda^+$, as defined above, with $\Xspace=\cell_v$ and $H=R_w$.
Let $\Xi_\lambda^-$ (resp., $\Xi_\lambda^+$) be the vertical decomposition of $\M_\lambda^-$ 
(resp., $\M_\lambda^+$). Let $\Xi_w$ be the set of cells in all these $2|R_w|$ vertical 
decompositions. By \lemref{union}, $|\Xi_w|= O^*(r\cdot r^3) = O^*(r^4)$ and
$\Xi_w$ can be constructed in overall $O^*(r^4)$ expected time.\footnote{%
  The cells of $\Xi_w$ may extend beyond $\cell_v$, but the subtree $\T^{(2+)}_v$ will only be queried with 
  lines $\ell$ such that $\ell^*\in\cell_v$, thereby ensuring that $\ell$ is positively oriented with respect to 
  all lines of $L_w^+$, so we do not have to clip the cells of $\Xi_w$ within $\cell_v$.}

For each line $\lambda\in R_w$ and each cell $\cell$ of $\Xi_\lambda^-$ (resp., $\Xi_\lambda^+$), we define 
$L_\cell^+ \subset L_w^+$ to be the subset of points $\theta\in L_w^+$ for which the surface 
$d = \psi_{\lambda,\theta}^+(a,b,c)$ (resp., $d= \psi_{\lambda,\theta}^-(a,b,c)$) crosses $\cell$. 
With a suitable choice of $\alpha$, the size of each conflict list is at most $n/r$, with high probability. 
For each cell $\cell\in\Xi_w$, we create a child $z = z_\tau$ of $w$, and 
recursively build at $z$ the data structure on $L_\cell^+$. 
In addition, we also store a third-level data structure at $w$:
For each line $\lambda \in R_w$, we construct the two-dimensional arrangement $\A_\lambda$ of the set of curves 
$\B_\lambda = \{\beta_{\lambda,\theta} \mid \theta \in L_w^+\setminus R_w\}$,
and preprocess it for point-location queries. 
For each edge $e\in\A_\lambda$, if $e$ is a portion of the curve 
$\beta_{\lambda,\theta}$, we set $L_e=\{\theta\}$. For a vertex 
$\chi=(\chi_a,\chi_b)$ of $\A_\lambda$,
let $L_\chi \subseteq L_w^+\setminus R_w$ be the subset of lines whose curves are incident 
on $\chi$. By definition, all lines in $L_\chi$ are parallel.
We store the lines $\theta$ of $L_\chi$ in a 
sorted list ordered by $\zeta_\theta(\chi_a,\chi_b)$
so that, for a query line $\ell$ of the form 
$\ell^*=(\chi_a,\chi_b,\ell_c,\ell_d)$, 
the line in $L_\chi$ nearest to $\ell$ can be computed in $O(\log n)$ time. 

Following the same argument as in~\subsecref{Fast_NN1},
the size of the  overall structure
is $O^*(n^4)$, and that it can be constructed in $O^*(n^4)$ expected time.

\paragraph{Query procedure.}
Let $\ell$, with $\ell^*=(a,b,c,d)$, be a query line. Recall that $\ell^\downarrow = (c,d)$ is the projection of $\ell$ onto the plane $\pi_\ell$.
	The query procedure recursively traces a root-to-leaf path $\Pi$ in $\T^{(1)}$ such that 
$\ell^* \in \cell_v$ for all nodes $v\in\Pi$, computes $\NN(\ell,L_v)$ at each node $v\in\Pi$, and returns
$\NN(\ell,L_\mathrm{root})$. Each node $v$ visited by the query procedure is processed as follows:
First, suppose $v$ is the leaf of $\Pi$. If $\cell_v\subset \ph_{\lambda,2}$ of some line $\ell\in L$, i.e., 
$\ell$ touches $\lambda$, then one such line $\lambda_v$ is stored at $v$. We return $\lambda_v$.
Otherwise, we compute $\NN(\ell,L_v)$ in a brute-force manner. 

Next, suppose $v$ is a non-leaf node, and let $w$ be the child of $v$ such that $\ell^*\in\cell_w$.  
Note that $L_v = L_w^+\cup L_w^-\cup L_w^\circ\cup L_w$. We compute the nearest neighbor of $\ell$ for each of the 
four subsets and return the one that is nearest to $\ell$.
We recursively compute $\NN(\ell, L_w)$, and compute the other three using the secondary structures stored at $w$.
If $|L_w^\circ|\le 1$, then computing $\NN(\ell,L_w^\circ)$ is straightforward, so assume that
$|L_w^\circ|>1$. We query the subdivision $\Gamma_w\subset\Sspace^1$ with the unit vector 
$\ell^\downarrow/\|\ell^\downarrow\|$, as described above, and compute $\NN(\ell,L_w^\circ)$ in $O(\log n)$ time.
Finally, since $\ell\in\cell_v$, $\ell$ is positively (resp., negatively) 
oriented with respect to all lines in $L_w^+$ (resp., $L_w^-$), so we query
$\T^{(2+)}_v$ (resp., $\T^{(2-)}_v$) with $\ell^*$, as in \subsecref{Fast_NN1}, and compute 
$\NN(\ell,L_w^+)$ (resp., $\NN(\ell,L_w^-)$) in $O^*(1)$ time.\footnote{
This step also requires searching in the arrangements $\A_\lambda$ of 
curves stored at the nodes of $\T^{(2+)}_w$ (resp., $\T^{(2-)}_w$) 
visited by the query procedure, as in \subsecref{Fast_NN1}.}
This completes the description of the query procedure.

Following the same analysis in \subsecref{Fast_NN1}, the overall query time remains $O^*(1)$.
Finally, combining this data structure with the linear-size data structure mentioned in the beginning 
of the section, and using Theorem A.2 in~\cite{AAEKS}, we obtain the following main result of this section:
\begin{theorem} \label{thm:n4lines}
Let $L$ be a set of $n$ lines in $\reals^3$ and $s\in[n,n^4]$ a storage parameter. $L$ can be preprocessed,
in $O^*(s)$ expected time, into a data structure of size $O^*(s)$, so that, for any query line 
$\ell\in\reals^3$, $\NN(\ell,L)$ can be computed in $O^*(n/s^{1/4})$ time. 
\end{theorem}

\section{Nearest-Neighbor Queries with Segments or Triangles in $\reals^3$}
\label{sec:Triangle-NN}

In this section, we extend the data structures of Sections~\ref{sec:Fast_NN} and~\ref{sec:plnn} to the general case
where the input is a set $\DDelta$ of $n$ triangles in $\reals^3$ and the query object is also a 
triangle $\tau$ in $\reals^3$, and the goal is to compute $\NN(\tau,\DDelta)$. We begin with a simpler case in which the
input is a set of segments in $\reals^3$ and the query object is also a segment in $\reals^3$.
This will be one of the main ingredients for handling the general case, and will turn out to be the most
expensive (and the most involved) part of the full algorithm for the case of triangles.

\subsection{The case of segments}
Let $\E=\{e_1, \ldots, e_n\}$ be a set of $n$ segments in $\reals^3$, and let $P$ be the set of their
endpoints. We wish to preprocess $\E$ into a data structure so that for a query 
segment $g$, $\NN(g,\E)$ can be computed quickly. Our data structure is based on the 
following simple, almost tautological, observation. 

\begin{lemma} \label{lem:seg-CP}
Let $e$ and $g$ be two segments in $\reals^3$. Let $(p,q) \in e \times g$ be the closest pair of points
on these segments. Then either (i) $p$ is an endpoint of $e$, or (ii) $q$ is an endpoint of $g$, or 
(iii) $p$ and $q$ lie in the relative interiors of $e$ and $g$, respectively, in which case they are 
also the closest pair of points in $\ell_e$ and $\ell_g$, the lines containing these two segments. 
\end{lemma}

Using \lemref{seg-CP}, we construct three separate data structures:
\begin{description}
\item[$\Psi_1$:] For a parameter $s\in[n,n^3]$, 
	we preprocess $\E$ into a data structure of size $O^*(s)$, in $O^*(s)$ expected time,
        using \thmref{NNPT}, so that, for a query point $p \in \reals^3$, $\NN(p,\E)$ can be computed 
        in $O^*(n/s^{1/3})$ expected time.
\item[$\Psi_2$:] For a parameter $s\in [n,n^4]$, 
        we preprocess $P$ into a data structure of size $O^*(s)$, in $O^*(s)$ expected time, 
        using \thmref{tr-pt}, so that, for a query segment $g$ in $\reals^3$, $\NN(g,P)$ can be 
		computed in $O^*((n/s^{1/4})^{2/3})$ expected time. 
\item[$\Psi_3$:] For a parameter $s\in[n,n^5]$, 
        we preprocess $\E$ into a data structure of size $O^*(s)$, in $O^*(s)$ expected time, 
        as described below (see \thmref{seg2segs}),
	so that, for a query segment $g$ in $\reals^3$, the structure first identifies the subset $\E_g \subseteq \E$ 
        of segments $e$ such that the closest pair of points in $e\times g$ lie in the relative 
        interiors of $e$ and $g$,
        and then it computes $\NN(\ell_g, L_g)$, where $L_g =\{ \ell_e \mid e \in \E_g\}$. 
		The expected query time is $O^*((n/s^{1/5})^{15/16}+(n/s^{1/4})^{16/15})$.
\end{description}

Since $\Psi_1$ and $\Psi_2$ can be obtained using techniques presented earlier in the paper,
we focus on $\Psi_3$.
We first derive the condition for the closest pair of points on two segments to lie in their 
respective relative interiors, and then, based on this condition, describe the data structure $\Psi_3$. 

\paragraph{Closest pair.}
Let $e$ and $g$ be two segments in $\reals^3$, and let $\ell_e$ (resp., $\ell_g$) be the line containing 
$e$ (resp., $g$). We parameterize a line $\ell$ in $\reals^3$ by the pair $(a,b)$, so that 
$\ell = \{ a+b\lambda \mid \lambda\in\reals\}$, where $a = \ell\cap \{z=0\}$ and $b$ is a unit vector 
in the direction of $\ell$ (this is a variant of the representation that we have used earlier in the 
paper, and it ignores the case where $\ell$ is horizontal, a case that can indeed be ignored
using a generic coordinate frame, or treating it separately, using a simpler structure). 
Thus each of $a$, $b$ has two degrees of freedom. 
Let $(a_e, b_e)$ (resp., $(a_g, b_g)$) be the parameters describing $\ell_e$ (resp., $\ell_g$). 
Let $\alpha_e^-, \alpha_e^+ \in \reals$ be the two parameters specifying the endpoints of $e$, 
namely $e = \{ a_e+b_e\lambda \mid \lambda\in[\alpha_e^-,\alpha_e^+] \}$. 
Similarly, let $\alpha_g^-,\alpha_g^+$ specify the endpoints of $g$, so
$g = \{ a_g + b_g\mu \mid \mu\in[\alpha_g^-,\alpha_g^+] \}$.
Thus a segment $e$ can be represented as a point $e^*=(a_e,b_e,\alpha_e^-,\alpha_e^+)$ in $\reals^6$.
We thus seek parameters $\lambda$, $\mu\in\reals$ that minimize 
\begin{align*}
	\dist^2(\ell_e,\ell_g) &= | a_e + \lambda b_e - a_g - \mu b_g |^2 \\
	&= |a_e-a_g|^2 + \lambda^2 + \mu^2 -2\lambda\mu (b_g\cdot b_e)+ 2\lambda(a_e-a_g)\cdot b_e 
- 2\mu(a_e-a_g)\cdot b_g ,
\end{align*}
and then require that $\alpha_e^-\le \lambda\le \alpha_e^+$ and $\alpha_g^-\le \mu\le \alpha_g^+$.
The desired values $\lambda$, $\mu$ then satisfy the equations 
(of the $\lambda$- and $\mu$-derivatives of the right-hand side of the above expression being $0$)
\begin{align*}
2\lambda - 2(b_g\cdot b_e)\mu + 2(a_e-a_g)\cdot b_e & = 0 \\
 2\mu - 2(b_g\cdot b_e)\lambda - 2(a_e-a_g)\cdot b_g & = 0 , \qquad\text{or}
\end{align*}
\begin{gather*} 
\lambda = \frac{(a_e-a_g)\cdot \left[ (b_g\cdot b_e)b_{g} - b_e \right]} {1 - (b_g\cdot b_e)^2} \\
\mu = \frac{(a_e-a_g)\cdot \left[ b_g - (b_g\cdot b_e)b_e \right]} {1 - (b_g\cdot b_e)^2} .
\end{gather*}

Then the condition for the closest pair to lie in the relative interiors of $e$ and $g$ is
\begin{gather*}
  \alpha_e^- < 
  \frac{(a_e-a_g)\cdot \left[ (b_g\cdot b_e)b_g - b_e \right]} {1 - (b_g\cdot b_e)^2} < 
  \alpha_e^+ \\
  \alpha_g^- < 
  \frac{(a_e-a_g)\cdot \left[ b_g - (b_g\cdot b_e)b_e \right]} {1 - (b_g\cdot b_e)^2} < 
  \alpha_g^+.
\end{gather*}

The denominator is always non-negative. 
It is $0$ when $\ell_e$ is parallel to $\ell_g$.
In such a case $\dist(\ell_g,\ell_e)$ is attained at any point on $\ell_g$ (and a corresponding point 
on $\ell_e$), and in particular at an endpoint of $g$. Therefore distances to these parallel lines 
will be considered when we handle the endpoints of $g$, using the structure $\Psi_1$.
We therefore only consider input segments for which the denominator is positive, and then
the above inequalities are the same as the conjunction of
\begin{gather} 
  \alpha_e^- \left(1 - (b_e\cdot b_g)^2\right) < 
  (a_e-a_g)\cdot \left[ (b_e\cdot b_g)b_g - b_e \right] \label{eq:in1}\\
  \alpha_e^+ \left(1 - (b_e\cdot b_g)^2\right) > 
(a_e-a_g)\cdot \left[ (b_e\cdot b_g)b_g - b_e \right] \label{eq:in2}\\
  \alpha_g^- \left(1 - (b_e\cdot b_g)^2\right) < 
  (a_{e}-a_g)\cdot \left[ b_g - (b_g\cdot b_e)b_e \right] \label{eq:in3}\\
\alpha_g^+ \left(1 - (b_e\cdot b_g)^2\right) > 
(a_{e}-a_g)\cdot \left[ b_g - (b_g\cdot b_e)b_e \right]
	. \label{eq:in4}
\end{gather}

As mentioned above, the parametric dimension of a segment in $\reals^3$ is $6$.
However, each inequality in \eqref{in1}--\eqref{in4} 
involves only at most $5$ of them (using only one of the parameters specifying the 
endpoints). Therefore, the \emph{reduced parametric dimension}
of both input and query segments in $\reals^3$ for our predicate is only $5$. 
Note that the reduced parametric dimensions of the input and the query objects are $5$ and $4$,
respectively, in (\ref{eq:in1}) and (\ref{eq:in2}), and they are
$4$ and $5$, respectively, in (\ref{eq:in3}) and (\ref{eq:in4}).

\paragraph{The data structure $\Psi_3$.}
We construct a five-level data structure. The first four levels of $\Psi_3$ are partition trees that,
for a query segment $g$, identify the subset $\E_g$ of segments $e$ for which \eqref{in1}--\eqref{in4} hold, 
namely, those for which the closest pair of points on $e$ and $g$ lie in their respective relative interiors. 
Each node of the first four levels 
is associated with a canonical subset of $\E$, and $\E_g$ is returned as the union of a few canonical subsets,
those associated with the nodes of the fourth-level partition tree that the query reaches.

The first level constructs a partition tree $\T^{(1)}$ on $\E$ using a $5$-dimensional semi-algebraic 
range-searching data structure for ranges whose parametric dimension is $4$
(see~\cite{AAEZ,MP}), such that for a query segment $g$, it filters out the 
subset $\E_g^{(1)}$ of $\E$, as the union of a few canonical subsets, associated with the nodes of
$\T^{(1)}$ that the query reaches, 
for which \eqref{in1} holds.
For each node $v$ of $\T^{(1)}$, we construct a second-level partition tree $\T_v^{(2)}$ on the associated
subset $\E_v$, again using a $5$-dimensional semi-algebraic range-searching data structure for ranges 
with parametric dimension $4$, that, for a query segment $g$, filters out the segments 
of $\E_v$ for which \eqref{in2} holds. Similarly, the next two levels 
of partition trees are constructed to filter out the input segments $e$ from the appropriate
canonical subset from the preceding level, for which \eqref{in3}, and then \eqref{in4}, holds.
We stress again that in (\ref{eq:in1}) and (\ref{eq:in2}) the input segments
have reduced parametric dimension $5$ and the query segments have dimension $4$, and in
(\ref{eq:in3}) and (\ref{eq:in4}) it is the other way around.

Finally, for each canonical subset $\E_w$ of the fourth-level partition tree, we preprocess $L_w$,
the set of lines containing the segments of $\E_w$, using \thmref{n4lines}, into a data structure
that, for a query line $\ell$, returns $\NN(\ell,L_w)$, with the
performance bounds stated in the theorem.

\paragraph{Query procedure.}
For a query segment $g$, we query the top-level partition tree 
$\T^{(1)}$ and compute the subset $\E_g^{(1)}$ of segments that satisfy \eqref{in1} with $g$, 
as the union of a family of canonical subsets associated with the set $\F^{(1)}$ of the
nodes in $\T^{(1)}$ that the query reaches.
For each node $v\in \F^{(1)}$, we query the second-level partition tree $\T_v^{(2)}$ to compute 
the subset $\E_{g,v}^{(2)}$ of segments that satisfy \eqref{in2} with $g$, 
again as the union of a small-size family of canonical subsets associated with a set $\F_v^{(2)}$ 
of nodes in $\T_v^{(2)}$. Note that $\bigcup_{v\in \F^{(1)}}\bigcup_{w\in\F_v^{(2)}} \E_w$ is the 
subset of segments of $\E$ that satisfy both \eqref{in1} and \eqref{in2}, i.e., their closest 
points to $g$ lie in their relative interiors. By continuing to query in the third- and fourth-level
partition trees, in a similar manner, we compute $\E_g$ as a family of canonical 
subsets associated with a suitable subset of the nodes of the fourth-level partition trees. 
For each such canonical subset $\E_w$, we query the fifth-level data structure, on the set
$L_w$ of the lines containing the segments of $\E_w$, compute 
$\NN(\ell_g,L_w)$, and return the nearest among the lines produced throughout the process. 

A more na\"ive analysis of the performance of the structure proceeds as follows.
Since the reduced parametric dimension of both query and input segments is (at most) $5$ at each 
of the first four levels of $\Psi_3$, and it is $4$ for the fifth level, 
Theorem A.2 in~\cite{AAEKS} implies that, for a storage parameter $s\in [n,n^5]$, 
we can construct $\Psi_3$ with size and expected preprocessing cost
$O^*(s)$, so that the overall query time is $O^*(n/s^{1/5})$.

However, a more careful analysis of the run time of the query procedure leads to a slightly 
improved bound on the query time. Specifically, for each $i \le 5$, let $t_o^i$ (resp., $t_q^i$) 
denote the reduced parametric dimension of the input (resp., query) segments of the polynomial 
inequality used in the $i$-th level of the data structure, and let $t_o = \max_i t_o^i$ and 
$t_q = \max_i t_q^i$, both of which are $5$. Again, proceeding na\"ively, the analysis 
in~\cite[Appendix]{AAEKS} shows that, for a storage parameter $s\in[n,n^{t_q}]$, the query time is 
$O^*((n/s^{1/t_q})^{\tfrac{1-1/t_o}{1-1/t_q}})$, 
which leads to the query time of $O^*(n/s^{1/5})$ as mentioned above. 
But a more refined solution of the recurrence for the query time, based on the technique
of~\cite{AAEKS} (see eq.~(A11) in Appendix A.2 of that paper) can be written in our context as 
\[ 
Q(n,s) = O^* \left ( \sum_{i=1}^5 \biggl ( \frac{n}{s^{1/t_q^i}}\biggr )^{\frac{1-1/t_o^i}{1-1/t_q^i}} \right ) ;
\]
which implies that the query time is $O^*\left( (n/s^{1/5})^{15/16}+(n/s^{1/4})^{16/15}\right)$. 
For $s\in[n,n^{31/19}]$, the second term dominates and thus the query time 
for these values of $s$ is $O^*((n/s^{1/4})^{16/15})$. 
For $s\in [n^{31/19},n^5]$, the first term dominates and thus the query time becomes 
$O^*((n/s^{1/5})^{15/16})$. (At the breakpoint value $s = n^{31/19}$, the query time is $O^*(n^{12/19})$.) 

If the input is a set of lines instead of segments, then the reduced parametric dimension of 
the input objects is $4$. Moreover, we do not have to construct the first two levels of the 
partition trees, i.e., $\Psi_3$ is only a three-level data structure.
In this case, the query time improves to $O^*((n/s^{1/5})^{15/16})$ for $s \in [n,n^5]$. 
(The improvement is for the range $s\in [n,n^{31/19}]$.)
Analogously, if the input is a set of segments but the query objects are lines, then 
the reduced parametric dimension of query objects reduces to $4$ and we do not have to 
construct the third and fourth levels of partition trees in $\Psi_3$. Again, by 
Theorem~A.2 in~\cite{AAEKS}, the query time becomes 
$O^*((n/s^{1/4})^{16/15})$ for $s \in [n,n^4]$. 
(Here the improvement is for the range $s\in [n^{31/19},n^4]$; note the reduced range of $s$.)
We thus obtain the following:

\begin{theorem} \label{thm:seg2segs}
(i) Let $\E$ be a set of $n$ segments in $\reals^3$, and let $s\in[n, n^5]$ be a storage parameter.
$\E$ can be preprocessed, in randomized expected time $O^*(s)$, into a data structure
of size $O^*(s)$, so that for a query segment $g$, $\NN(g,\E)$ can be computed 
in $O^*( (n/s^{1/5})^{15/16}+(n/s^{1/4})^{16/15})$ time.  
If $\E$ is a set of $n$ lines,  then the query time can be improved to 
$O^*((n/s^{1/5})^{15/16})$.

\smallskip\noindent
(ii) If the query objects are lines but the input is a set $\E$ of segments, then, for a storage 
parameter $s\in [n,n^4]$, $\E$ can be preprocessed, in randomized expected time $O^*(s)$, into 
a data structure of size $O^*(s)$, so that for a query line $\ell$, $\NN(\ell,\E)$ can be computed 
in $O^*((n/s^{1/4})^{16/15})$ time.
\end{theorem}

\subsection{The case of triangles}

Next, we consider the general case where the input is a set $\DDelta = \{\triangle_1, \ldots, \triangle_n\}$ 
of $n$ triangles in $\reals^3$, and the query objects are also triangles in $\reals^3$. Let $\E$ and 
$P$ be the set of edges and vertices, respectively, of the triangles in $\DDelta$. 

We use the following lemma, whose proof is straightforward.
\begin{lemma} \label{lem:tri-NN-cond}
Let $\triangle, \triangle'$ be two triangles in $\reals^3$. The closest pair of points in $\triangle\times\triangle'$ 
satisfy the following properties:
\begin{itemize}
\item[(i)] If $\triangle$ and $\triangle'$ intersect, the closest pair is $(p,p)$, where $p$ is an intersection 
point. Furthermore, in this case, either an edge of $\triangle$ intersects $\triangle'$, or 
an edge of $\triangle'$ intersects $\triangle$.

\item[(ii)] If $\triangle$ and $\triangle'$ do not intersect, let $(p,p') \in \triangle\times\triangle'$ be the 
closest pair of points. Then either (a) $p$ is a vertex of $\triangle$, 
or (b) $p'$ is a vertex of $\triangle'$, or (c) each of $p$ and $p'$ lies on the relative interior 
of respective edges of $\triangle$ and $\triangle'$. In cases (a) and (b) the other point $p$ or $p'$
can lie anywhere on its triangle.
\end{itemize}
\end{lemma}

\paragraph{Overall data structure.}
Using \lemref{tri-NN-cond}, we construct five separate data structures, one for each case of the lemma
(note that these structures are different from the previous structures that have been given the same names).

\begin{itemize}
\item[$\Psi_1$:] The goal here is to preprocess $\DDelta$ into a data 
structure that, for a query segment $e$, returns a triangle of $\DDelta$ intersecting $e$, if one exists. 
Using a $4$-dimensional semi-algebraic range searching data structure~\cite{AM93,AMS,AAEZ}, 
for a storage parameter $s\in [n,n^4]$, a query can be answered in $O^*(n/s^{1/4})$ time using 
$O^*(s)$ space and expected preprocessing cost. (Roughly, a segment $e$ intersects a triangle 
$\triangle$ iff (i) the endpoints of $e$ lie on different sides of the plane supporting $\triangle$,
and (ii) the oriented line supporting $e$ has the same (positive or negative) orientation
with respect to the three lines supporting the edges of $\triangle$ and cyclically oriented around $\triangle$.
Each of these constraints can be restated as a range searching task, where in (i) we face range searching
where both input and query objects have parametric dimensions $3$, and in (ii) these
dimensions are both $4$.)

\item[$\Psi_2$:] The goal here is to preprocess $\E$, the set of edges of the triangles of $\DDelta$, into a data 
structure that, for a query triangle $\tau$, returns a segment $e$ of $\E$ intersecting $\tau$,
if one exists. Again, using a $4$-dimensional multi-level semi-algebraic range searching 
data structure, with query objects of parametric dimension $4$~\cite{AM93,AMS,AAEZ}, 
for a storage parameter $s\in [n,n^4]$, a query can be answered in $O^*(n/s^{1/4})$ time using $O^*(s)$ space
and expected preprocessing time. (Symmetric but similar considerations to those used for $\Psi_1$ apply here too.)

\item[$\Psi_3$:] The goal here is to preprocess $\DDelta$ into a data structure that, for a query point $q$, returns 
$\NN(q,\DDelta)$. Using \thmref{NNPT}, for a storage parameter $s\in [n,n^3]$, 
such a structure, with $O^*(s)$ storage and expected preprocessing time, can be constructed so that
a query can be answered in $O^*(n/s^{1/3})$ time.

\item[$\Psi_4$:] The goal here is to preprocess $P$, the set of vertices of the triangles of $\DDelta$, 
into a data structure that, for a query triangle $\tau$, returns $\NN(\tau,P)$. Using \thmref{tr-pt},
for a storage parameter $s\in [n,n^4]$, 
such a structure, with $O^*(s)$ storage and expected preprocessing time, can be constructed so that
		a query can be answered in $O^*((n/s^{1/4})^{8/9})$ time.

\item[$\Psi_5$:] The goal here is to preprocess $\E$, the set of edges of the triangles of $\DDelta$, 
into a data structure that, for a query segment $g$, returns $\NN(g,\E)$. 
Using \thmref{seg2segs}, for a storage parameter $s\in [n,n^5]$, 
such a structure, with $O^*(s)$ storage and expected preprocessing time, can be constructed so that
a query can be answered in $O^*( (n/s^{1/5})^{15/16}+(n/s^{1/4})^{16/15})$ time.
\end{itemize}

We note that for $\Psi_1$, and $\Psi_2$, we can achieve better performance, using the mechanisms 
in~\cite{AAEKS} and~\cite{ES-22}, but these improvements do not affect the total query bound---see below.

\paragraph{Query procedure.} 
Let $\tau$ be a query triangle. First, by querying $\Psi_1$ with the edges 
of $\tau$, we determine whether an edge of $\tau$ intersects any triangle of $\DDelta$. If so, we 
return such a triangle and stop. Otherwise, we query $\Psi_2$ with $\tau$. If an edge of a 
triangle $\triangle\in \DDelta$ intersects $\tau$, we return $\triangle$ and stop. 

So assume that $\tau$ does not intersect any triangle of $\DDelta$. We first query $\Psi_3$ with the vertices of
$\tau$ and compute the triangle $\triangle_1$ that is closest to a vertex of $\tau$. Next, we query $\Psi_4$ with
$\tau$ and compute $p_\tau = \NN(\tau,P)$. Let $\triangle_2$ be the triangle of $\DDelta$ that contains $p_\tau$.
Finally, we query $\Psi_5$ with the edges of $\tau$ and compute the segment (triangle edge) $e_\tau$ 
of $\E$ that is nearest to an edge of $\tau$ (where the distance is attained at interior points of 
both segments). Let $\triangle_3\in\DDelta$ be the 
triangle that contains $e_\tau$. Among $\triangle_1, \triangle_2, \triangle_3$, we return the one that is closest to $\tau$. 
The overall query time is $O^*( (n/s^{1/5})^{15/16}+(n/s^{1/4})^{16/15})$; this is the cost of a 
query on $\Psi_5$, which dominates the cost of all the other queries, as is easily verified. 
Hence, we obtain the following:
\begin{theorem} 
	\label{thm:tri2tris}
Let $\DDelta$ be a set of $n$ triangles in $\reals^3$, and let $s\in[n, n^5]$ be a storage parameter.
$\DDelta$ can be preprocessed, in randomized expected time $O^*(s)$, into a data structure
of size $O^*(s)$, so that, for a query triangle $\tau$, $\NN(\tau,\DDelta)$ can be computed in 
$O^*( (n/s^{1/5})^{15/16}+(n/s^{1/4})^{16/15})$ time.
\end{theorem}

\section{Bichromatic Closest Pair between Points and Lines/Triangles}
\label{sec:ann}

Finally, we consider the following offline variant of NN searching in $\reals^3$. 
Let $L$ be a set of $n$ lines and $P$ a set of $m$ points in $\reals^3$.
We consider the problem of computing the closest line-point pair in $L\times P$, i.e., the pair $(\ell,p)$ 
that minimizes $\dist(\ell,p)$. We denote this minimum distance as $\dist(L,P)$.
Obviously, we can use \thmref{n4} to solve this problem: choose an appropriate value of the storage parameter,
preprocess $P$ into a data structure, using which we can answer NN queries, compute $\NN(\ell,P)$ for all 
lines $\ell \in L$, and return the closest pair among them. 
The total expected time of this algorithm can easily be shown to be $O^*(m^{4/7}n^{6/7}+m+n)$.
However, we can do better, by combining the point-enclosure technique reviewed in \secref{prelim} 
with the techniques in \secref{plnn} (rather than the machinery of \secref{Fast_NN}). 
Furthermore, our technique works, within the same asymptotic bound, also when we replace lines 
with segments or pairwise-disjoint triangles. For simplicity, we describe the algorithm for lines.

First, we note that it suffices to describe an algorithm for the decision problem, namely, given a 
parameter $r_0$, determine whether $\dist(P,L) \le r_0$. We can then plug the decision algorithm into
Chan's technique~\cite{Chan99} to compute the closest pair within the same asymptotic runtime. 
Without loss of generality assume, for simplicity, that $r_0=1$. Let $S$ be the set of balls 
of radius $r_0$, centered at 
the points of $P$. We wish to determine whether $S\times L$ contains an intersecting ball-line pair. 
We describe two algorithms. The first one is efficient for $m\le n^{1/3}$, and the second one,
which uses the first algorithm as a subroutine, is efficient for all values of $m$ and $n$.

\paragraph{The unbalanced case.}
Here we use the techniques from \secref{plnn}. First assume that $S$ and $L$ satisfy conditions 
(C1) and (C$2^+$), as stated in \secref{plnn} (handling the case of (C$2^-$) is done in a fully 
symmetric manner), namely, for every pair $(\sigma,\ell) \in S\times L$ their 
$xy$-projections $\sigma^\downarrow$ and $\ell^\downarrow$ 
intersect and $\ell$ lies below $\sigma$, in the sense defined in \secref{plnn}. As described there, 
each line $\ell\in L$ can be mapped to a point 
$\ell^* = (\ell_a,\ell_b,\ell_c,\ell_d) \in\Lspace$ and each ball $\sigma\in S$ can be mapped to 
a trivariate function $f_\sigma(a,b,c)$ such that $\ell$ intersects $\sigma$ if and only if 
$\ell_d \ge f_\sigma(\ell_a,\ell_b,\ell_c)$. 
Set $\F = \{f_\sigma\mid \sigma \in S\}$; set $\F^+$ to be the set of regions $f_\sigma^+$,
each consisting of the points on or above the graph of the corresponding function $f_\sigma$;
set $L^* = \{ \ell^* \mid \ell \in L\}$. 
$S\times L$ contains an intersecting pair iff there is a point $\ell^* \in L^*$ that lies above the 
lower envelope $\E_\F$ of $\F$. Applying \lemref{enc}(ii) to the set $L^*$ of $n$ points and the set
$\F^+$ of $m$ regions in $\reals^4$, the latter condition can be detected in 
$O^*(m^3+n)$ expected time. Using a standard batching technique (of partitioning $S$ into subsets, 
each of size at most $n^{1/3}$), the expected run time can be improved to $O^*(mn^{2/3}+n)$. 

Next, assume that $S$ and $L$ satisfy (C1). We describe a cutting-based 
divide-and-conquer procedure for detecting an intersection between a ball of $S$ and a line of $L$.
For a ball $\sigma\in S$, let $h_\sigma$ be the plane in $\reals^3$ dual to the center 
of $\sigma$ (which is a point of $P$).\footnote{%
  Recall that, in a standard duality that preserves the above / below relation, 
  the dual of a point $p=(p_1,\ldots,p_d) \in \reals^d$ is the hyperplane 
  $p^*: x_d=p_1x_1+\cdots+p_{d-1}x_{d-1}-p_d$ and vice versa; see~\cite{dBCKO-08}.}
Let $H(S) = \{ h_\sigma \mid \sigma \in S\}$. For a line $\ell\in L$, let $\pi_\ell$ be, as defined in \secref{plnn}, the 
plane containing $\ell$ and orthogonal to the vertical plane passing through $\ell$,
and let $\pi_\ell^*\in\reals^3$ be the point dual to $\pi_\ell$. By duality, 
the center of $\sigma$ lies above (resp., below) $\pi_\ell$ 
(i.e., $\ell$ lies below (resp., above) $\sigma$) if and only if 
$\pi_\ell^*$ lies in the halfspace $h_\sigma^+$ (resp., $h_\sigma^-$), the 
halfspace lying above (resp., below) $h_\sigma$.
We choose a sufficiently large constant $t>1$ and compute a $(1/t)$-cutting $\Delta$ of 
size of $O(t^3)$ of $H(S)$. For each cell 
$\triangle\in\Delta$, let $S_\triangle = \{ \sigma\in S \mid h_\sigma\cap\triangle \ne \emptyset\}$,
$S^+_\triangle = \{ \sigma\in S \mid \triangle \subset h_\sigma^+\}$, 
$S^-_\triangle = \{ \sigma\in S \mid \triangle \subset h_\sigma^-\}$, 
and $L_\triangle = \{ \ell\in L \mid \pi_\ell^* \in \triangle \}$. 
Set $m_\triangle = |S_\triangle|$ and $n_\triangle = |L_\triangle|$.
By construction, $m_\triangle \le m/t$ for each $\triangle$, and $\sum_\triangle n_\triangle  = n$.

For each cell $\triangle\in\Delta$, we solve the subproblem $(S_\triangle,L_\triangle)$ recursively.
Every pair in $S^+_\triangle\times L_\triangle$ satisfies (C1) and (C$2^+$), so we detect intersection 
between $S^+_\triangle$ and $L_\triangle$, using the procedure mentioned above, in 
$O^*(n_\triangle^{2/3} m_\triangle+n_\triangle)$ expected time. 
Similarly, every pair in $S^-_\triangle\times L_\triangle$ satisfies (C1) and (C$2^-$), so 
we detect an intersection between $S^-_\triangle$ and $L_\triangle$, using a symmetric 
procedure, in $O^*(n_\triangle^{2/3} m_\triangle +n_\triangle)$ expected time. Summing over 
all cells $\triangle$ of $\DDelta$, using the above bounds on $n_\triangle$ and $m_\triangle$
combined with H\"older's inequality, the total expected time spent in solving these secondary 
subproblems is $O^*(n^{2/3}m+n)$, as is easily verified.

The recursion terminates at subproblems where $m_\triangle = O(1)$, and the problem is then solved,
in brute force, in $O(n_\triangle)$ time.
Let $T(m,n)$ be the maximum expected time taken by the procedure with $m$ input balls and $n$ 
input lines. Then we obtain the following recurrence:
\[ 
	T(m,n) \le \sum_{\triangle\in\Delta} T(m_\triangle, n_\triangle) + O^*(n^{2/3}m+n) ,
\]
where $m_\triangle \le m/t$ for all $\triangle \in \Delta$ and $\sum_\triangle n_\triangle \le n$. 
As in the earlier contexts, using induction, H\"older's inequality, and the fact that
the number of cells at each recursive instance is $O(t^3)$, it can be verified that the solution 
of the above recurrence is $O^*(n^{2/3}m+n)$.

We now let $S$ and $L$ be arbitrary 
sets of balls of radius $r_0$ and of lines in $\reals^3$ (i.e., they may not satisfy (C1) or (C2)).
We again use a cutting-based divide-and-conquer approach (but in two dimensions this time, again
using a sufficiently large parameter for the cutting),
where at each recursive step we generate $O(1)$ subproblems that satisfy (C1), and we use 
the above algorithm to solve each of them.
Let $S^\downarrow=\{\sigma^\downarrow \mid \sigma\in S\}$ 
and $L^\downarrow = \{\ell^\downarrow \mid \ell \in L\}$ be the 
$xy$-projections of  $S$ and $L$, respectively. For a line 
$\ell^\downarrow\in L^\downarrow$, let $\left(\ell^\downarrow\right)^*\in\reals^2$ be the point dual to 
$\ell^\downarrow$, and for a unit disk 
$\sigma^\downarrow \in S^\downarrow$, let $\left(\sigma^\downarrow\right)^* \subset\reals^2$ be the dual of 
$\sigma^\downarrow$, which is a region lying between two hyperbolic arcs, such 
that a line $\lambda$ in the plane intersects $\sigma^\downarrow$ if 
and only if $\lambda^*\in\left(\sigma^\downarrow\right)^*$. Set 
$S^* = \{\left(\sigma^\downarrow\right)^* \mid \sigma \in S\}$.

The algorithm proceeds as follows: We choose a sufficiently large constant $r>1$, 
compute a $(1/r)$-cutting 
$\Xi$ of $S^*$ of size $O(r^2)$~\cite{dBS-95}. For each cell $\cell\in\Xi$,
let $S_\cell = \{\sigma \in S \mid \bd\left(\sigma^\downarrow\right)^* \cap \cell \ne \emptyset\}$, 
$\Sigma_\cell = \{ \sigma \in S \mid \cell \subseteq \left(\sigma^\downarrow\right)^*\}$, and 
$L_\cell = \{ \ell \in L \mid \left(\ell^\downarrow\right)^* \in \cell\}$; put $n_\cell = |L_\cell|$. 
By construction, $|S_\cell| \le m/r$ and $\sum_\cell n_\cell = \sum_\cell |L_\cell| = n$. For each cell 
$\cell\in\Xi$, we solve the (top-level) subproblem for 
$(S_\cell, L_\cell)$ recursively. Since $\Sigma_\cell$ and $L_\cell$ satisfy (C1), we use the aforementioned 
divide-and-conquer procedure to detect an intersection between 
$\Sigma_\cell$ and $L_\cell$ in expected time $O^*(n^{2/3}_\cell m+n_\cell)$. Note that 
no ball in $S \setminus (S_\cell \cup \Sigma_\cell)$ can 
intersect a line of $L_\cell$ (as their $xy$-projections do not intersect), so the 
above two subproblems (over all cells $\cell$ of $\Xi$) detect 
all potential intersections between $S$ and $L$.

The total expected time spent in solving the secondary subproblems is, arguing as above, $O^*(n^{2/3}m+n)$. 
Hence, following a similar analysis as above, one can conclude that the total expected time 
of the overall algorithm is also $O^*(n^{2/3}m+n)$.

\begin{lemma} 
\label{lem:cpx}
The closest pair between a set of $m$ points in $\reals^3$ and a set of $n$ lines in $\reals^3$ 
	can be computed in $O^*(mn^{2/3} + n)$ randomized expected time.
\end{lemma}

Following the approach outlined in \secref{plnn}, the above algorithm can be extended to the cases 
where the set $L$ of lines is replaced with a set of segments or a set of pairwise disjoint 
triangles in $\reals^3$. The only difference is that one has to perform a few more levels of 
divide-and-conquer, to reduce the case of segments to the case of lines, and reduce the case
of triangles to the case of segments. We omit the straightforward, albeit somewhat tedious, 
details from here and conclude the following:

\begin{lemma} 
\label{lem:bcp-tri}
The closest pair between a set of $m$ points in $\reals^3$ and a set of $n$ segments or triangles 
in $\reals^3$ can be computed in $O^*(mn^{2/3} + n)$ randomized expected time.
\end{lemma}

\paragraph{The general case.}
The bound in the lemma is optimal when $m\le n^{1/3}$. We next consider the case of larger 
values of $m$. Let $r>1$ be a parameter (not a 
constant anymore) to be chosen later. For a line $\ell \in L$, let $C_\ell$ be the cylinder of 
radius $1$
with $\ell$ as its axis, and set $\C=\{C_\ell \mid \ell \in L\}$. We compute 
a $(1/r)$-cutting $\Xi$ of the complement of the union of $\C$,
of size $O^*(r^2)$ (see \cite{VD1} and \cite{AS-00, Ezra-11}), along with the conflict list of each cell, 
using \lemref{cuttings}. For each point $p \in P$, we also compute the cell of $\Xi$ containing $p$. 
If $p$ does not lie in any cell of $\Xi$, then we conclude that $p\in\U(\C)$ (in fact, $p$ lies inside 
the union of the sample of $\C$ that we uses to construct $\Xi$), i.e.,
$\dist(p,L) < 1$, and we stop the procedure. So assume that every point $p\in P$ lies in some cell of $\Xi$. 
For each cell $\cell\in\Xi$, let $L_\cell = \{ \ell \mid \bd C_\ell \cap \cell \ne\emptyset\}$ 
be the set of lines whose cylinder boundaries cross $\cell$, and let
$P_\cell = P \cap \cell$; set $n_\cell = |L_\cell|$ and $m_\cell = |P_\cell|$. 
The total expected time spent in this step is $O^*(nr+m)$; see \lemref{cuttings}.
For each cell $\cell\in\Xi$, we test, in 
$O^*(m_\cell n_\cell^{2/3} + n_\cell) = O(m_\cell(n/r)^{2/3} + n/r)$ expected time,  
using the above algorithm, whether $\dist(P_\cell,L_\cell) \le 1$. If any of the subproblems 
returns yes, we conclude that $\dist(P,L) \le 1$; otherwise 
$\dist(P,L) > 1$. The total expected time of the algorithm is 
\[ 
\sum_{\cell\in\Xi} O^* \left (m_\cell \left(\frac{n}{r}\right)^{2/3} + 
\frac{n}{r} \right ) = O \left (m \left(\frac{n}{r}\right )^{2/3} + nr \right ) .
\]
To balance the terms, we choose $r=\ceil{m^{3/5}/n^{1/5}}$. This choice makes sense when
$1\le r\le n$, that is, when $n^{1/3}\le m\le n^2$. When $m < n^{1/3}$, we use the weaker
algorithm of Lemma~\ref{lem:cpx}, and the overall bound is $O^*(n)$. When $m > n^2$, we 
replace the cutting by the entire arrangement, and apply the same procedure, which is now trivial
within each cell of the arrangement, and the running time is $O^*(n^2+m) = O^*(m)$.
For $m$ in the middle range, the bound is $O^*(m^{3/5}n^{4/5})$. Overall,
the expected runtime is $O^*(m^{3/5}n^{4/5} + m + n)$. Hence, we conclude the following:

\begin{theorem} \label{thm:bcp}
The closest pair between a set of $m$ points in $\reals^3$ and a set of $n$ lines in $\reals^3$ 
can be computed in $O^*(m^{3/5}n^{4/5} + m + n)$ randomized expected time.
\end{theorem}

Finally, for a segment $e$ in $\reals^3$, let $\cigar_{r_0}(e)$ be the Minkowski sum of $e$ with a 
ball of radius $r_0$ in $\reals^3$, which, as earlier, we refer to as a cigar. It is known that 
the complexity of the union of $n$ cigars of the same radius is $O^*(n^2)$~\cite{AS-00,Ezra-11}. 
Replacing cylinders with cigars and following the above algorithm, we can compute the closest pair 
between a set of $m$ points and $n$ line segments in expected time $O^*(m^{3/5}n^{4/5} + m + n)$. 
Similarly, if $\T$ is a set of $n$ pairwise-disjoint triangles, then the complexity of the union 
of the Minkowski sum of these triangles with a unit ball is also $O^*(n^2)$,\footnote{%
  If the triangles are allowed to intersect, then the complexity of the union can be cubic in the worst case.} 
so the algorithm again extends to this case, and we thus obtain the following:

\begin{theorem} \label{thm:bcp-segtri}
The closest pair between a set of $m$ points in $\reals^3$ and a set of $n$ segments or pairwise-disjoint 
triangles in $\reals^3$ can be computed in $O^*(m^{3/5}n^{4/5} + m + n)$ randomized expected time.
\end{theorem}

\section{Conclusion}

In this paper we presented data structures for answering
a variety of nearest-neighbor queries involving points,
lines, segments, and  
triangles in $\reals^3$
that achieve considerably better performance bounds 
than the previously best known ones. 
Our data structures exploit several geometric 
properties of the underlying problems 
and rely on both combinatorial and algorithmic results on arrangement of surfaces, especially our 
recent results in \cite{VD1} on the complexity 
of the vertical decomposition of several substructures of 
arrangements in three and four dimensions.

We conclude by mentioning a few open problems.
No non-trivial lower bounds are known for nearest-neighbor data structures
involving lines in $\reals^3$, beyond those that follow 
from the known lower bounds on semi-algebraic range searching. For instance, it seems unlikely that the bound in
\thmref{NNLP} can be improved, but the remaining bounds are not known to be optimal.
This raises the question of whether
they
are (almost) best possible, or whether one can obtain data structures with significantly faster performance.
\begin{enumerate}[(i)]
	\item Can the nearest neighbor of a line amid a set of points in $\reals^3$ be computed in $O(\log n)$ time 
		using a data structure with $O^*(n^3)$ storage?  Alternatively, can one prove a nearly $\Omega(n^4)$ lower bound 
		on the storage complexity (for such fast queries), perhaps by adapting the approach in~\cite{AC23}?  
	\item Can the nearest neighbor of a query point in a set 
		of lines in $\reals^3$ be computed in $O^*(n^{1/2})$ time using a linear-size data structure, 
		or in $O(\log n)$ time using a 
		data structure of $O^*(n^2)$ size?
	\item Can the nearest neighbor of a segment amid $n$ segments in $\reals^3$ be computed in $O^*(n^{3/4})$ 
		time using $O^*(n)$ space? 
\end{enumerate}



\end{document}